\documentclass[runningheads,envcountsame
]{llncs}
\usepackage[T1]{fontenc}
\usepackage[boxed,vlined]{algorithm2e} 
\usepackage{tikz}
\usepackage{fullpage}
\usepackage[draft]{todonotes}
\usepackage{dsfont}
\usepackage{palatino}
\usepackage{amsfonts}
\usepackage{amsmath}
\usepackage{thm-restate}
\usepackage{color}
\usepackage[colorlinks,citecolor=red,urlcolor=blue,bookmarks=false,hypertexnames=true]{hyperref} 
\usepackage[hyperpageref]{backref}
\usepackage{amssymb}
\usepackage{breqn}

\usepackage{tikz}
\usetikzlibrary{calc}
\usetikzlibrary{angles}
\usepackage{pgfplots}
\usepackage[compatible]{algpseudocode}

\usepackage{graphicx}

\usepackage{mathtools}
\DeclarePairedDelimiter{\twonorm}{\lVert}{\rVert}
\DeclarePairedDelimiter{\bracket}{[}{]}
\DeclareMathOperator{\mR}{\mathbb{R}}
\DeclareMathOperator{\mZ}{\mathbb{Z}}
\DeclareMathOperator{\cL}{\mathcal{L}}
\DeclareMathOperator{\lat}{\mathcal{L}}
\DeclareMathOperator{\basis}{\mathbf{B}}
\DeclareMathOperator{\cR}{\mathcal{R}}

\DeclareMathOperator{\cY}{\mathcal{Y}}
\DeclareMathOperator{\cD}{\mathcal{D}}

\DeclareMathOperator{\vol}{vol}

\newcommand{\Exp}{\mathbb{E}}
\DeclareMathOperator{\real}{\mathbb{R}}
\DeclareMathOperator{\nat}{\mathbb{N}}

\newcommand{\intg}{\mathbb{Z}}

\newcommand{\Prob}[2]{\Pr_{#1}\bracket{#2}}

\newcommand{\inProd}[2]{\langle{#1},{#2}\rangle}

\newcommand{\poly}{\text{poly}}
\newcommand{\dual}[1]{#1^*}

\newcommand{\svp}{\textsf{SVP}}
\newcommand{\SVP}{\textsf{SVP}}
\newcommand{\BDD}{\textsf{BDD}}
\newcommand{\DGS}{\textsf{DGS}}
\newcommand{\cvp}{\textsf{CVP}}
\newcommand{\CVP}{\textsf{CVP}}
\newcommand{\CVPP}{\textsf{CVPP}}
\newcommand{\BDDP}{\textsf{BDDP}}
\newcommand{\ZLIP}{\textsf{$\mZ$LIP}}

\newcommand{\mat}[1]{\mathbf{#1}}
\newcommand{\lattice}[1]{\mathcal{#1}}
\newcommand{\DGauss}[2]{D_{#1\ifthenelse{\equal{#1}{}}{}{,#2}}}

\newcommand{\sddist}{\operatorname{d}_{\text{SD}}}

\newcommand\numberthis{\addtocounter{equation}{1}\tag{\theequation}}

\newcommand{\vect}[1]{\boldsymbol{#1}}
\newcommand{\eps}{\varepsilon}
\renewcommand{\epsilon}{\varepsilon}
\newcommand{\topic}[1]{\subsubsection{#1.}}
\newcommand{\norm}[1]{\left\lVert#1\right\rVert}

\DeclarePairedDelimiter{\set}{\{}{\}}
\DeclarePairedDelimiter\bra{\langle}{\rvert}
\DeclarePairedDelimiter\ket{\lvert}{\rangle}

\newcommand{\RETURN}{\STATE\textbf{return }}

\DeclareUnicodeCharacter{00B4}{'}
\allowdisplaybreaks
\pgfplotsset{compat=1.16} 
\begin{document}

\author{Divesh Aggarwal\inst{1}\thanks{This work was partially supported by the Singapore Ministry of Education under grant MOE2019-T2-1-145, National Research Foundation under grant R-710-000-012-135, NRF award NRF-NRFI09-0005"}  \and
Yanlin Chen\inst{2} \and
Rajendra Kumar\inst{3}\thanks{This research has been supported in part by the National Research Foundation Singapore under its AI Singapore Programme [Award Number: AISG-RP-2018-005]
} \and
Yixin Shen\inst{4}\thanks{This research was partially supported by EPSRC grant EP/W02778X/2 and by the French Agence Nationale de la Recherche through the France 2030 program under grant agreement No. ANR-22-PETQ-0008 PQ-TLS.}}
\authorrunning{D. Aggarwal et al.}

\institute{Centre for Quantum Technologies and National
            University of Singapore. email: \texttt{dcsdiva@nus.edu.sg}. \and 
             QuSoft and CWI, the Netherlands. email: \texttt{yanlin@cwi.nl} \and
Indian Institute of Technology Delhi, New Delhi, India. email: \texttt{rjndr2503@gmail.com} 
            \and 
            Univ Rennes, Inria, CNRS, IRISA, Rennes, France. email: \texttt{yixin.shen@inria.fr}
            }

\title{Improved Classical and Quantum Algorithms for the Shortest Vector Problem via Bounded Distance Decoding\thanks{A preliminary version of this paper titled "Improved (Provable) algorithms for the Shortest Vector Problem via Bounded Distance Decoding" appeared in the proceedings of the 38th International Symposium on Theoretical Aspects of Computer Science (STACS 2021)~\cite{ACKS21}.}}

\maketitle
\begin{abstract}The most important computational problem on lattices is the Shortest Vector Problem ($\SVP$). In this paper, we present new algorithms that improve the state-of-the-art for provable classical/quantum algorithms for $\SVP$. We present the following results.
\begin{enumerate}
\item A new algorithm for $\svp$ that provides a smooth tradeoff between time complexity and memory requirement. For any positive integer $4\leq q\leq \sqrt{n}$, our algorithm takes $q^{13n+o(n)}$ time and requires $poly(n)\cdot q^{16n/q^2}$ memory. This tradeoff which ranges from enumeration ($q=\sqrt{n}$) to sieving ($q$ constant), is a consequence of a new time-memory tradeoff for Discrete Gaussian sampling above the smoothing parameter. 

\item A quantum algorithm for $\svp$ that runs in time $2^{0.950n+o(n)}$ and requires $2^{0.5n+o(n)}$ classical memory and $\poly(n)$ qubits. In Quantum Random Access Memory (QRAM) model this algorithm takes only $2^{0.835n+o(n)}$ time and requires a QRAM of size $2^{0.293n+o(n)}$, $\poly(n)$ qubits and $2^{0.5n}$ classical space. This improves over the previously  fastest classical (which is also the fastest quantum) algorithm due to~\cite{ADRSD15} that has a time and space complexity $2^{n+o(n)}$.

\item A classical algorithm for $\svp$ that runs in time $2^{1.669n+o(n)}$ time and $2^{0.5n+o(n)}$ space. This improves over an algorithm of~\cite{CCL18} that has the same space complexity.
\end{enumerate}

The time complexity of our classical and quantum algorithms are obtained using a known upper bound on a quantity related to the lattice kissing number which is $2^{0.402n}$. We conjecture that for most lattices this quantity is a $2^{o(n)}$. Assuming that this is the case, our classical algorithm runs in time $2^{1.292n+o(n)}$, our quantum algorithm runs in time $2^{0.750n+o(n)}$ and our quantum algorithm in QRAM model runs in time $2^{0.667n+o(n)}$.

As a direct application of our result, using the reduction in \cite{Ducas23}, we obtain a provable quantum algorithm for the Lattice Isomorphism Problem in the case of the trivial lattice $\mZ^n$ ($\ZLIP$) that runs in time $2^{0.417n+o(n)}$. Our algorithm requires a QRAM of size $2^{0.147n+o(n)}$, $\poly(n)$ qubits and $2^{0.25n}$ classical space. 

\end{abstract}

\section{Introduction}

A lattice $\cL = \cL(\vect{b}_1, \ldots, \vect{b}_n) := \{\sum_{i=1}^n z_i \vect{b}_i : z_i \in \mZ\}$ is the set of all integer combinations of linearly independent vectors $\vect{b}_1,\dots,\vect{b}_n \in \mR^n$.
 We call $n$ the rank of the lattice and $(\vect{b}_1, \ldots, \vect{b}_n)$ a basis of the lattice.
 
The most important computational problem on lattices is the Shortest Vector Problem ($\SVP$). 
Given a basis for a lattice $\cL \subseteq \mR^n$,
$\SVP$ asks us to compute a non-zero vector in $\cL$ with the smallest Euclidean norm. 
Starting from the '80s, the use of approximate and exact solvers for $\SVP$ (and other lattice problems) gained prominence for their applications in algorithmic number theory~\cite{LLL82}, convex optimization~\cite{Lenstra83,Kannan87,FrankT87}, coding theory~\cite{Buda89}, and cryptanalysis tool~\cite{Shamir84,Brickell84,LagariasO85}. The security of many cryptographic primitives is based on the worst-case hardness  
of (a decision variant of) approximate $\mathsf{SVP}$ to within polynomial 
factors~\cite{Ajtai96,MR04,Regev09,Regev06,MR08,Gentry09,BV14} in the sense that any cryptanalytic attack on these cryptosystems that runs in time polynomial in the security parameter implies a polynomial time algorithm to solve approximate $\mathsf{SVP}$ to within polynomial factors. 
Such cryptosystems have attracted a lot of research interest due to their conjectured resistance to quantum attacks.

The $\SVP$ is a well-studied computational problem in both its exact and approximate (decision) versions. By a randomized reduction, it is known to be NP-hard to approximate within any constant factor, and hard to approximate within a factor $n^{c/\log \log n}$ for some $c > 0$ under reasonable complexity-theoretic assumptions~\cite{van1981another,Micciancio98,Khot05,HR07}.  For an approximation factor $2^{\mathcal{O}(n)}$, one can solve $\SVP$ in time polynomial in $n$ using the celebrated LLL lattice basis reduction algorithm~\cite{LLL82}.
In general, the fastest known algorithm(s) for approximating $\SVP$ within factors polynomial in $n$ rely on (a variant of) the BKZ lattice basis reduction algorithm~\cite{Schnorr1987AHO,SE94,AKS01,GN08,HPS11,ALNS19}, which can be seen as a generalization of the LLL algorithm and  gives an
$r^{n/r}$ approximation in $2^{\mathcal{O}(r)}\poly(n)$ time. All these algorithms internally use an algorithm for solving (near) exact $\SVP$ in lower-dimensional lattices.
Therefore, finding faster algorithms to solve exact $\SVP$ is critical to choosing security parameters of cryptographic primitives.

As one would expect from the hardness results above, all known algorithms  for solving exact $\SVP$, 
including the ones we present here, require at least exponential time. In fact, the fastest known algorithms also require exponential space. There has been some recent evidence~\cite{AS18b} showing that one cannot hope to get a $2^{o(n)}$ time algorithm for $\SVP$ if one believes in complexity theoretic conjectures such as the (Gap) Exponential Time Hypothesis. Most of the known algorithms for $\SVP$ can be broadly classified into two classes: (i) the algorithms that require memory polynomial in $n$ but run in time $n^{\mathcal{O}(n)}$ and (ii) the algorithms that require memory $2^{\mathcal{O}(n)}$ and run in time $2^{\mathcal{O}(n)}$.

The first class, initiated by Kannan~\cite{Kannan87,Helfrich85,HS07,GN10,MW15,albrecht2020faster}, combines basis reduction with exhaustive enumeration inside Euclidean balls. While enumerating vectors requires $2^{\mathcal{O}(n\log n)}$ time, it is much more space-efficient than other kinds of algorithms for exact $\SVP$.
 
Another class of algorithms, and currently the fastest, is based on sieving. First developed by Ajtai, Kumar, and Sivakumar~\cite{AKS01}, they generate many lattice vectors and then divide-and-sieve to create shorter and shorter vectors iteratively. A sequence of improvements~\cite{Regev04,NV08,MV10,PS09,ADRSD15,AS18}, has led to a $2^{n+o(n)}$ time and space algorithm by sieving the lattice vectors and carefully controlling the distribution of the output, thereby outputting a set of lattice vectors that contains the shortest vector with overwhelming probability.

An alternative approach using the Voronoi cell of the lattice was proposed by Micciancio and Voulgaris~\cite{MV13} and gives a deterministic $2^{2n+o(n)}$-time and $2^{n+o(n)}$-space algorithm for $\SVP$ (and many other lattice problems). 

There are variants~\cite{NV08,MV10,LMP15,BDGL16} of the above mentioned sieving algorithms that, under some heuristic assumptions, have an asymptotically smaller (but still $2^{\Theta(n)}$) time and space complexity than their provable counterparts. 

\paragraph{Algorithms giving a time/space tradeoff.} Even though sieving algorithms are asymptotically the fastest known algorithms for $\SVP$, the memory requirement, in high dimension, has historically been a limiting factor to run these
algorithms.  
Some recent works \cite{Ducas18,ADHKPS19} have shown how to use new tricks to make it possible to use sieving on high-dimensional lattices in practice and benefit from their efficient running time \cite{Records}. 

Nevertheless, it would be ideal and has been a long standing open question to obtain an algorithm that achieves the ``best of both worlds", i.e. an algorithm that runs in time $2^{\mathcal{O}(n)}$ and requires memory polynomial in $n$. 
In the absence of such an algorithm, it is desirable to have a smooth tradeoff between the time and memory requirement  that interpolates between the current best sieving algorithms and the current best enumeration algorithms. 

To this end, Bai, Laarhoven, and Stehl{\'{e}}~\cite{BLS16}  proposed the tuple sieving algorithm, providing such a tradeoff based on heuristic assumptions similar in nature to prior sieving algorithms. This algorithm was later proven to have time and space complexity $k^{O(n)}$ and $k^{O(n/k) }$, under the same heuristic assumptions~\cite{HK17}. One can vary the parameter $k$ to obtain a smooth time/space tradeoff. Nevertheless, it is still desirable to obtain a provable variant of this algorithm that does not rely on any heuristics.

Kirchner and Fouque~\cite{Kirchner16} attempted to do this. They claim an algorithm for solving $\SVP$ in time $k^{\Theta(n)}$ and in space $k^{\Theta(n/k)}$ for any positive integer $k>1$. Unfortunately, their analysis falls short of supporting their claimed result, and the correctness of the algorithm is not clear. We refer the reader to Section~\ref{section_comparison} for more details.

In addition to the above, Chen, Chung, and Lai~\cite{CCL18} propose a variant of the algorithm based on Discrete Gaussian sampling in~\cite{ADRSD15}. Their algorithm runs in time $2^{2.05n+o(n)}$ and the memory requirement is $2^{0.5n+o(n)}$. The quantum variant of their algorithm runs in time $2^{1.2553n+o(n)}$ and has the same space complexity. Their algorithm has the best space complexity among known provably correct algorithms that run in time $2^{O(n)}$.

A number of works have also investigated the potential quantum
speedups for lattice algorithms, and $\SVP$ in particular. A similar
landscape to the classical one exists, although the quantum memory model has its importance. While quantum enumeration algorithms only require qubits~\cite{ANS18}, sieving algorithms require more powerful QRAMs~\cite{LMP15,Kir19}.

\subsection{Our results}

We first present a new algorithm for $\SVP$ that provides a smooth tradeoff between the time complexity and memory requirement of $\SVP$ without any heuristic assumptions. This algorithm is obtained by giving a new algorithm for sampling lattice vectors from the Discrete Gaussian distribution that runs in time $k^{\mathcal{O}(n)}$ and requires $k^{\mathcal{O}(n/k)}$ space.

\begin{theorem}[Time-space tradeoff for smooth discrete Gaussian, informal]\label{theorem_main_dgs}
There is an algorithm that takes as input a lattice $\cL \subset \mathbb{R}^n$, a positive integer $q$, and a parameter $s$ above the smoothing parameter of $\cL$, and outputs $q^{16n/q^2}$ samples from $D_{\cL,s}$ using $q^{13n+o(n)}$ time and $poly(q)\cdot q^{16n/q^2}$ space. 
\end{theorem}

Using the standard reduction from  Bounded Distance Decoding ($\BDD$) with preprocessing (where an  algorithm solving the problem is allowed unlimited preprocessing time on the lattice before the algorithm receives the target vector) to Discrete Gaussian Sampling ($\DGS$) from~\cite{DRS14} and a reduction from $\SVP$ to $\BDD$ given in~\cite{CCL18}, we obtain the following. 
\begin{theorem}[Time-space tradeoff for $\SVP$]\label{theorem_main_svp}
Let $n \in \mathbb{N}, q \in [4,\sqrt{n}]$ be a positive integer. Let $\cL$ be a lattice of rank $n$. There is a randomized algorithm that solves $\svp$ in time $ q^{13n+o(n)}$ and in space $poly(n)\cdot q^{\frac{16n}{q^2}}$.
\end{theorem}

If we take $k=q^2$, then the time complexity of the previous $\svp$ algorithm becomes $k^{6.5n+o(n)}$ and the space complexity $\poly(n)\cdot k^{(8n/k)}$. Our tradeoff is thus the same (up to a constant in the exponents) as what was claimed by Kirchner and Fouque~\cite{Kirchner16} and proven in~\cite{HK17} {\em under heuristic assumptions}.

\medskip

Our second result is a quantum algorithm for $\SVP$ that improves over the current {\em fastest quantum algorithm} for $\SVP$~\cite{ADRSD15} (Notice that the algorithm in~\cite{ADRSD15} is still the fastest classical algorithm for $\SVP$). 

\begin{theorem}[Quantum Algorithm for $\SVP$]\label{theorem_main_algorithm_quantum}
There is a quantum algorithm that solves $\SVP$ in $2^{0.950n+o(n)}$ time and classical $2^{0.5n+o(n)}$ space with an additional $\poly(N)$ qubits. In the Quantum Random Access Memory (QRAM) model, there is an algorithm that solves $\SVP$ in $2^{0.8345n+o(n)}$ time and require a QRAM of size $2^{0.293n+o(n)}$, $\poly(n)$ qubits and $2^{0.5n+o(n)}$ classical space. 
\end{theorem}

Our third result is a classical algorithm for $\SVP$ that improves over the  algorithm from~\cite{CCL18} and results in the fastest classical algorithm that has a space complexity $2^{0.5n+o(n)}$. 
\begin{theorem}[Algorithm for $\SVP$ with $2^{0.5n+o(n)}$ space]\label{theorem_main_algorithm_classical}
There is a classical algorithm that solves $\SVP$ in $2^{1.669n+o(n)}$ time and $2^{0.5n+o(n)}$ space.
\end{theorem}

The time complexity of our second and third results are obtained using a known upper bound on a quantity $\beta(\cL)^n$ related to the kissing number, which depends on
the lattice and is always upper bounded by $2^{0.402n}$.
We analysed the dependency of the running time of our algorithm
in this quantity $\beta(\cL)$ and plotted (Figure~\ref{fig:complexity_svp_all}) the graph
of the complexity exponent as a function of $\beta(\cL)$.
In practice, for most lattices, $\beta(\cL)^n$ is often a $2^{o(n)}$ \footnote{  
Please refer to our discussion on kissing number in Section~\ref{section_preliminaries}.}.
In this case, the running time of our algorithm is significantly better than when
using the generic upper bound on $\beta(\cL)$.

\begin{theorem}
    For any family $(\cL_n)_n$ of full-rank lattices such that $\cL_n\subseteq\mathbb{R}^n$
    and $\beta(\cL_n)^n=2^{o(n)}$, there
    are algorithms to solve the $\svp$ on $\cL_n$:
    \begin{itemize}
        \item in classical time $2^{1.292n+o(n)}$ and space $2^{0.5n}$,
        \item in quantum time $2^{0.750n+o(n)}$, classical space $2^{0.5n}$
            and $\poly(n)$ qubits,
        \item in quantum time $2^{0.667n+o(n)}$, classical space $2^{0.5n}$,
            $\poly(n)$ qubits and using a QRAM of size $2^{0.167n+o(n)}$.
    \end{itemize}
\end{theorem}

Below in Table \ref{table:classical_algorithms} and \ref{table:quantum_algorithms}, we summarize known provable classical and quantum algorithms respectively. Note that all the classical algorithms are also quantum algorithms but they don't use any quantum power.

\begin{table}
\centering
\setlength{\tabcolsep}{20pt}
\renewcommand{\arraystretch}{1.5}
\begin{tabular}{|c| c| c| c|} 
 \hline
 \textbf{Time Complexity}  & \textbf{Space Complexity} & \textbf{Reference} \\ [0.5ex] 
 \hline
 $n^{\frac{n}{2e}+o(n)}$ & $\poly(n)$ & \cite{Kannan87,HS07}  \\
 \hline
 $2^{\mathcal{O}(n)}$ & $2^{\mathcal{O}(n)}$ & \cite{AKS01}  \\
 \hline
 $2^{2.465n+o(n)}$ & $2^{1.233n+o(n)}$ & \cite{PS09}\\ 
 \hline
 $2^{2n+o(n)}$ & $2^{n+o(n)}$ & \cite{MV10}  \\
 \hline
 $2^{n+o(n)}$ & $2^{n+o(n)}$ & \cite{ADRSD15}  \\
 \hline
 $2^{2.05n+o(n)}$ & $2^{0.5n+o(n)}$ & \cite{CCL18}\\ 
 \hline
 $2^{1.669n+o(n)}$ & $2^{0.5n+o(n)}$ & This paper\\ 
 [1ex] 
 \hline
\end{tabular}
\caption{Classical algorithms for Shortest vector problem.}
\label{table:classical_algorithms}
\end{table}

\begin{table}
\centering
\setlength{\tabcolsep}{20pt}
\renewcommand{\arraystretch}{1.5}
\begin{tabular}{|c| c| c| c|} 
 \hline
 \textbf{Time Complexity}  & \textbf{Space Complexity} & \textbf{Reference} \\ [0.5ex] 
 \hline
 $2^{1.799n+o(n)}$ & $2^{1.286n+o(n)}$  QRAM & \cite{LMP15}  \\
 \hline
 $2^{1.2553n+o(n)}$ & $2^{0.5n+o(n)}$ & \cite{CCL18}  \\
  \hline
 $2^{0.950n+o(n)}$ & $2^{0.5n+o(n)}$ & This paper\\
 \hline
 $2^{0.835n+o(n)}$ & $2^{0.5n+o(n)}$ Classical space and $2^{0.293n+o(n)}$ QRAM & This paper\\ [1ex] 
 \hline
\end{tabular}
\caption{Quantum algorithms for Shortest vector problem.}
\label{table:quantum_algorithms}
\end{table}

\topic{Roadmap}
We give a high-level overview of our proofs in Section~\ref{sec:overview}. We compare our results with the previous known algorithms that claim/conjecture a time-space tradeoff for $\SVP$ in Section~\ref{section_comparison}.
Section~\ref{section_preliminaries}  contain some preliminaries on lattices. The proofs of the time-space tradeoff for Discrete Gaussian sampling above the smoothing parameter and the time-space tradeoff for $\SVP$ are given in Section \ref{tradeoff}. In Section~\ref{section:BDDP-in-QRAM}, we present a Quantum algorithm for $\BDD$ with preprocessing, that we will use later in our quantum Algorithm for $\SVP$. In section~\ref{section-:algorithm-BDD}, we present improved algorithms for $\BDD$. Our quantum and classical algorithms for solving $\SVP$ with space complexity $2^{0.5n+o(n)}$ are presented in Section~\ref{sec:spherical_caps} respectively. Section~\ref{section_relation_kissing} shows how the time complexity of our algorithms varies on a quantity related to the kissing number. Section~\ref{section-:ZLIP} gives a direct application of our results to the Lattice Isomorphism Problem in the case of the trivial lattice $\mZ^n$ ($\ZLIP$).

\begin{remark} The paper is an extended version of \cite{ACKS21}, published at the 38th International Symposium on Theoretical Aspects of Computer Science (STACS 2021). Compared to \cite{ACKS21}, this paper gives a quantum speedup for Bounded distance decoding in processing using QRAM. Using this quantum speedup, we get a much faster quantum algorithm for $\SVP$. We also realised that one of the conditions in Theorem~\ref{theorem_enlargeBDD} is unnecessary \footnote{See Remark~\ref{rm:unnecessary_condition}}, this leads to a better time complexity in all our algorithms compared to the previous version. Furthermore, we study the dependency of our algorithms on a quantity related to the kissing number in Section~\ref{section_relation_kissing}.
\end{remark}
\subsection{Proof overview}\label{sec:overview}

We now include a high-level description of our proofs. Before describing our proof ideas, we emphasize that it was shown in~\cite{DRS14,ADRSD15} that given an algorithm for $\DGS$ a constant factor $c$ above the smoothing parameter, we can solve the problem of $\BDD$ where the target vector is within distance $\alpha \lambda_1(\cL)$ of the lattice, where the constant $\alpha < 0.5$ depends on the constant $c$. Additionally, using~\cite{CCL18}, one can enumerate all lattice points within distance $p\delta$ to a target $\vect{t}$ by querying $p^n$ times a $\BDD$ oracle with decoding distance $\delta$ (or $p^{n/2}$ times if we are given a quantum $\BDD$ oracle). Thus, by choosing $p = \lceil \lambda_1(\cL)/\delta \rceil$ and $\vect{t} = \vect{0}$, an algorithm for $\BDD$ immediately gives us an algorithm for $\SVP$. Therefore, it suffices to give an algorithm for $\DGS$ above the smoothing parameter.

\topic{Time-space tradeoff for $\DGS$ above smoothing} 

Recall that  efficient algorithms are known for sampling from a discrete Gaussian with a large enough parameter (width) \cite{Klein00,GPV08,BLPRS13}.
In~\cite{ADRSD15}, the authors begin by sampling $N = 2^{n+o(n)}$ vectors from the Discrete Gaussian distribution with (large) parameter $s$  and then look for pairs of vectors whose sum is in $2\cL$, or equivalently pairs of vectors that lie in the same coset $\vect{c} \in \cL /2\cL$. Since there are $2^n$ cosets, if we take $\Omega(2^n)$ samples from $D_{\cL, s}$, almost all of the resulting vectors (except at most $2^n$ vectors) will be paired and are statistically close to independent samples from the distribution $D_{\cL, s/\sqrt{2}}$, provided that the parameter $s$ is sufficiently above the smoothing parameter.

To reduce the space complexity, we modify the idea of the algorithm by generating random samples and checking if the summation of $d$ of those samples is in $q\cL$ for some integer $q$. Intuitively, if we start with  two lists of vectors ($L_1$ and $L_2$) of size $q^{\mathcal{O}(n/d)}$ from $D_{\cL, s}$, where $s$ is sufficiently above the smoothing parameter, each of these vectors is contained in any coset $q\cL + \vect{c}$ for any $\vect{c} \in \cL/q\cL$ with probability roughly $1/q^n$. We therefore expect that the coset of a uniformly random d-combination of vectors from $L_2$ is uniformly distributed in $\cL/q\cL$. The proof of this statement follows from the Leftover Hash Lemma~\cite{impagliazzo1989pseudo}. We therefore expect that for any vector $\vect{v}\in L_1$, with high probability, there is a set of $d$ vectors $\vect{x}_1, \ldots, \vect{x}_d$ in $L_2$ that sum to a vector in $q \cL+\vect{v}$, and hence $\frac{1}{q}\left(\sum_{i=1}^d \vect{x}_i -\vect{v}\right) \in \cL$. A lemma by Micciancio and Peikert (\cite{MP13}) shows that this vector is statistically close to a sample from the distribution, $D_{\cL, s\sqrt{d+1}/q}$. We can find such a combination by trying all subsets of $d$ vectors. 

We would like to repeat this and find $q^{\mathcal{O}(n/d)}$ (nearly) independent vectors in $q\cL$. It is not immediately clear how to continue since, in order to guarantee independence, one would not want to reuse the already used vectors $\vect{x}_1, \ldots, \vect{x}_d$ and conditioned on the choice of these vectors, the distribution of the cosets containing the remaining vectors is disturbed and is no longer nearly uniform. By using a simple combinatorial argument, we show that even after removing any $1/\poly(d)$ fraction of vectors from the list $L_2$, the $d$-combination of vectors in $L_2$ has at least $cq^n$ different cosets. This is sufficient to output $q^{\mathcal{O}(n/d)}$ independent vectors in $q\cL$ with overwhelming probability.

\topic{Quantum algorithm for $\BDDP$ in QRAM} In \cite{DRS14}, Dadush, Regev, and Stephens-Davidowitz gave an algorithm for $\BDD$ with preprocessing (or $\BDDP$), which requires advice containing vectors sampled from discrete Gaussian sampling over the dual lattice. The idea is to use the \emph{periodic Gaussian function} $f$ (defined in Section~\ref{section:BDDP-in-QRAM})  to go near to the closest lattice vector. The function $f$ is periodic over the lattice, and its value depends only on the distance between the input vector and the lattice. Do the gradient ascent by iteratively updating the target vector using values of $\nabla f$ and $f$ such that the distance of the target vector from the closest lattice vector decreases. In this work, we show that we can reduce the time complexity of this algorithm by using the quantum amplitude estimation technique (given in \cite{BHMT02}) with the assumption that the advice string is stored in QRAM. More specifically, we show that just by using $\mathcal{O}(\sqrt{N})$ arithmetic operations in QRAM, we can solve $\BDDP$ where $N$ is the size of the advice string required in \cite{DRS14}.

\topic{An improved algorithm for DGS at the smoothing parameter}
The $\BDD$ to $\DGS$ reduction from \cite{DRS14} requires samples from
$D_{\cL,\eta_\eps(\cL)}$ for $\eps = 2^{-cn}$ for some constant $c$.
In~\cite{ADRSD15}, the authors gave an algorithm that runs in time $2^{n/2+o(n)}$ and outputs $2^{n/2+o(n)}$ samples from $ D_{\cL,s}$ for any $s \geq \sqrt{2}\eta_{0.5}(\cL)$, i.e. a factor $\sqrt{2}$ above the
smoothing parameter. In order to obtain samples at the smoothing
parameter, we construct a dense lattice $\cL'$ of smaller smoothing
parameter than $\cL$. We then sample $2^{0.5n+o(n)}$ vectors from
$D_{\cL',s}$ and reject those that are not in $\cL$. This allows us
to obtain an improved algorithm for BDD with preprocessing.

\topic{Covering surface of a ball by spherical caps}
This result improves the quantum algorithm from~\cite{CCL18}.
As we mentioned above, one can enumerate all lattice points within a $p\delta$ distance to a target $\vect{t}$ by querying $p^n$ times a $\BDD$ oracle with decoding distance $\delta$. Our algorithm for $\BDD$ is obtained by preparing samples from the discrete Gaussian distribution. However, note that the decoding distance of $\BDD$ oracle built by discrete Gaussian samples as shown in~\cite{DRS14} is successful if the target vector is within a radius  $\alpha \lambda_1(\cL)$ for $\alpha < 1/2$ (there is a tradeoff between $\alpha$ and the number of $\DGS$ samples needed), and therefore, if we choose $\vect{t}$ to be $\vect{0}$, as we do in the other algorithms mentioned above, then $p$ has to be at least $3$ to ensure that the shortest vector is one of the vectors output by the enumeration algorithm mentioned above. 
We observe here that if we choose a target $\vect{t}$ to be a random vector on a sphere of a well-chosen radius centered  at the origin, then the shortest vector will be within a radius $2\delta$ from the target $\vect{t}$ with some probability $P$, and thus we can find the shortest vector by making $2^n/P$ calls to the $\BDD$ oracle. An appropriate\footnote{The optimal choice of $\alpha$ is obtained
by numerically optimisation, see Section~\ref{section_relation_kissing}.} choice of the target $\vect{t}$ and the factor $\alpha$ gives an algorithm that runs in time $2^n \cdot 2^{0.669n + o(n)}$.
We then explain how to obtain a quantum quantum speed up of this algorithm
(with its corresponding optimized value of $\alpha$ which is different from the classical one) that runs in time $2^{n/2} \cdot 2^{0.45n+o(n)}$
with a polynomial number of qubits. Finally, if we can store the DGS
samples in a QRAM, we can obtain a further speed up and an algorithm
that runs in time $2^{n/2} \cdot 2^{0.3345n+o(n)}$ and uses a QRAM
of size $2^{0.293n+o(n)}$.

\topic{Dependency on a quantity related to the kissing number}
The running time of the above algorithms crucially depends on a quantity related to the kissing number of the input lattice.
This quantity plays a role in the $\BDD$ to $\DGS$ reduction when relating the decoding radius $\alpha$ to the $\varepsilon$ when sampling at the smoothing parameter $\eta_\varepsilon$. Our algorithms takes significantly less time for the smaller values of this quantity.
However, the only known upper bound on this quantity seems to be
very pessimistic for most lattices. Since we have used this upper bound to derive
the complexity of our algorithm (except in Section~\ref{section_relation_kissing}),
this means that the actual running time of this algorithm might be much better
for most lattices. For a more elaborate discussion on this, see Section~\ref{section_relation_kissing}. 
As far as we know, this is the first time that the time/space complexity dependency on a quantity related to the input lattice has been investigated. For example, the time complexity of the $\SVP$ algorithm in \cite{MV10,PS09} depends on the geometric kissing number which is a universal quantity depending only on the dimension $n$, for which we know a lower bound of $2^{0.2075n}$.

\subsection{Comparison with previous algorithms giving a time/space tradeoff}\label{section_comparison}

Kirchner and Fouque~\cite{Kirchner16} begin their algorithm by sampling an exponential number of vectors from the discrete Gaussian distribution $D_{\cL,s}$ and then using a pigeon-hole principle, show that there are two distinct sums of $d$ vectors (for an appropriate $d$) that are equal mod $q \cL$, for some large enough integer $q$. This results in a $\{-1,0,1\}$ combination of input lattice vectors (of  Hamming weight at most $2d$) in $q \cL$; a similar idea was used in~\cite{BLS16} to construct their tuple sieving algorithm. In both algorithms, it is difficult to control (i) the distribution of the resulting vectors, (ii) the dependence between resulting vectors. 

Bai et al~\cite{BLS16} get around the above issues by making a heuristic assumption that the resulting vectors behave like
independent samples from a ``nice enough'' distribution.
\cite{HK17} proved that this heuristic indeed leads to the
time-memory tradeoff conjectured in \cite{BLS16}, but don't prove correctness.

Kirchner and Fouque, on the other hand, use the pigeon-hole principle to argue that there exist coefficients $\alpha_1, \ldots, \alpha_{2d} \in \{-1,0,1\}$ and $2d$ lattice vectors in the set of input vectors $\vect{v}_1, \ldots, \vect{v}_{2d}$ such that $\frac{\sum_{i=1}^{2d} \alpha_i \vect{v}_i}{q} \in \cL$. It is then stated that $\frac{\sum_{i=1}^{2d} \alpha_i \vect{v}_i}{q}$ has a nice enough Discrete Gaussian distribution. We observe that while the resulting distribution obtained will indeed be close to a discrete Gaussian distribution, we have no control over the parameter $s$ of this distribution and it can be anywhere between $1/q$ and $\sqrt{2d}/q$ depending on the number of nonzero coordinates in $(\alpha_1, \ldots, \alpha_{2d})$.  
For instance, let $\vect{v}_1,\cdots,\vect{v}_5$ be input vectors which are all from $D_{\cL,s}$ for some large $s$ and we want to find the collision in $q\cL$ for some positive integer $q$. Suppose that we find a combination $\vect{w}_1=\vect{v}_1+\vect{v}_2-(\vect{v}_1+\vect{v}_3) \in q\cL$ and another combination $\vect{w}_2=\vect{v}_2+\vect{v}_3-(\vect{v}_4+\vect{v}_5)\in q\cL$, then by Theorem \ref{theorem_SIS}, one would expect that $\vect{w}_1/q\sim D_{\cL,\sqrt{2}s/q}$ and $\vect{w}_2/q\sim D_{\cL,\sqrt{4}s/q}$. This means that the output of the exhaustive search algorithm by Kirchner and Fouque will behave like samples taken from discrete Gaussian distributions with different parameters, making it difficult to keep track of the standard deviation after several steps of the algorithm, and to obtain samples from the Discrete Gaussian distribution at the desired parameter above the smoothing parameter. 
We overcome this issue by showing that there is a combination of the input vectors with a {\em fixed} Hamming weight that is in $q\cL$ as mentioned in Section~\ref{sec:overview}.

There are other technical details that were overlooked in~\cite{Kirchner16}. In particular, one needs be careful with respect to the errors, both in the probability of failure and the statistical distance of the input/output. Indeed the algorithm performs an exponential number
of steps, it is not enough to show that the algorithm succeeds with
``overwhelming probability'' and that the output has a  
``negligible statistical distance'' from the desired output. However, many of the claimed error bounds in ~\cite{Kirchner16}
are not proven making it difficult to verify the proof of the Exhaustive
Search (Theorem~3.4), and of Theorem~3.6.

\section{Preliminaries}\label{section_preliminaries}

Let $\nat = \{1,2,\ldots, \}$. We use bold letters $\vect{x}$ for vectors and denote a vector's coordinates with indices $x_i$. We use $\log$ to represent the logarithm base 2 and $\ln$ to represent the natural logarithm. Throughout the paper, $n$ will always be the dimension of the ambient space $\real^n$.
We will denote the principal branch of Lambert's W function by $W(x)$,
see \cite{CGHJK96} for an introduction to this function.

\topic{Lattices}
A \emph{lattice} $\cL$ is a discrete subgroup of $\mR^{m}$,
or equivalently the set \[\cL(\vect{b}_{1},\dots,\vect{b}_{n})=\left\{ \sum_{i=1}^{n}x_{i}\vect{b}_{i} ~:~ x_{i}\in\mZ\right\} \]
of all integer combinations of $n$ linearly independent vectors
$\vect{b}_{1},\dots,\vect{b}_{n} \in \mR^{m}$. Such $\vect{b}_i$'s form a \emph{basis} of $\cL$.
The lattice $\cL$ is said to be \emph{full-rank} if $n=m$.
We denote by $\lambda_{1}(\cL)$ the first minimum of $\cL$, defined as the length of a shortest non-zero vector of $\cL$.

For a rank $n$ lattice $\lat \subset \mR^n$, the {\em dual lattice}, denoted $\lat^*$, is defined as the set of all points in $\mathsf{span} (\lat)$ that have integer inner products with all lattice points,
\[ \lat^* = \{ \vect{w} \in \mathsf{span}(\lat) : \forall \vect{y} \in \lat, \langle \vect{w},\vect{y}\rangle \in \mZ \}\;. \]
Similarly, for a lattice basis $\basis = (\vect{b}_1,\ldots, \vect{b}_n)$, we define the dual basis $\basis^*=(\vect{b}_1^*,\ldots, \vect{b}_n^*)$ to be the unique set of vectors in $\mathsf{span}(\lat)$ satisfying $\langle \vect{b}_i^*, \vect{b}_j\rangle = 1 $ if $i = j$, and $0$, otherwise. It is easy to show that $\lat^*$ is itself a rank $n$ lattice and $\basis^*$ is a basis of $\lat^*$. Given a lattice  $\basis = (\vect{b}_1,\ldots, \vect{b}_n)$, we denote $\|\basis\|_2= \max\limits_{i} \|\vect{b}_i \|$.

\topic{Kissing number and related quantities}
For any lattice $\cL \subset \real^n$ and $d>0$, let $N(\cL,r)$ denote the number of nonzero lattice
vectors of length at most $r$.
A natural question is to bound this quantity in terms of $r$. When $r<\lambda_1(\cL)$, only the origin
lies inside the ball so $N(\cL,r)=0$. When $r=\lambda_1(\cL)$, this quantity 
is known as the kissing number $\tau(\cL)$ of the lattice:
\[\label{sym:kissing_number}
    \tau(\cL)=|\set{\vect{x}\in\cL:\|\vect{x}\|=\lambda_1(\cL)}|.
\]
Finally when $r\to\infty$, $N(\cL,r)=\tfrac{r^n\vol(B_{n}(1))}{\det\cL}+o(r^n)$ by the geometric interpretation
of the determinant of a lattice. The precise behavior for intermediate values of $r$, however, is
unclear and for that reason we introduce the quantity
\begin{equation}
\label{sym:gamma_number}
    \gamma(\cL)=\inf \{\gamma:\forall r\geq 1,\; N(\cL,r\lambda_1(\cL))\leq \gamma \cdot r^n\}.
\end{equation}
It is clear by the definition that $\gamma(\cL)\geq\tau(\cL)$.
The best known upper bound on this quantity comes from the breakthrough work of Kabatyanskii and Levenshtein \cite{KL78}
which implies \cite{PS09} that
\begin{equation}\label{eq:bound_gamma_KL}
    \gamma(\cL)\leqslant 2^{0.402n}.
\end{equation}

\begin{remark}
To the best of our knowledge, we do not know any family of lattices with kissing number $2^{\Omega(n)}$.\footnote{Serge Vl{\u{a}}du{\c{t}}~\cite{vluaduct2019lattices} gave a construction for a set of lattices and claimed their kissing number is $2^{0.0338n+o(n)}$, while Bennett,  Golovnev, and Stephens-Davidowitz~\cite{BGS24difficulty} recently showed Vl{\u{a}}du{\c{t}}'s construction is invalid. Therefore, showing the existence of a family of lattices with exponential kissing numbers remains an open problem.}
Given this, it seems reasonable to conjecture that for any lattice that one would come across in practice, the kissing number is $2^{o(n)}$. Also, 
given the close connection between $\tau(\cL)$ and $\gamma(\cL)$, we conjecture
that $\gamma(\cL)$ is also $2^{o(n)}$ for almost all lattices. We leave it as an open problem whether there exists a family of lattices for which $\gamma(\cL)$ is exponential in the dimension of the lattice. 
In view of the fact that $\gamma(\cL)$ can be anywhere between $1$ and $2^{0.402n}$,
we will study the dependence of the time complexity of our algorithms on $\gamma(\cL)$ by introducing
\begin{equation}
\label{sym:beta_number}
    \beta(\cL)=\gamma(\cL)^{1/n}.
\end{equation}
The upper bound above can then be reformulated as $\beta(\cL)\leqslant 2^{0.402}$ for any lattice $\cL$. 
\end{remark}

\topic{Probability distributions}
Given two random variables $X$ and $Y$ on a set $E$, we denote by $\sddist$ the
\emph{statistical distance} between $X$ and $Y$, which is defined by
\begin{align*}
    \sddist(X,Y) &=\tfrac{1}{2}\sum_{z\in E}\left|\Prob{X}{X=z}-\Prob{Y}{Y=z}\right| \\
    &= \sum_{z\in E\: : \:\Prob{X}{X=z}> \Prob{Y}{Y=z} }\left(\Prob{X}{X=z}-\Prob{Y}{Y=z}\right) \;.
\end{align*}
We write $X$ is $\eps$-close to $Y$ to denote that the statistical distance between $X$ and $Y$ is at most $\eps$.
Given a finite set $E$, we denote by $U_E$ an uniform random variable on $E$, i.e., for all $x\in E$, $\Prob{U_E}{U_E=x}=\tfrac{1}{|E|}$. 

\topic{Discrete Gaussian Distribution}
For any $s>0$, define $\rho_s(\vect{x})=\exp(-\pi\twonorm{\vect{x}}^2/s^2)$ for all $\vect{x}\in\mR^n$.
We write $\rho$ for $\rho_1$. For a discrete set $S$, we extend $\rho$ to sets by
$\rho_s(S)=\sum_{\vect{x}\in S}\rho_s(\vect{x})$. Given a lattice $\lattice{L}$,
the \emph{discrete Gaussian} $\DGauss{\lattice{L}}{s}$ is the distribution over $\lattice{L}$
such that the probability of a vector $\vect{y}\in\lattice{L}$ is proportional to $\rho_s(\vect{y})$:
\[
  \Prob{X\sim\DGauss{\lattice{L}}{s}}{X=\vect{y}}=\frac{\rho_s(\vect{y})}{\rho_s(\lattice{L})}.
\]

\topic{Data processing inequality}\label{sec:data_process_ineq} When analyzing the output distribution of an algorithm, it is
often convenient to assume that the input distribution is ideal (e.g. uniform). On the other hand,
we will want to run the algorithm on non-ideal input distribution (e.g. with a slight deviation
from uniform). In this case, the output distribution will deviate from the ideal output distribution
and it is important to quantify this divergence. The statistical distance satisfies the following
useful inequality, known as the \emph{data processing inequality}:
\[
    \sddist(f(X),f(Y))\leqslant \sddist(X,Y)
\]
for any two distributions $X$ and $Y$ and any (possibly randomized) algorithm $f$. In other words,
the error does not increase under the application of $f$.

\subsection{Lattice problems}

The following problems play a central role in this paper.
For convenience, when we discuss the running time of algorithms solving the problems
below, we ignore polynomial factors in the bit-length of the individual input basis vectors (i.e. we assume that the input basis has bit-size polynomial in the ambient
dimension $n$).

\begin{definition}
\label{def:dgs}
For $\delta = \delta(n) \geq 0$, $\sigma$ a function that maps lattices to non-negative real numbers, and $m = m(n) \in \mathbb{N}$, $\delta\text{-}\DGS_\sigma^m $ (the Discrete Gaussian Sampling problem) is defined as follows: 
The input is a basis $\basis$ for a lattice $\lat \subset \mR^n$ and a parameter $s > \sigma(\lat)$. The goal is to output a sequence of $m$ vectors whose joint distribution is $\delta$-close to $m$ independent samples from $D_{\lat, s}$.
\end{definition}

We omit the parameter $\delta$ if $\delta = 0$, and the parameter $m$ if $m = 1$. We stress that $\delta$ bounds the statistical distance between the \emph{joint} distribution of the output vectors and $m$ independent samples from $D_{\lat,s}$. We consider the following lattice problems.

\begin{definition}
The search problem $\SVP$ (Shortest Vector Problem) is defined as follows: The input is a basis $\basis$ for a lattice $\lat \subset \mR^n$. The goal is to output a vector $\vect{y} \in \lat$ with $\|\vect{y}\| = \lambda_1(\lat)$.
\end{definition}

\begin{definition}
The search problem $\CVP$ (Closest Vector Problem) is defined as follows: The input is a basis $\basis$ for a lattice $\lat \subset \mR^n$ and a target vector $\vect{t} \in \mR^n$. The goal is to output a vector $\vect{y} \in \lat $ with $\|\vect{y} - \vect{t}\| =  \mathsf{dist}(\vect{t}, \lat)$.
\end{definition}

\begin{definition}
For $\alpha = \alpha(n) < 1/2$, the search problem $\alpha\text{-}\BDD$ (Bounded Distance Decoding) is defined as follows: The input is a basis $\basis$ for a lattice $\lat \subset \mR^n$ and a target vector $\vect{t} \in \mR^n$ with $\mathsf{dist}(\vect{t},\cL) \leq \alpha \cdot \lambda_1(\lat)$. The goal is to output a vector $\vect{y} \in \lat $ with $\|\vect{y} - \vect{t}\| =  \mathsf{dist}(\vect{t}, \lat)$.
\end{definition}

Note that while our other problems become more difficult as the approximation factor $\gamma$ becomes smaller, $\alpha\text{-}\BDD$ becomes more difficult as $\alpha$ gets larger. 

\begin{definition}
The search problems $\gamma$-$\CVPP$ and $\gamma$-$\BDDP$ are the preprocessing analogues of $\gamma$-$\CVP$ and $\gamma$-$\BDD$ respectively. The input for preprocessing is a basis $\mathbf{B}$ of lattice $\cL \subset \real^n$. Given the advice from the preprocessing algorithm and the target vector $\vect{t}\in \cR^n$. The goal is to return solution of $\gamma$-$\CVP$ and $\gamma$-$\BDD$ respectively. The preprocessing algorithm is allowed to take arbitrary time.
\end{definition}

For a lattice $\lattice{L}$ and $\varepsilon>0$, the \emph{smoothing parameter} $\eta_\varepsilon(\lattice{L})$
is the smallest $s$ such that $\rho_{1/s}(\dual{\lattice{L}})=1+\varepsilon$. Recall that if $\lattice{L}$
is a lattice and $\vect{v}\in\lattice{L}$ then $\rho_s(\lattice{L}+\vect{v})=\rho_s(\lattice{L})$ for all $s$. The \emph{smoothing parameter} has the following well-known property.

\begin{lemma}[{\cite[Claim~3.8]{Regev09}}]\label{lemma_above}
	For any lattice $\cL\subset \mathbb{R}^n$, $\vect{c} \in \mathbb{R}^n$, $\epsilon > 0$, and $s \geq \eta_\epsilon(\cL)$,
	\[
	\frac{1 - \epsilon}{1+\epsilon}\leq \frac{\rho_s(\cL + \vect{c})}{\rho_s(\cL)} \leq 1 \;.
	\]
\end{lemma}

\begin{corollary}
	\label{cor:above_smoothing}
	Let $\cL\subset \mathbb{R}^n$ be a lattice, $q$ be a positive integer, and let $s \geq \eta_\epsilon(q\cL)$. Let $C$ be a random coset in $\cL/q\cL$ sampled such that for any $\vec{c}\in {\cL}/{q\cL}$,\; $\Pr[C = q\cL + \vect{c}] = \frac{\rho_s(q \cL + \vect{c})}{\rho_s(\cL)}$. Also, let $U$ be a coset in $\cL/q\cL$ sampled uniformly at random. Then
	\[
	\sddist(C, U) \le 2\eps \;.
	\]
\end{corollary}
\begin{proof}
	By Lemma~\ref{lemma_above}, we have that \[\rho_s(q\cL) \geq \rho_s(q\cL + \vect{c}) \geq \frac{1-\eps}{1+\eps} \rho_s(q\cL) \;, \] for any $\vect{c} \in \cL/q\cL$ and hence, 
	\[
	q^n\rho_s(q\cL)\geq \sum_{\vect{c}\in \cL/q\cL}\rho_s(q\cL+\vect{c})=\rho_s(\cL)
	\]
	Therefore, 
	\[
	\frac{\rho_s(q\cL + \vect{c})}{\rho_s(\cL)} \geq \frac{1-\eps}{1+\eps} \cdot \frac{\rho_s(q\cL)}{\rho_s(\cL)} \geq \frac{1-\eps}{1+\eps} \cdot \frac{1}{q^n} \;.
	\]
	We conclude that 
	\begin{align*}
		\sddist(C, U) &= \sum_{\vect{c} \in \cL/q\cL \: : \: \Pr[\vect{C} = \vect{c}] < \Pr[U = \vect{c}]} \left(\Pr[U = \vect{c}] - \Pr[C = \vect{c}]\right) \\
		&\leq \sum_{\vect{c} \in \cL/q\cL \: : \: \Pr[C = \vect{c}] < \Pr[U = \vect{c}]} \Pr[U = \vect{c}] \left(1 - \frac{1-\eps}{1+\eps}\right) \\
		&\leq \sum_{\vect{c} \in \cL/q\cL} \Pr[U = \vect{c}]  \frac{2\eps}{1+\eps} \\
		&\leq \frac{2\eps}{1+\eps} \;,
	\end{align*}
	as needed. 
\end{proof}

We will need the following lemma which is initially proved in \cite{banaszczyk1993new}.

\begin{lemma}[{\cite[Lemma~2.13]{DRS14}}]\label{lemma:DGS-large-vector-bound}
	For any lattice $\cL\subset \real^n$, $s>0$ and any $t\geq 1$, 
	\[\Pr_{\vect{y}\sim \cD_{\cL,s}} \left[\|\vect{y}\|\geq t\sqrt{\frac{n}{2\pi}}\cdot s\right] \leq e^{-\frac{n}{2}(t-1)^2}.\]
\end{lemma}

The following lemma gives a bound on the smoothing parameter. 
\begin{lemma}[{\cite[Lemma~2.7]{ADRSD15}}]\label{lemma_smoothing}
	For any lattice $\cL \subset \real^n, \epsilon \in (0,1)$ and $k>1$, we have $k\eta_{\epsilon}(\cL) > \eta_{\epsilon^{k^2}}(\cL)$
\end{lemma}

Micciancio and Peikert~\cite{MP13} showed the following result about the resulting distribution from the sum of many Gaussian samples.

\begin{theorem}[{\cite[Theorem~3.3]{MP13}}]\label{theorem_SIS}
	Let $\cL$ be an $n$ dimensional lattice, $\vect{z}\in \intg^m$ a nonzero integer vector, $s_i\geq \sqrt{2}\|\vect{z}\|_\infty \cdot \eta_\epsilon(\cL)$, and $\cL+\vect{c_i}$ arbitrary cosets of $\cL$ for $i=1\cdots, m$. Let $\vect{y_i}$ be independent vectors with distributions $D_{\cL+\vect{c_i},s_i}$, respectively. Then the distribution of $\vect{y}=\sum\limits_iz_i\vect{y_i}$ is $\epsilon m$ close to $D_{Y,s}$, where $Y=gcd(\vect{z})\cL+\sum\limits_iz_i\vect{c_i}$, and $s=\sqrt{\sum\limits(z_is_i)^2}$.
\end{theorem}

We need to recall the definition of the \emph{honest} Discrete Gaussian Sampling problem,
introduced in \cite{ADRSD15}.

\begin{definition}[{\cite[Definition~5.1]{ADRSD15}}]
	For $\varepsilon\geqslant 0$, $\sigma$ a function that maps lattices to non-negative real numbers,
	and $m\in\mathbb{N}$, the \emph{honest} Discrete Gaussian Sampling problem
	$\varepsilon$-$\operatorname{hDGS}^m_\sigma$
	is defined as follows: the input is a basis $B$ for a lattice $\cL\subset\mathbb{R}^n$ and a parameter $s>0$.
	The goal is for the output distribution to be $\varepsilon$-close to $D^{m'}_{\cL,s}$ for some
	independent random variable $m'\geqslant0$. If $s>\sigma(\cL)$ then $m'$ must be equal to $m$.
\end{definition}

\begin{theorem}[{\cite[Theorem~5.11]{ADRSD15}}]\label{theorem_dgsabovesmoothing}
	Let $\sigma$ be the function that maps a lattice $\cL$ to $\sqrt{2}\eta_{1/2}(\cL)$.
	Then, there is an algorithm that solves $\exp(-\Omega(\kappa))$-$\operatorname{hDGS}^{2^{n/2}}_{\sigma}$
	in time $2^{n/2+\operatorname{polylog}(\kappa)+o(n)}$ for an $\kappa\geqslant\Omega(n)$.
\end{theorem}

We will need the following reduction from $\alpha$-$\BDD$ to $\DGS$ that was shown in~\cite{DRS14}.

\begin{theorem}[{\cite[Theorem~3.1]{DRS14}}]\label{theorem_BDDtoDGS}
	For any $\epsilon \in(0,1/200)$, lattice $\cL\subset \real^n$, and  
	$\alpha \leq \frac{\sqrt{\ln(1/\epsilon)/\pi-o(1)}}{ 2\eta_\epsilon(\cL^*)\lambda_1{\cL}}.$
	There exists an algorithm that solves $\alpha$-$\BDDP$
	using $m \cdot poly(n)$ arithmetic operations where
	$m = O(\frac{n \log(1/\epsilon)}{\sqrt{\epsilon}})$. Moreover, the preprocessing advice  only consists
	of $m$ vectors sampled from $D_{\cL^*,\eta_\epsilon(\cL^*)}$.
\end{theorem}

\begin{remark}\label{remark:space_complexity_theorem_BDDtoDGS}
	We are going to use this reduction in the superpolynomial regime: typically $m$ will be exponential
	in $n$. This leaves unclear the space complexity of the reduction. The reduction works by evaluating
	a polynomial number of times functions of the form $\sum_{i=1}^mf_i(\vect{x})$ where each $f_i$ is a
	polynomial time computable function that depends on the $i^{th}$ DGS sample. Furthermore, all
	the complexities above are in terms of arithmetic operations, not bit complexity. If we assume
	that all the DGS samples have $\poly(n)$ bit-size then the reduction has time complexity $m\cdot \poly(n)$
	and space complexity $O(\poly(n)+\log m)$ \emph{excluding the storage space of the $m$ vectors provided by the DGS}.
	Finally, as noted in the proof of the theorem in \cite{DRS14}, only the preprocessing is
	probabilistic and with probability at least $1-2^{-\Omega(n)}$ over the choice 
	of the samples, the algorithm will deterministically solve all $\alpha$-$\BDD$ instances.
\end{remark}

We need the following relation between the first minimum of lattice and the smoothing parameter of dual lattice. We will use this to compute the decoding distance of the $\BDD$ oracle.

\begin{lemma}[Variant of {\cite[Lemma~6.1]{ADRSD15}}]\label{lemma_smoothinglambda}
	For any lattice $\cL \subset \mathbb{R}^n$ and $\epsilon \in (0,1)$,
	\begin{align}
		\sqrt{\frac{\ln(1/\epsilon)}{\pi}} < \lambda_1(\cL) \eta_\epsilon(\cL^*) < \sqrt{\frac{\beta(\cL)^2n}{2\pi e}}\cdot \epsilon^{-1/n}\cdot(1+o(1)), \label{ineq:smoothingtosvlarge}
	\end{align}
	and if $\epsilon \leq (e/\beta(\cL)^2+o(1))^{-\frac{n}{2}}$, we also have
	\begin{align}
		\sqrt{\frac{\ln(1/\epsilon)}{\pi}} < \lambda_1(\cL) \eta_\epsilon(\cL^*) < \sqrt{\frac{\ln(1/\epsilon)+n\ln\beta(\cL)+o(n)}{\pi}}. \label{ineq:smoothingtosvsmall}
	\end{align}
	as $n$ tends to infinity.
\end{lemma}

\begin{remark}
	As noted in \cite{ADRSD15} below Lemma~6.1, the inequality \eqref{ineq:smoothingtosvlarge}
	actually holds for all $\varepsilon\in(0,1)$ so we dropped the condition in the first case.
\end{remark}

\begin{remark}
	In Lemma~6.1 of \cite{ADRSD15}, $\beta$ comes from
	Lemma~4.2 of the same paper and only needs to satisfy the equation
	$|\cL\cap T_r|\leqslant\beta^{n+o(n)}r^n$ for all $r$ where
	$T_r=\{x\in\mathbb{R}^n:r\leqslant\|x\|\leqslant(1+\tfrac{1}{n})r\}$,
	assuming the lattice is normalized so that $\lambda_1(\cL)=1$.
	A trivial upper bound on $|\cL\cap T_r|$ is $N(\cL,r')$,
	the number of point in $\cL$ of radius at most $r'=(1+\tfrac{1}{n})r$.
	By our definitions \eqref{sym:gamma_number} and \eqref{sym:beta_number}, this is bounded by
	$\gamma(\cL)\big((1+\tfrac{1}{n})r\big)^{n}\leqslant\beta(\cL)^{n}er^n$ and therefore we can replace $\beta$ by our $\beta(\cL)$.
\end{remark}

The following theorem proved in \cite{CCL18}, is required to solve $\svp$ by an exponential number of calls to $\alpha$-$\BDD$ oracle.
\begin{theorem}[{\cite[Theorem~8]{CCL18}}]\label{theorem_enlargeBDD}
	Given a basis matrix $\basis \in \mathbb{R}^{n \times n}$ for lattice $\cL(\basis)\subset \mathbb{R}^n$, a target vector $\vect{t} \in \mathbb{R}^n$, an $\alpha$-$\BDD$ oracle $\BDD_\alpha$ with $\alpha < 0.5$, and an integer scalar $p>0$. Let $f^{\alpha}_p:\mathbb{Z}^n_p\rightarrow \mathbb{R}^n$ be 
	$f^\alpha_p(\vect{s})= -p \cdot \BDD_\alpha(\cL, (\basis \vect{s} - \vect{t})/p) + \basis\vect{s}$, then the list $m=\{f^\alpha_p(\vect{s})\mid \vect{s} \in \mathbb{Z}^n_p\}$ contains all lattice points within distance $p\alpha\lambda_1(\cL)$ to $\vect{t}$.
\end{theorem}

\begin{remark}\label{rm:unnecessary_condition}
	In the original Theorem~8 from \cite{CCL18}, there is an extra condition on the target vector $\vect{t}$, ie, $\vect{t}$ needs to satisfy $\mathsf{dist}(\cL,\vect{t})\leq \alpha\lambda_1(\cL)$. However, a thorough inspection of the proof reveals that this condition is unnecessary, and essentially the same proof follows without it.
\end{remark}

We will need the following theorems to sample the $\DGS$ vectors with a large width. 

\begin{theorem}[\cite{ADRSD15},Proposition 2.17]\label{intial_dgs}
	For any $\epsilon\leq 0.99$, there is an algorithm that takes as input a lattice $\cL \in \real^n$,
	$M\in \intg_{>0}$ (the desired number of output vectors), and $s > 2^{n \log \log n/\log n} \cdot \eta_\epsilon(\cL)$, and outputs $M$ independent
	samples from $D_{\cL,s}$ in time $M\cdot poly(n)$.
\end{theorem}

\begin{lemma}[{\cite[Lemma~5.12]{ADRSD15}}]\label{lemma_mod2}
	There is a probabilistic polynomial-time algorithm that takes as input a lattice $\cL\subset \real^n$
	of rank $n$ and an integer $a$ with $n/2\leq a <n $ and returns a super lattice $\cL'\supset \cL$ of index $2^a$ with $\cL' \subseteq \cL/2$ such that for any $\eps \in (0,1)$, we have $\eta_{\eps'}(\cL')\leq \eta_\eps(\cL)/\sqrt{2}$ with probability at least $1/2$ where $\eps':= 2\eps^2 +2^{(n/2)+1-a}(1+\eps)$.
\end{lemma}

\subsection{Quantum Computation} In this paper we use the Dirac ket-bra notation. A qubit is a unit vector in $\mathbb{C}^2$ with two (ordered) basis vectors $\{\ket{0},\ket{1}\}$. 
The following gates form a universal set of gates
\begin{align*}
	H&= \frac{1}{\sqrt{2}}\begin{bmatrix}1 &1\\1&-1\end{bmatrix},\
	S=\begin{bmatrix}1 &0\\0& i\end{bmatrix},\    T=e^{i\pi/8}\begin{bmatrix}e^{-i\pi/8} &0\\0& e^{i\pi/8}\end{bmatrix},\\
	CNOT&=\ket{0}\bra{0}\otimes I+ \ket{1}\bra{1}\otimes X.
\end{align*}
The Toffoli gate, a three-qubit gate, is defined by
\[
\text{Toffoli}\ket{a}\ket{b}\ket{c}=
\begin{cases}
	\ket{a}\ket{b}\ket{1\oplus c}, & \mbox{if $a=b=1$;}\\
	\ket{a}\ket{b}\ket{c}, & \mbox{otherwise},
\end{cases}
\]
for $a,b,c\in\{0,1\}$. Toffoli gate can be efficiently decomposed into $CNOT, H, S,$ and $T$ gates~\cite{NC00} and hence it is considered as an elementary quantum gate in this paper. It is easy to see that a NAND gate can be implemented by a Toffoli gate: $\text{Toffoli}\ket{a}\ket{b}\ket{1}=\ket{a}\ket{b}\ket{\text{NAND}(a,b)}$,
where $\text{NAND}(a,b)=0$, if $(a,b)=(1,1)$, and $\text{NAND}(a,b)=1$, otherwise.
In particular, Toffoli gate together with ancilla preparation are universal for classical computation, that is for any classical function, we can implement it as (controlled) quantum one, although at a non-negligible cost due to the reversibility of quantum circuits.
\begin{theorem}[\cite{Bennett89,LS90}]\label{th:irreversible_to_reversible}
	Given any $\epsilon>0$ and any classical computation with running time $T$ and space complexity $S$,
	there exists an equivalent reversible classical computation with running time $O(T^{1+\epsilon}/S^\epsilon)$
	and space complexity $O(S(1+\ln(T/S)))$.
\end{theorem}

\begin{corollary}\label{cor:irreversible_to_quantum}
	Given any $\epsilon>0$ and any classical computation with running time $T$ and space complexity $S$,
	there exists an equivalent quantum circuit of size $O(T^{1+\epsilon}/S^\epsilon)$
	using $O(S(1+\ln(T/S)))$ qubits.
\end{corollary}

\topic{Search problem} One of the most well-known quantum algorithms is Grover's unstructured search algorithm~\cite{DBLP:conf/stoc/Grover96}.
Suppose we have a set of objects named $\{0,1,\dots, N-1\}$, of which some are \emph{targets}.
We say that an oracle $\mathcal{O}$ \emph{identifies the targets} if, in the classical (resp. quantum)
setting, $\mathcal{O}(i)=1$ (resp. $\mathcal{O}\ket{i} = -\ket{i}$) when $i$ is a target
and $\mathcal{O}(i)=0$ (resp. $\mathcal{O}\ket{i}= \ket{i}$) otherwise.
Given such an oracle $\mathcal{O}$,
the goal is to find a target $j \in \{0,1,\dots, N-1\}$ by making queries to the oracle $\mathcal{O}$.

In the search problem, one will try to minimize the number of queries to the oracle. In the classical setting, one need $\mathcal{O}(N)$ queries to solve this problem. Grover, on the other hand, provided a quantum algorithm, that solves the search problem with only $\mathcal{O}(\sqrt{N})$ queries~\cite{Grover96}.  When the number of targets is unknown, Brassard \emph{et al.} provided a modified Grover algorithm that solves the search problem with $\mathcal{O}(\sqrt{N})$ queries~\cite{BBHT96}, which is of the same order as the query complexity of the Grover search.  Moreover, D{\"{u}}rr and H{\o}yer showed given an unsorted table of $N$ values,
there exists a quantum algorithm that finds the index of minimum with only $\mathcal{O}(\sqrt{N})$ queries~\cite{DH96}, with constant probability of error.
\begin{theorem}[\cite{DH96}, Theorem~1]\label{theorem_Qmin}
	Let $T[1, 2, . . . , N]$ be an unsorted table of $N$ items, each holding a value from
	an ordered set. Suppose that we have a quantum oracle $\mathcal{O}_T$ such that $\mathcal{O}_T\ket{i}\ket{0} = \ket{i}\ket{T[i]}$. Then there exists a quantum algorithm that finds the index $y$ such that $T[y]$ is the minimum with probability at least $1/$2 and $\mathcal{O}(\sqrt{N})$ queries to $\mathcal{O}_T$.
\end{theorem}

\subsection{Probability} 
We need the following lemma on the distribution of vector inner product which directly follows from the Leftover Hash Lemma~\cite{impagliazzo1989pseudo}.
\begin{lemma}\label{lemma:ip}
	Let $\mathbb{G}$ be a finite abelian group, and let $f$ be a positive integer. Let $\cY \subseteq \{0,1\}^f$. Define the inner product $\langle\cdot,\cdot \rangle:\mathbb{G}^f\times \cY \to \mathbb{G}$ by $\langle x,y \rangle=\sum_i x_i y_i$ for all $x\in \mathbb{G}^f, y \in \cY$. Let $X, Y$ be independent and uniformly random variables on $\mathbb{G}^f, \cY$, respectively.
	Then
	$$
	\sddist((\inProd{X}{Y},X),(U_{\mathbb{G}},X))
	\leq \frac{1}{2} \cdot \sqrt{\frac{|\mathbb{G}|}{|\cY|}} \;,
	$$
	where $U_{\mathbb{G}}$ is uniform in $\mathbb{G}$ and independent of $X$. 
\end{lemma}

We will also need the Chernoff-Hoeffding bound~\cite{hoeffding1963probability}.

\begin{lemma}\label{lem:Chernoff}
	Let $X_1,\hdots,X_{M}$ be the independent and identically distributed random boolean variables of expectation $p$. Then for $\varepsilon>0$,
	\[
	\Pr\left[\frac{1}{M}\sum\limits_{i=1}^M X_i \leq p(1 -\delta)     \right]
	\leq \left(\frac{e^{-\delta}}{(1-\delta)^{1-\delta}}\right)^{pM}.
	\]
\end{lemma}

\section{Algorithms with a time-memory tradeoff for lattice problems}
\label{tradeoff}

In this section, we present a new algorithm for Discrete Gaussian sampling above the smoothing parameter.  

\subsection{Algorithm for Discrete Gaussian Sampling}
We now present the main result of this section.

\begin{theorem}
\label{thm:DGS}
	Let $n \in \mathbb{N}, q \ge 2, d \in [1,n]$ be positive integers, and let $\eps > 0$. Let $C$ be any positive integer. Let $\cL$ be a lattice of rank $n$,
	and let  $s\geq 2\sqrt{d} q \eta_\eps (\cL)$. There is an algorithm that, given $N =  160d^2 \cdot C \cdot q^{n/d} $ independent samples from $D_{\cL,s}$, outputs a list of vectors that is $(4\eps^{2d} N + 11C q^{-5n/2})$-close to $C q^{n/d}$ independent vectors from $D_{\cL,\frac{\sqrt{8d+1}}{q}s}$. The algorithm runs in time $C \cdot (10e\cdot d)^{8d}\cdot  q^{8n+n/d+o(n)}$ and requires memory $\poly(d)\cdot q^{n/d}$. 
\end{theorem}

\begin{algorithm}
\caption{TradeOffSieve($L$)\label{alg:TradeOffSieve}}
\begin{algorithmic}[1]
    \REQUIRE{$L=\set{\vect{x}_1,\ldots,\vect{x}_N}$ is a list of $N$ vectors in $\cL$}
    \STATE $L_1\gets\{\vect{x}_1,\hdots,\vect{x}_{\frac{N}{2}}\}$
    \STATE $L_2\gets\{\vect{x}_{\frac{N}{2}+1},\hdots,\vect{x}_N\}$
    \STATE $Q\gets 0$
    \WHILE{$Q<q^{n/d}$ and $L_1$ is not empty}
        \STATE $\vect{v}\gets$ be the first vector in $L_1$.
        \STATE\label{tradeoff_svp:step:3} Find $8d$  vectors (by trying all $8d$-tuples)
            $\vect{x}_{i_1}, \ldots, \vect{x}_{i_{8d}}$ from $L_2$ s.t. $\vect{c}_{i_1}+ \cdots+ \vect{c}_{i_{8d}}-\vect{v} \in q \cL$.
        \IF{such vectors exist}
            \STATE\label{tradeoff_svp:step:8} Output the vector $\frac{\vect{x}_{i_1}+ \cdots+ \vect{x}_{i_{8d}}-\vect{v}}{q} \in \cL$
            \STATE $Q\gets Q+1$.
            \STATE Remove vectors $\vect{x}_{i_1}, \cdots, \vect{x}_{i_{8d}}$ from $L_2$
        \ENDIF
        \STATE Remove vector $\vec{v}$ from $L_1$
    \ENDWHILE
\end{algorithmic}
\end{algorithm}

\begin{proof}
We prove the result for $C = 1$, and the general result follows by repeating the algorithm. 
Let $\{\vect{x}_1,\hdots,\vect{x}_N\}$ be the $N$ input vectors and let $\{\vect{c}_1,\hdots,\vect{c}_N\}$ be the corresponding cosets in  $\cL /q\cL$.
We will analyze Algorithm~\ref{alg:TradeOffSieve}; note that it produces samples in a streaming
fashion. The time complexity of the algorithm is 
\[ \frac{N}{2}\cdot {N/2 \choose 8d} \le \frac{N}{2}\left(\frac{eN}{16d}\right)^{8d} \le (10e\cdot d)^{8d}\cdot q^{8n+n/d+o(n)}\;, \]
and memory requirement of the algorithm is immediate. We now show correctness:
we will make repeated use of the data processing inequality (see Section~\ref{sec:data_process_ineq})
and accumulate the error terms until the end of the proof.

Let $\eps'=\eps^{2d}$ so that $s\geq \sqrt{2}\eta_{\eps'}(q\cL)$ by Lemma \ref{lemma_smoothing}.
First, we can assume that the vectors $\vect{x}_i$ for $i \in [N]$ are obtained by first sampling
$\vect{c}_i \in {\cL}/{q\cL}$ such that $\Pr[\vect{c}_i = \vect{c}] = \Pr[D_{\cL,s} \in q\cL+\vect{c}]$
and then sampling the vector $\vect{x}_i$ according to $D_{q\cL+\vect{c}_i, s}$. Indeed,
let $X\sim\DGauss{\cL}{s}$ and $C$ be a random coset in $\cL/q\cL$ sampled such that $\Pr[C = q\cL + \vect{c}] = \frac{\rho_s(q \cL + \vect{c})}{\rho_s(\cL)}$,
then for any $\vect{x}\in\cL$,
\[
    \Pr[X=\vect{x}]
        =\frac{\rho_s(\vect{x})}{\rho_s(\cL)}
        =\frac{\rho_s(q \cL + \vect{c})}{\rho_s(\cL)} \frac{\rho_s(\vect{x})}{\rho_s(q \cL + \vect{c})}
        =\Pr[C=\vect{c}]\Pr_{Y\sim\DGauss{q\cL+\vect{c}}{s}}[Y=\vect{x}]
\]
where $\vect{c}=q\cL+\vect{x}$.
Moreover, by Corollary~\ref{cor:above_smoothing}, this distribution is $2\eps' N$-close
to sampling $\vect{c}_i$ for $i \in [N]$, independently and uniformly from $\cL /q \cL$, and then sampling the vectors $\vect{x}_i$ according to $D_{q\cL+\vect{c}_i, s}$. 
We now assume that the input is sampled from this distribution.

Second, we can assume that the algorithm initially gets only the corresponding cosets as input,
and the vectors $\vect{x}_{i_j} \in q\cL + \vect{c}_{i_j}$ for $j \in [8d]$, and $\vect{v} \in q\cL + \vect{c}$
are sampled from $D_{q\cL+\vect{c}_{i_j}, s}$ and $D_{q\cL+\vect{c},s}$ only before such a tuple is needed in Step~\ref{tradeoff_svp:step:8} of the algorithm.
Indeed, notice that the test at line~\ref{tradeoff_svp:step:3} does not actually depend on the particular
value of the $\vect{x}_{i_j}$ and $\vect{v}$ but only on their cosets.
Since any input vector is used only once in Step~\ref{tradeoff_svp:step:8}, these samples are independent of all prior steps.
This implies, by Theorem~\ref{theorem_SIS}, that the vector obtained in Step~\ref{tradeoff_svp:step:8} of the algorithm is $2\eps'$-close to being distributed as $D_{\cL, s \frac{\sqrt{8d+1}}{q}}$.

It remains to show that our algorithm finds $q^{n/d}$ vectors (with high probability). Let $N'=\frac{N}{2}$ be an integer, $X$ be a random variable uniform over $(\cL/q\cL)^{N'}$, and let $Y$ be a random variable independent of $X$ and uniform over vectors in $\{0,1\}^{N'}$ with Hamming weight $8d$. The number of such vectors is 
\begin{equation}\label{eq:bound_number_vectors}
{N' \choose 8d} \ge \left(\frac{N'}{8d}\right)^{8d} \geq q^{8n}\;.
\end{equation}
Let $U$ be a uniformly random coset of $\cL/q\cL$. By  Lemma~\ref{lemma:ip} and \eqref{eq:bound_number_vectors},
we have 
\[
        \sddist((\inProd{X}{Y},X),(U,X))
            \leq \frac{1}{2} \cdot \sqrt{\frac{q^n}{ q^{8n}}}
            =\frac{1}{2}q^{-7n/2}\;.
\] 
By Markov inequality we have
\begin{align*}
    \Pr_{x\gets X}\left[\sddist(\inProd{x}{Y},U)\geqslant \frac{q^{-n}}{10}\right]
        &\leqslant\frac{10}{q^{-n}}\Exp_{x\gets X}\left[\sddist(\inProd{x}{Y},U)\right]\\
        &=\frac{20}{q^{-n}}\sddist((\inProd{X}{Y},X),(U,X))\\
        &\leqslant \frac{20}{q^{-n}}\cdot\frac{1}{2}q^{-7n/2}=10q^{-5n/2}.
\end{align*}
Hence, with probability at least $1-10q^{-5n/2}$ over the choice of $x \leftarrow X$,
we have that $\sddist(\inProd{x}{Y},U)\leqslant \frac{q^{-n}}{10}$ and thus for any $\vect{v}\in \cL/q\cL$,
\begin{equation}
  q^{-n} + \frac{q^{-n}}{10} >\Pr[\langle x, Y \rangle = \vect{v}\bmod q\cL] > q^{-n} - \frac{q^{-n}}{10}.
  \label{uniform}
\end{equation}
It follows that, by introducing a statistical distance of at most $10q^{-5n/2}$ on the input, we
can assume that the input vectors in list $L_2$ satisfy \eqref{uniform}.
Notice that after the algorithm found $i$ vectors for any $i < q^{n/d} $, it has removed $8id$ vectors from $L_2$.  We will show that for each vector from $L_1$ (which is uniformly sampled from $\cL/q\cL$) with constant probability we will find $8d$-vectors in Step ~\ref{tradeoff_svp:step:3}.

After $i<q^{n/d}$ output vectors have been found, there are $M=N'-8id$ vectors remaining in the list $L_2$. There are $\binom{M}{8d}$ different $8d$-combinations possible with vectors remaining in $L_2$. 
\begin{align*}
    \binom{N'}{8d}/\binom{M}{8d}
    &= \frac{N'\cdots (N'-8d+1)}{M\cdots (M-8d+1)}\\
    &< \left( \frac{N'-8d}{N'-8d(i+1)}\right)^{8d}\\
    &< \left(1+\frac{8dq^{n/d}}{N'-8dq^{n/d}-8d}\right)^{8d}\\
    &= \left(1+\frac{1}{10d-1-\frac{1}{q^{n/d}}}\right)^{8d}&&\text{since $N'=80d^2q^{n/d}$ for $C=1$}\\
    &\leq   \left(1+\frac{1}{10d-3/2}\right)^{8d}< \frac{5}{2}.
    \addtocounter{equation}{1}\tag{\theequation}\label{eq:ratio_vectors}
\end{align*}
At the beginning of the algorithm, there are ${N'\choose 8d}$ combinations, and hence by \eqref{uniform},
each of the $q^n$ cosets appears
at least $0.9q^{-n}{N'\choose 8d}$ times. After $i<q^{n/d}$ output vectors have been found, there are only $\binom{M}{8d}$ combinations left,
and 
${N'\choose 8d}-{M\choose 8d}$ possible combinations have been removed. We say that a coset $\vect{c}$ disappears if there is no set of $8d$ vectors in $L_2$ that add to $\vect{c}$. In order for a coset to disappear, all of the at least $0.9q^{-n}{N'\choose 8d}$
combinations from the initial list must be removed. Hence, the number of cosets that disappear is at most
$\tfrac{{N'\choose 8d}-{M\choose 8d}}{0.9q^{-n}{N'\choose 8d}}
<\tfrac{3/5}{0.9}q^n=\tfrac{2}{3}q^n$ distinct cosets by \eqref{eq:ratio_vectors}.
Hence with probability at least $1/3$, we find $8d$  vectors $\vect{x}_{i_1}, \ldots, \vect{x}_{i_{8d}}$ from $L_2$ such that $\vect{x}_{i_1}+ \cdots+ \vect{x}_{i_{8d}}-\vect{v} \in q \cL$.
By Chernoff-Hoeffding bound with probability greater than $1-e^{-d^2q^{n/d}}$, the algorithm finds at least $q^{n/d}$ vectors.
In total, the statistical distance from the desired distribution is
\[2\epsilon'N+2\epsilon'q^{n/d}+10\cdot q^{-5n/2}+e^{-d^2q^{n/d}} \leq 4\epsilon'N+11\cdot q^{-5n/2}.\]

\end{proof}

\begin{corollary}\label{cor:DGS}
    Let $n\in \mathbb{N}$, $q \in [4,\sqrt{n}]$ be an integer, and let $\eps = q^{-32n/q^2}$.
    Let $\cL$ be a lattice of rank $n$, and let  $s \geq \eta_{\epsilon}(\cL)$.
    There is an algorithm that outputs a list of vectors that is $q^{-\Omega(n)}$-close to
    $q^{16n/{q^2}}$ independent vectors from $D_{\cL,s}$.
    The algorithm runs in time $q^{13n+o(n)}$ and requires memory
    $\poly(n)\cdot q^{16n/{q^2}}$. 
\end{corollary}

\begin{proof}
Choose $d$ so that $16d-16< q^2\leqslant 16d$, which is possible when $q\geqslant 4$, and let $\alpha=q/\sqrt{8d+1}$ --- this is the ratio by which we decrease the Gaussian width in Theorem~\ref{thm:DGS} --- and note that $\alpha \geq 1.2$. 

Let  $p = \lceil 2\sqrt{d}q \rceil < q^2$ and $k$ be the smallest integer such that $\alpha^k \cdot p \ge 2^{n \log \log n/\log n}$. Thus $k = O({n \log \log n/\log n})$. Let $g= \alpha^k p s \ge 2^{n \log \log n/\log n}\cdot \eta_\epsilon (\cL)$.   By Theorem \ref{intial_dgs}, in time $N_0\cdot \poly(n)$, we get  $N_0=(160d^2)^k q^{n/d}$
samples from $D_{\cL,g}$.

We now iterate $k$ times the algorithm from Theorem~\ref{thm:DGS}.
Initially we have $N_0$ vectors. At the beginning of the $i$-th iteration for $i \le k-1$,
we have  $N_i:=N_0 \cdot (160d^2)^{-i}$ vectors that are $\Delta_i$-close to being independently distributed
from $D_{\cL,\alpha^{-i}g}$, where
$
    \alpha^{-i}g
        \geqslant \alpha p\cdot \eta_{\eps}(\cL)
$. 
Hence, we can apply Theorem~\ref{thm:DGS} and get $N_{i+1}=N_i/160d^2$ vectors
that are $\Delta_{i+1}$-close to being independently distributed from $D_{\cL,\alpha^{-(i+1)}g}$,
where $\Delta_{i+1}\leqslant\Delta_i+4\eps^{2d}N_i + 11(160d^2)^{k-i} q^{-5n/2}$. At each
iteration we had $N_i \geq 160d^2q^{n/d}$ vectors, a necessary condition to
apply Theorem~\ref{thm:DGS}. Therefore after $k$ iterations,
we have at least $N_k= N_0/(160d^2)^k = q^{n/d}$ samples
that are $\Delta_k$-close to being independently distributed from $D_{\cL,\alpha^{-k}g} $, where
\begin{align*}
    \Delta_k
        &\leqslant 11 q^{-5n/2} \sum_{i=1}^{k} (160d^2)^{k-i}+ \sum\limits_{i=0}^{k-1}10d\eps^{2d}N_i\\
        &\le 11 (160d^2)^k q^{-5n/2}+10dq^{-4n} q^{n/d}\sum_{i=0}^{k-1} (160d^2)^{k-i} && \text{since }16d\geqslant q^2\\
        &\le  \left(11q^{-5n/2}+10d q^{-4n + n/d}\right)(160d^2)^{k+1}=q^{-5n/2+o(n)}
        &&\text{since }(160d^2)^{k+1}=q^{o(n)}.
\end{align*}
Any vector distributed as $D_{\cL,ps}$ is in $p\cL$ with probability at least $p^{-n}$. We repeat the algorithm $2p^n=O(q^{2n})$ times to obtain $p^n \cdot 2 \cdot q^{n/d}$ vectors that are $2 p^nq^{-5n/2+o(n)} = q^{-n/2 + o(n)}$ close to $2p^n \cdot q^{n/d}$ independent samples from $D_{\cL,ps}$.
Of these samples obtained, we only keep vectors
that fall in $p\cL$ and divide them by $p$. Let $M=p^n\cdot 2\cdot q^{n/d}$.
By Chernoff-Hoeffding (Lemma~\ref{lem:Chernoff}) with $P=p^{-n}$, and $\delta=\tfrac{1}{2}$, the probability
to obtain less than $(1-\delta)PM=q^{n/d}$ samples is at most $\left(\frac{e^{-\delta}}{(1-\delta)^{1-\delta}}\right)^{PM}\leqslant e^{-\tfrac{1}{10}q^{n/d}}$.
Furthermore, $d\leqslant\tfrac{q^2+16}{16}$ and $q\mapsto\frac{\ln q}{16+q^2}$ is decreasing for
$q\geqslant4$, hence for $q\leqslant\sqrt{n}$,
\[
    q^{n/d}
        \geqslant e^{16n\frac{\ln q}{16+q^2}}
        \geqslant e^{16n\frac{\ln\sqrt{n}}{16+n}}
        \geqslant e^{16\ln\sqrt{n}-o(1)}
        =\Omega(n^8).
\]
Hence with probability greater than $1-e^{-\tfrac{1}{10}q^{n/d}}=1-q^{-\Omega(n^8)}$,
we get $q^{n/d}$ vectors from the distribution $D_{\cL,s}$. The statistical distance from the desired distribution is 
$q^{-\Omega(n^8)} +q^{-n/2 +o(n)}\leq q^{-n/2 +o(n)}$.
We repeat this for $\frac{q^{16n/q^2}}{q^{n/d}}$ times, to get $q^{16n/q^2}$ vectors. The total statistical distance from the desired distribution is $\frac{q^{16n/q^2}}{q^{n/d}} \cdot q^{-n/2 +o(n)} \leq q^{-\Omega(n)}$.
The total running time is bounded by
\[
    q^{2n}\left(\frac{q^{16n/q^2}}{q^{n/d}}\right) \left(\poly(n)\cdot N_0+\sum_{i=0}^{k-1}(10ed)^{8d} \cdot (160d^2)^{k-i}q^{8n+n/d+o(n)}\right)
    \leqslant q^{13n+o(n)}.
\]
The memory usage is slightly more involved: we can think of the $k$ iterations
as a pipeline with $k$ intermediate lists and we observe that as soon as a list
(at any level) has more than $160d^2q^{16n/q^2}$ elements, we can apply
Theorem~\ref{thm:DGS} to produce $q^{16n/q^2}$ vectors at the next level.
Hence, we can ensure
that at any time, each level contains at most $160d^2q^{16n/q^2}$ vectors, so in total
we only need to store at most $k\cdot 160d^2q^{16n/q^2}=\poly(n)q^{16n/q^2}$ vectors,
to which we add the memory usage of the algorithm of Theorem~\ref{thm:DGS} which
is bounded by $\poly(n)\cdot q^{n/d}\leqslant \poly(n)\cdot q^{16n/q^2}$.
Finally, we run the filter ($p\cL$) on the fly at the end of the $k$ iterations
to avoid storing useless samples.
\qed
\end{proof}

This tradeoff works for any $q \geq 4$, and the running time can be bounded by $c_1^{n+o(n)} \cdot q^{c_2 n}$ for some constants $c_1$ and $c_2$ that we have not tried to optimize.

\subsection{Algorithms for $\BDD$ and $\svp$}
\begin{theorem}\label{thm:BDD}
	Let $n \in \mathbb{N}, q \in [4,\sqrt{n}]$ be a positive integer. Let $\cL$ be a lattice of rank $n$. There is a randomized algorithm that solves $0.1/q$-$\BDD$ in time $q^{13n+o(n)}$ and requires memory $\poly(n)\cdot q^{16n/q^2}$. 
\end{theorem}

\begin{proof}
    Let $\epsilon = q^{\frac{-32n}{q^2}}$ and $s\geqslant \eta_\eps (\cL^*)$.
    From Corollary \ref{cor:DGS}, there exists an algorithm that outputs $q^{16n/q^2}$ vectors whose distribution
    is statistically close to $D_{\cL^*,s}$ in  time $q^{13n + o(n)}$ and space $\poly(n)\cdot q^{16n/q^2}$.
    By repeating this algorithm $\poly(n)$ times, we can therefore build a
    $\DGS^{\poly(n)q^{16n}/q^2}_{\eta_\varepsilon}$ oracle such that
    each call takes time $q^{13n + o(n)}$ and space $\poly(n)\cdot q^{16n/q^2}$.

    By Theorem~\ref{theorem_BDDtoDGS_new} we can construct a $\alpha-$BDD
    such that each call takes time $m\cdot\poly(n)$ and space $\poly(n)$, where
    $\alpha=\phi(\cL)/\lambda_1(\cL)=\frac{\sqrt{\ln(1/\epsilon)/\pi-o(1)}}{ 2\eta_\epsilon(\cL^*)\lambda_1(\cL)}$
    and $m=O(\frac{n\log(1/\eps)}{\sqrt{\eps}})=O(\frac{n^2}{q^2}q^{16n/q^2}\log q)\leqslant\poly(n)q^{16n/q^2}$.
    The preprocessing consists of $\poly(n)$ calls to the $\DGS^m_{\eta_\varepsilon}$ sampler described above
    and requires space $m\cdot\poly(n)$. Hence the total complexity is $\poly(n)q^{13n + o(n)}=q^{13n + o(n)}$
    in time and $m\cdot\poly(n)=\poly(n)q^{16n/q^2}$ in space.
    By using inequality~\eqref{ineq:smoothingtosvlarge},
    in Lemma~\ref{lemma_smoothinglambda}, we have that
    \[
        \lambda_1(\cL) \eta_\epsilon(\cL^*) < \sqrt{\frac{\beta(\cL)^2n}{2\pi e}}\cdot \epsilon^{-1/n}(1+o(1)).
    \]
    Hence we can guarantee that
    \[
        \alpha(\cL)
            \geq \sqrt{\frac{\ln(1/\epsilon)}{2n (\beta^2/e) \eps^{-2/n}}}\cdot (1-o(1))
            \geq \frac{1}{q}\sqrt{\frac{32\cdot e\cdot  \ln q}{2\beta^2 q^{64/q^2}}}\cdot(1-o(1))
            \geq (10q)^{-1}.
    \]
\end{proof}

\begin{theorem}\label{theorem_SVP_TM_tradeoff}
Let $n \in \mathbb{N}, q \in [4,\sqrt{n}]$ be a positive integer. Let $\cL$ be a lattice of rank $n$. There is a randomized algorithm that solves $\svp$ in time $ q^{13n+o(n)}$ and in space $\poly(n)\cdot q^{\frac{16n}{q^2}}$.
\end{theorem}

\begin{proof}
By Theorem \ref{thm:BDD}, we can construct a $\frac{0.1}{q}$-$\BDD$ oracle in time $ q^{13n+o(n)}$ and in space $\poly(n)\cdot q^{\frac{16n}{q^2}}$. 
Each execution of the BBD oracle now takes $m=\mathcal{O}(\frac{n^2}{q^2}q^{16n/q^2})$ time. By Theorem \ref{theorem_enlargeBDD}, with $(10q)^n$ queries to $\frac{0.1}{q}$-$\BDD$ oracle, we can find the shortest vector. The total time complexity is $q^{13n+o(n)}+ \frac{n^2}{q^2}q^{16n/q^2}\cdot (10q)^n= q^{13n+o(n)}$.
\qed
\end{proof}

\begin{remark}
    If we take $q=\sqrt{n}$, Theorem \ref{theorem_SVP_TM_tradeoff} gives a $\svp$ algorithm that takes $n^{\mathcal{O}(n)}$ time and $\poly(n)$ space. The constant in the exponent of time complexity is worse than the best enumeration algorithms. When $q$ is a large enough constant, for any constant $\epsilon > 0$, there exists a constant $C = C(\epsilon) > 2$, such that there is a $2^{Cn}$ time and $2^{\epsilon n}$ space algorithm for $\DGS$, and $\svp$. In particular, the time complexity of the algorithm in this regime is worse than the best sieving algorithms. 

\end{remark}

\begin{remark}
    In \cite{ACKLS21}, authors gave $2^{\eps n}$ time reductions between the lattice problems in different $\ell_p$ norms. Their reduction increases the approximation factor by a constant that depends on $\epsilon$. Combining their results with the above theorem gives an algorithm with a time-memory tradeoff for the constant-factor approximation of all lattice problems.
\end{remark}

\section{Quantum speedup for BDDP using QRAM}\label{section:BDDP-in-QRAM}

The goal of this section is to obtain a quantum speedup of Theorem~\ref{theorem_BDDtoDGS}. We improve the time complexity of the algorithm by almost the square root factor, but we also require the advice string to be stored in QRAM. Before presenting our contribution, we will give an overview of the known algorithms for $\BDDP$. 

Most of the known $\BDDP$ algorithms (including the one in \cite{DRS14}) are based on the algorithm for decision-$\CVPP$ by Aharonov and Regev \cite{AR05}. We first revisit the algorithm for decision-$\CVPP$.  In this work, the authors introduced the \emph{periodic Gaussian function} $f:\cR^{n}\rightarrow\cR^{+}$,
\[
f(\vect{t}):=\frac{\rho(\vect{t}+\cL)}{\rho(\cL)}\]
towards giving an algorithm for $\mathcal{O}(\sqrt{n/\log n})$ approximation of decision-$\CVPP$. They observed that by the Poisson summation formula, we get the identity
\begin{equation}\label{eq:f}f(\vect{t})= \displaystyle \mathop{\mathbb{E}}_{\vect{w}\sim {\mathcal D}_{{\mathcal L}^*}}[\cos (2\pi\langle \vect{w},\vect{t}\rangle)].\end{equation}

 They also showed that when distance of the target vector $\vect{t}$ from lattice $\cL$ is at least $\sqrt{n}$, then $f(\vect{t})$ is negligible, and when distance between $\vect{t}$ and lattice $\cL$ is at most $\sqrt{\log n}$, then $f(\vect{t})$ is non-negligible. This function $f$ evaluated on any vector $\vect{t}$ is an infinite sum, and is not easy to evaluate efficiently. Their algorithm crucially relied on the observation in Equation~\ref{eq:f} that shows that the function $f$ can be estimated by using a polynomial size advice string with at most $1/\poly(n)$ error. They gave the estimator 
\begin{equation}\label{eq:defn-f_W}
f_{W}(\vect{t})=\frac{1}{N}\sum\limits_{i=1}^N \cos(2\pi\langle \vect{w}_i,\vect{t}\rangle)\end{equation}  where $W=(\vect{w}_1,\cdots, \vect{w}_N)\in \cL^*$ are i.i.d. samples from $\cD_{\cL^*}$; and showed that $f_W\approx f$ with at most $1/\poly(n)$ error when $N$ is a large enough number bounded by a polynomial in $n$.  

Later, Liu, Lyubashevsky, and Micciancio \cite{LLM06} gave an algorithm for the approximation of search $\BDDP$. The idea is to iteratively update the target vector $\vect{t}$ such that its distance from the closest lattice vector decreases and eventually, it is easy to efficiently find the closest lattice vector. They are able to solve the $\alpha$-$\BDDP$ for $\alpha\leq\mathcal{O}\left(\sqrt{\frac{\log n}{n}}\right)$. Dadush, Regev and Stephens-Davidowitz, gave an improvement by a careful analysis of the function $f(\vect{t})$. They proposed that by iteratively updating $\vect{t}$ by an approximation of
\[\vect{t}+ \frac{\nabla f(\vect{t})}{2\pi f(\vect{t})},\]
we can go near the closest lattice vector. Their algorithm solves $\alpha$-$\BDDP$, for $\alpha = \frac{\sqrt{\ln(1/\epsilon)/\pi-o(1)}}{ 2\eta_\epsilon(\cL^*) \lambda_1(\cL)}$. The advice string consists of $N$ vectors from $\cD_{\cL^*,\eta_\eps(\cL^*)}$ and the algorithm performs $\mathcal{O}(N+ \poly(n))$ arithmetic operations where $N=\mathcal{O}\left(\frac{n\log(1/\eps)}{\sqrt{\eps}}\right)$. In this section, we will show that if the advice string is stored in QRAM then we can acheive the same approxmation of $\BDDP$ by using only $\mathcal{O}(\sqrt{N}+\poly(n))$  arithmetic operations.

We will start with listing some of the lemmas and  theorems from \cite{DRS14} that we will directly use in our proof. After that, we will show the quantum improvement in the estimation of function $f_W$ and $\nabla f_W$. In the last part of this section, we will present the main result of this section.

\subsection{Results from \cite{DRS14}}

\begin{lemma}\cite{DRS14}[Lemma~2.14]\label{lemma:DRS2.14}
Let $\cL \subset \mathbb{R}^n$
be a lattice of rank $n$. Then, for all $\vect{t} \in \mathbb{R}^n$
, $f(\vect{t}) \geq \rho(\vect{t})$.
\end{lemma}

\begin{lemma}\cite{DRS14}[Lemma~4.7]\label{lemma:DRS4.7}
Let $\cL \subset \mathbb{R}^n$ be a lattice with $\rho(\cL) = 1 + \epsilon$ for $\epsilon \in (0, 1/400)$. Let $W = (\vect{w}_1,\dots, \vect{w}_N)$ be sampled independently from $\cD_{\cL^*}$ with $N \geq \Omega(n/\sqrt{\epsilon})$. Then,
\[
\Pr[\exists \vect{t}, \|\vect{t}\|\leq \epsilon^{1/8}/(1000n): \|\nabla f_W(\vect{t})/(2\pi f_W(\vect{t}))+\vect{t}\|> \epsilon^{0.25}\|\vect{t}\|] \leq 2^{-\Omega(n)}.
\]
\end{lemma}

\begin{lemma}\cite{DRS14}[Lemma~4.10]\label{lemma:DRS4.10}
Let $\cL \subset \mathbb{R}^n$ be a lattice with $\rho(\cL) = 1+ \epsilon$ with $\epsilon \in (0, 1/400)$. Let $s_\epsilon =\left(\frac{1}{\pi}\ln\frac{2(1+\epsilon)}{\epsilon}\right)^{0.5}$.
Let $W = (\vect{w}_1,\dots, \vect{w}_N)$ be sampled independently from $\cD_{\cL^*}$. Then, for $\epsilon^2 \leq s \leq 10$, if $N \geq \Omega(n \ln(1/\epsilon)/s^2)$,
\[
\Pr[\exists \vect{t} \in \mathbb{R}^n,\epsilon^{1/8}/(1000n)\leq \|\vect{t}\|\leq s_\epsilon:\|\nabla f_W(\vect{t})-\nabla f(\vect{t})\|> s\|\vect{t}\|]\leq 2^{-\Omega(N\cdot s^2)}.
\]
\end{lemma}

\begin{lemma}\cite{DRS14}[Lemma~4.12]\label{lemma:DRS4.12}
Let $\cL \subset \mathbb{R}^n$ be a lattice with $\rho(\cL) = 1+ \epsilon$ with $\epsilon \in (0, 1/400)$. Let $s_\epsilon =(\frac{1}{\pi}\ln\frac{2(1+\epsilon)}{\epsilon})^{0.5}$.
Let $W = (\vect{w}_1,\dots,\vect{w}_N)$ be sampled independently from $\cD_{\cL^*}$. Then, for $\epsilon^2 \leq s \leq 10$, if $N \geq \Omega(n \ln(1/\epsilon)/s^2)$, then
\[
\Pr\left[\exists \vect{t} \in \mathbb{R}^n, \|\vect{t}\|\leq s_\epsilon:\| f_W(\vect{t})- f(\vect{t})\|> s\right]\leq 2^{-\Omega(N\cdot s^2)}.
\]
\end{lemma}

\begin{lemma}\cite{DRS14}[Part 1 of Lemma~4.5]\label{lemma:DRS4.5}
Let $\cL \subset \mathbb{R}^n$ be a lattice with $\rho(\cL) = 1+ \epsilon$ for some $\epsilon>0$, and let $W = (\vect{w}_1,\dots, \vect{w}_N)$ be sampled independently from $\cD_{\cL^*}$. Then, for $s \geq 0$, $N \min(s,s^2) \geq \Omega(n)$, and $\Delta_\epsilon=\frac{4\pi\epsilon}{1+\epsilon}(\ln\frac{2+2\epsilon}{\epsilon}+1)$, we have
\[Pr[\| H f_W(\vect{0}) + 2\pi \mat{I}_n\| > \Delta_\epsilon+s] \leq 2^{-\Omega(N\min{s,s^2})},\]

where $Hf_W$ denote the Hessian matrix of $f_W$.
\end{lemma}

\begin{lemma}\cite{DRS14}[Part 1 of Lemma~4.7]\label{cor:DRS4.7}
Let $\cL \subset \mathbb{R}^n$ be a lattice with $\rho(\cL) = 1+ \epsilon$ for some $\epsilon>0$, and let $W = (\vect{w}_1,\dots, \vect{w}_N)$ be sampled independently from $\cD_{\cL^*}$. Then we have
\[
Pr\left[\exists \vect{t} \in \mathbb{R}^n, \|\vect{t}\|\leq \frac{\epsilon^{1/8}}{1000n} : \|\nabla f_W(\vect{t})\| > (2\pi + 4\epsilon^{0.25})\|\vect{t}\|\right] \leq 2^{-\Omega(n)},
\]
and
\[
Pr\left[\exists \vect{t} \in \mathbb{R}^n, \|\vect{t}\|\leq \frac{\epsilon^{1/8}}{1000n} : | f_W(\vect{t})| < 1-\frac{\epsilon^{0.25}}{100}\right] \leq 2^{-\Omega(n)}.
\]
\end{lemma}

The following theorem tells us that if we can compute $f(\vect{t})$ and its gradient, and suppose that $t$ is sufficiently close to the lattice, then we can simply apply the gradient ascent algorithm to find a vector $\vect{t}'=\frac{\nabla f(\vect{t})}{2\pi f(\vect{t})}+ \vect{t}$ such that the new vector $\vect{t}'$ is even closer to the closest lattice vector of $\vect{t}$.

\begin{theorem}\cite{DRS14}[Corollary~4.3]\label{thm:DRS}
Let $\epsilon \in (0, 1/400)$ and $\cL \subset \mathbb{R}^n$ a lattice with $\rho(\cL) = 1 +\epsilon$. Let $s_\epsilon = \left(\frac{1}{\pi}\ln \frac{2(1+\epsilon)}{\epsilon}\right)^{1/2}$. Then for all $\vect{t} \in \mathbb{R}^n$ satisfying $\|\vect{t}\|\leq s_\epsilon/2$,
\[
\left\|\frac{\nabla f(\vect{t})}{2\pi f(\vect{t})}+ \vect{t}\right\| \leq 12 (\epsilon/2)^{1-2\delta{(\vect{t})}}\|\vect{t}\|,
\]
where $\delta(t)=\max(1/8, \|\vect{t}\|/s_\epsilon)$. In particular, for $\delta(\vect{t}) \leq 1/2-2/(\pi s_\epsilon^2)$,
\[
\left\|\frac{\nabla f(\vect{t})}{2\pi f(\vect{t})}+ \vect{t}\right\| \leq \|\vect{t}\|/4.
\]
\end{theorem}

\subsection{Estimation of $f_W$ and $\nabla f_W$ in QRAM}

\begin{theorem}[Amplitude estimation \cite{BHMT02}, follows from Theorem 12]\label{thm:amplitude_estimation}
Let $\delta>0$. Given natural number $M$ and access to an $(n + 1)$-qubit unitary $U$ satisfying
\[
U\ket{0^n}\ket{0}= \sqrt{a}\ket{\phi_1}\ket{1} +\sqrt{1-a}\ket{\phi_0}\ket{0},
\]
where $\ket{\phi_1}$ and $\ket{\phi_0}$ are arbitrary $n$-qubit states and $0 < a < 1$,
there exists a quantum algorithm that uses $\mathcal{O}(M\cdot \log(1/\delta))$ applications of $U$ and $U^\dagger$ and $\tilde{\mathcal{O}}(M\cdot \log(1/\delta))$ time, outputs an estimate $\tilde{a}$ that with probability $\geq 1-\delta$ satisfies~\footnote{The original Theorem~12 in~\cite{BHMT02} considered an estimator that output $\tilde{a}$ with additive error at most $1/M$ with probability $\geq2/3$. One can repeat this process for $\mathcal{O}(\log(1/\delta))$ times and take the median number to increase the success probability to $\geq 1-\delta$.}
\[
|a-\tilde{a}| \leq  \frac{1}{M}.
\]
\end{theorem}

\begin{theorem}\label{thm:Quantum_estimation_f_W}
Let $N$ be a positive integer, lattice $\cL\subset \real^n$ and $W=\{\vect{w}_1,\dots,\vect{w}_N\}$ be the set of vectors from $\cL^*$. Let $O_W: \ket{j}\ket{0}\rightarrow \ket{j}\ket{\vect{w}_j}$. For any $\epsilon,\delta>0$, there exists a quantum algorithm that given target $\vect{t}\in \mathbb{R}^n$ outputs $\tilde{f}_W(\vect{t})$ which satisfy $|\tilde{f}_{W}(\vect{t})-f_W(\vect{t})|\leq \epsilon$ with probability $1-\delta$. The algorithm makes $\mathcal{O}(\epsilon^{-1}\cdot\log\frac{1}{\delta})$ queries to $O_W$ and requires $\epsilon^{-1}\cdot\log\frac{1}{\delta}\cdot \poly(\log n)$ elementary quantum gates. 
\end{theorem}
\begin{proof}
We define the positive controlled rotation oracle as, for any $a\in [-1,1]$
\[
O_{CR^+}: \ket{a} \ket{0} \rightarrow \begin{dcases}
    \ket{a}(\sqrt{a}\ket{1}+\sqrt{1-a}\ket{0}),& \text{if } a\geq 0\\
    \ket{a}\ket{0},              & \text{otherwise,}
\end{dcases}
\]
which can be implemented up to negligible error by $\poly(\log n)$ quantum elementary gates. Also, we define the cosine inner product oracle as for any $\vect{t},\vect{w}\in \mathbb{R}^n$
\[
O_{\cos}: \ket{\vect{w}}\ket{\vect{t}}\ket{0}\rightarrow \ket{\vect{w}}\ket{\vect{t}}\ket{\cos(2\pi \langle \vect{w},\vect{t}\rangle)},
\]
which can also be implemented by $\poly(\log n)$ quantum elementary gates.
Prepare the state $\frac{1}{\sqrt{N}}\sum\limits_{j=1}^N\ket{j}\ket{\vect{0}}\ket{\vect{t}}\ket{0}\ket{0}$, and then we apply $O_W$ on the first and second registers, apply $O_{\cos}$ on the second, third, fourth registers, and apply $O_{CR^+}$ on the fourth and fifth registers, we have
\begin{align*} &\frac{1}{\sqrt{N}}\sum\limits_{\substack{j\in[N] \text{ and }\\\cos(2\pi\langle \vect{w}_j,\vect{t}\rangle)>0}} \ket{j}\ket{\vect{w}_j}\ket{\vect{t}}\ket{\cos(2\pi\langle \vect{w}_j,\vect{t}\rangle)}\big(\sqrt{\cos(2\pi\langle \vect{w}_j,\vect{t}\rangle}\ket{1}+\sqrt{1-\cos(2\pi\langle \vect{w}_j,\vect{t}\rangle}\ket{0}\big)\\
& +\frac{1}{\sqrt{N}}\sum\limits_{\substack{j\in[N] \text{ and }\\\cos(2\pi\langle \vect{w}_j,\vect{t}\rangle)\leq0}} \ket{j}\ket{\vect{w}_j}\ket{\vect{t}}\ket{\cos(2\pi\langle \vect{w}_j,\vect{t}\rangle}\ket{0}.
\end{align*}
By rearranging the equation, we have
\begin{align*}
&\frac{1}{\sqrt{N}}\sum\limits_{\substack{j\in[N]\text{ and }\\\cos(2\pi \langle \vect{w}_j,\vect{t}\rangle)>0}} \sqrt{\cos(2\pi\langle \vect{w}_j,\vect{t}\rangle}\ket{j}\ket{\vect{w}_j}\ket{\vect{t}}\ket{\cos(2\pi\langle \vect{w}_j,\vect{t}\rangle)}\ket{1}\\
+&\frac{1}{\sqrt{N}}\big(\sum\limits_{\substack{j\in[N]\text{ and }\\\cos(2\pi \langle \vect{w}_j,\vect{t}\rangle)>0}}\sqrt{1-\cos(2\pi\langle \vect{w}_j,\vect{t}\rangle)}\ket{j}\ket{\vect{w}_j}\ket{\vect{t}}\ket{\cos(2\pi\langle \vect{w}_j,\vect{t}\rangle)} +\sum\limits_{\substack{j\in[N]\text{ and }\\\cos(2\pi \langle \vect{w}_j,\vect{t}\rangle\leq0}} \ket{j}\ket{\vect{w}_j}\ket{\vect{t}}\ket{\cos(2\pi\langle \vect{w}_j,\vect{t}\rangle\big)}\big)\ket{0}\\
=&\sqrt{a^+}\ket{\phi_1}\ket{1}+\sqrt{1-a^+}\ket{\phi_0}\ket{0},
\end{align*}
where $a^+=\sum\limits_{\substack{j\in[N] \text{ and } \\\cos(2\pi \langle \vect{w}_j,\vect{t}\rangle)>0}} \frac{\cos\left(2\pi\langle \vect{w}_j,\vect{t}\rangle\right)}{N}$. By applying Theorem~\ref{thm:amplitude_estimation}, we can estimate $a^+$ with additive error $\epsilon/2$ by using  $\mathcal{O}(\epsilon^{-1})$ applications of $O_W$, $O^\dagger_W$, and $\epsilon^{-1}\cdot \poly(\log n)$ elementary quantum gates. Following the same strategy, we can also estimate $a^-=\sum\limits_{\substack{j\in[N] \text{ and } \\\cos(2\pi \langle w_j,t)<0}} \frac{\cos(2\pi\langle w_j,t\rangle)}{N}$ with same additive error and by using same amount of queries and quantum elementary gates. Therefore, we can estimate $a^++a^-=\sum\limits_{j\in [N]} \frac{\cos(2\pi\langle w_j,t\rangle)}{N}$ with additive error $\epsilon$. By repeating the procedure $\Theta(\log \frac{1}{\delta})$ times and take the median among them, we finish the proof.
\end{proof}

\begin{theorem}\label{Cor:Quantum_estimation_nabla_f_W}
Let $N$ be a positive integer, lattice $\cL \subset \real^n$, and $W=\{\vect{w}_1,\dots,\vect{w}_N\}$ be the set of vectors in $\cL^*$. Let $O_W: \ket{j}\ket{0}\rightarrow \ket{j}\ket{\vect{w}_j}$. For any $\epsilon,\delta>0$, there exist a quantum algorithm that given target $\vect{t}\in \mathbb{R}^n$ and $w_{max}\geq \max\limits_{j\in [N]} \|\vect{w}_j\|$, output $\nabla_i \tilde{f}_W(\vect{t})$\footnote{ We are abusing the notation here, $\nabla_i \tilde{f}_W(\vect{t})$ just represent an estimated value of $\nabla_i {f}_W(\vect{t})$.} for any $i\in [n]$, which satisfy  $|\nabla_{i}\tilde{f}_W(\vect{t})-\nabla_if_W(\vect{t})|\leq  2\pi\epsilon\|\vect{t}\|\cdot w^2_{max}$, with probability $1-\delta$. The algorithm makes $O(\epsilon^{-1}\cdot \log \frac{1}{\delta})$ queries to $O_W$ and requires $\epsilon^{-1}\cdot \log \frac{1}{\delta}\cdot \poly(\log n)$ elementary quantum gates. 
\end{theorem}
\begin{proof}
From Equation~\ref{eq:defn-f_W}, we get $\nabla_i f_W(\vect{t})=\frac{-1}{N}\sum\limits_{j=1}^N \sin (2\pi\langle \vect{w}_j,\vect{t} \rangle)\cdot (w_j)_i$ for $i\in [n]$.
We observe that for all $i\in[n]$ and $j\in [N]$,
\[
|\sin(2\pi\langle \vect{w}_j,\vect{t}\rangle) \cdot(w_j)_i| \leq |2\pi\langle \vect{w}_j,\vect{t}\rangle \cdot w_{max}| \leq 2\pi \|\vect{t}\|\cdot w^2_{max},
\]
and hence $ \left|\frac{\sin(2\pi\langle \vect{w}_j,\vect{t}\rangle) \cdot(w_j)_i}{2\pi \|\vect{t}\|\cdot w^2_{max}}\right|\leq 1$ for all $i\in[n]$ and $j\in [N]$. For simplicity, we define $g_i(\vect{w},\vect{t})=\frac{\sin(2\pi\langle \vect{w},\vect{t}\rangle) \cdot(w)_i}{2\pi \|\vect{t}\|\cdot w^2_{max}}$.
Then we define the sine inner product oracle as for any $\vect{t},\vect{w}\in \mathbb{R}^n$ and $i\in [n]$
\[
O_{\sin}: \ket{\vect{w}}\ket{\vect{t}}\ket{i}\ket{0}\rightarrow \ket{\vect{w}}\ket{\vect{t}}\ket{i}\ket{g_i(\vect{w},\vect{t})},
\]
which can be implemented by $\poly(\log n)$ quantum elementary gates. Also, we define the positive controlled rotation oracle as, for any $a\in\mathbb{R}$
\[
O_{CR^+}: \ket{a} \ket{0} \rightarrow \begin{dcases}
    \ket{a}(\sqrt{a}\ket{1}+\sqrt{1-a}\ket{0}),& \text{if } a\geq 0\\
    \ket{a}\ket{0},              & \text{otherwise,}
\end{dcases}
\]
which can be implemented by $\poly(\log n)$ quantum elementary gates.

Prepare the state $\frac{1}{\sqrt{N}}\sum\limits_{j=1}^N\ket{j}\ket{\vect{0}}\ket{\vect{t}}\ket{i}\ket{0}\ket{0}$, and then we apply $O_W$ on the first and second registers, apply $O_{\sin}$ on the second, third, fourth, fifth registers, and apply $O_{CR^+}$ on the fifth and sixth registers, we have
\begin{align*}
 &\frac{1}{\sqrt{N}}\sum\limits_{\substack{j\in[N] \text{ and }\\g_i (w_j,t)>0}} \ket{j}\ket{\vect{w}_j}\ket{\vect{t}}\ket{i}\ket{g_i(\vect{w}_j,\vect{t})}\big(\sqrt{g_i(\vect{w}_j,\vect{t})}\ket{1}+\sqrt{1-g_i(\vect{w}_j,\vect{t})}\ket{0}\big)\\
 &+\frac{1}{\sqrt{N}}\sum\limits_{\substack{j\in[N] \text{ and }\\g_i(\vect{w}_j,\vect{t})\leq0}} \ket{j}\ket{\vect{w}_j}\ket{\vect{t}}\ket{i}\ket{g_i(\vect{w}_j,\vect{t})}\ket{0}\\
 &=\frac{1}{\sqrt{N}}\sum\limits_{\substack{j\in[N]\text{ and }\\g_i(\vect{w}_j,\vect{t})>0}} \sqrt{g_i(\vect{w}_j,\vect{t})}\ket{j}\ket{\vect{w}_j}\ket{\vect{t}}\ket{i}\ket{g_i(\vect{w}_j,\vect{t})}\ket{1}\\
&+\frac{1}{\sqrt{N}}\big(\sum\limits_{\substack{j\in[N]\text{ and }\\g_i(\vect{w}_j,\vect{t})>0}}\sqrt{1-g_i(\vect{w}_j,\vect{t})}\ket{j}\ket{\vect{w}_j}\ket{\vect{t}}\ket{i}\ket{g_i(\vect{w}_j,\vect{t})}+\sum\limits_{\substack{j\in[N]\text{ and }\\g_i(\vect{w}_j,\vect{t})\leq0}} \ket{j}\ket{\vect{w}_j}\ket{\vect{t}}\ket{i}\ket{g_i(\vect{w}_j,\vect{t})}\big)\ket{0}\\
&=\sqrt{a^+}\ket{\phi_1}\ket{1}+\sqrt{1-a^+}\ket{\phi_0}\ket{0},
\end{align*}
where $a^+=\sum\limits_{\substack{j\in[N] \text{ and } \\g_i(\vect{w}_j,\vect{t})>0}} \frac{g_i(\vect{w}_j,\vect{t})}{N}$. By applying Theorem~\ref{thm:amplitude_estimation}, we can estimate $a^+$ with additive error $\epsilon/2$ by using  $\mathcal{O}(\epsilon^{-1})$ applications of $O_W$, $O^\dagger_W$, and $\epsilon^{-1}\cdot \poly(\log n)$ elementary quantum gates. Follow the same strategy, we can also estimate $a^-=\sum\limits_{\substack{j\in[N] \text{ and } \\g_i(\vect{w}_j,\vect{t})<0}} \frac{g_i(\vect{w}_j,\vect{t})}{N}$ with same additive error and by using same amount of queries and quantum elementary gates. Therefore, we can estimate $a^++a^-=\sum\limits_{j\in [N]} \frac{g_i(\vect{w}_j,\vect{t})}{N}=\frac{-\nabla_i f_W(\vect{t})}{2\pi \|\vect{t}\|\cdot w^2_{max}}$ with additive error $\epsilon$. By repeating the procedure $\Theta(\log \frac{1}{\delta})$ times and taking the median among them, we finish the proof.
\end{proof}

We get the following corollary. 
\begin{corollary}\label{cor:quantum_estimate_f_w}
Let $\eps\in (0,/100)$, and $N=\lceil(n^8/\sqrt{\eps})\rceil$ be a positive integer, lattice $\cL\subset \real^n$ and $W=\{\vect{w}_1,\dots,\vect{w}_N\}$ be the set of vectors from $\cL^*$ such that $\forall i\in [N], \|\vect{w}_i\|\leq 2n+1$. Let $O_W: \ket{j}\ket{0}\rightarrow \ket{j}\ket{\vect{w}_j}$. There exists a quantum algorithm that given target $\vect{t}\in \mathbb{R}^n$ outputs $\tilde{f}_W(\vect{t})$ and $\nabla\tilde{f}_W(\vect{t})$ which satisfy \[|\tilde{f}_{W}(\vect{t})-f_W(\vect{t})|\leq \frac{\eps^{1/4}}{100} \text{ and } \|\nabla\tilde{f}_{W}(\vect{t})-\nabla f_W(\vect{t})\|\leq \frac{\eps^{1/4}}{100}\|\vect{t}\|\] with probability $1-2^{-\Omega(n)}$. The algorithm make $\mathcal{O}(\epsilon^{-1/4}\cdot n^4)$ queries to $O_W$ and require $\epsilon^{-1/4}\cdot n^4\cdot \poly(\log n)$ elementary quantum gates. 
\end{corollary}

\subsection{BDDP in QRAM}
From previous subsection, given $\vect{t}$, we can estimate $f_W(\vect{t})$ with small additive error. In this subsection, we will replace $f_W(\vect{t})$ with the approximation function $\tilde{f}_W(\vect{t})$, and show that doing gradient ascent on the approximation function $\tilde{f}_W(\vect{t})$ still helps us to find the closest vector, and hence we can use $\tilde{f}_W$ to solve $\BDDP$.

\begin{theorem}\label{thm:gradient_difference}
Let $\cL \subset \mathbb{R}^n$ be a lattice with $\rho(\cL) = 1+ \epsilon$ with $\epsilon \in (0, 1/400)$ and $N$ be a positive integer. Let $s_\epsilon =\left(\frac{1}{\pi}\ln\frac{2(1+\epsilon)}{\epsilon}\right)^{0.5}$, $\delta_{\max}=0.5-\frac{2}{\pi s_\epsilon^2}$, and $W = (\vect{w}_1,\dots, \vect{w}_N)$ be a set of vectors from $\cL^*$. Suppose that for some $\gamma>0$, $\vect{t}\in \mathbb{R}^n$, we can compute $\tilde{f}_{W}(\vect{t})$ and $\nabla \tilde{f}_{W}(\vect{t})$, and it holds that
\begin{enumerate}
    \item $\|\vect{t}\|\leq \min\{\delta_{\max} s_\epsilon, \sqrt{\ln(1/(4\gamma))/\pi}\},$
    \item $\|\nabla f_W(\vect{t})-\nabla f(\vect{t}) \|\leq \frac{\pi}{2}\gamma\|\vect{t}\|,$
    \item $\|\nabla f_W(\vect{t})-\nabla \tilde{f}_W(\vect{t}) \|\leq \frac{\pi}{2}\gamma\|\vect{t}\|,$
    \item $|f_W(\vect{t})-f(\vect{t})|\leq \gamma,$
    \item $|f_W(\vect{t})-\tilde{f}_W(\vect{t})|\leq \gamma.$
\end{enumerate}
Then ,
\[
\left\|\frac{\nabla \tilde{f}_W(\vect{t})}{2\pi \tilde{f}_W(\vect{t})}-\frac{\nabla f(\vect{t})}{2\pi f(\vect{t})}\right\|\leq \frac{6\gamma}{\rho(\vect{t})}\left\|\vect{t}\right\|.
\]
\end{theorem}
\begin{proof}
It's easy to see that above conditions imply 
\begin{equation}\label{eq:nabla_bound}
\|\nabla \tilde{f}_W(\vect{t})-\nabla {f}(\vect{t}) \|\leq {\pi}\gamma\|\vect{t}\|,\end{equation} and \begin{equation}\label{eq:difference_combine}
|\tilde{f}_W(\vect{t})-f(\vect{t})|\leq 2\gamma.\end{equation}
From Lemma~\ref{lemma:DRS2.14} and the condition on length of $\vect{t}$, we get $f(\vect{t})\geq \rho(\vect{t})\geq 4\gamma$. 
From triangle inequality, we get 
\begin{dmath}\label{eq:gradient}
\left\|\frac{\nabla\tilde{f}_W(\vect{t})}{2\pi \tilde{f}_W(\vect{t})}- \frac{\nabla f(\vect{t})}{2\pi f(\vect{t})} \right\|= \left\|\frac{\nabla\tilde{f}_W(\vect{t})-\nabla f(\vect{t})}{2\pi f(\vect{t})}\frac{f(\vect{t})}{\tilde{f}_W(\vect{t})}+\frac{\nabla f(\vect{t})}{2\pi f(\vect{t})}\left(\frac{f(\vect{t})}{\tilde{f}_W(\vect{t}-1)} \right) \right\|
\leq \left\|\frac{\nabla\tilde{f}_W(\vect{t})-\nabla f(\vect{t}) }{2\pi f(\vect{t})}  \right\|\frac{f(\vect{t})}{\tilde{f}_W(\vect{t})}+\left\| \frac{\nabla f(\vect{t})}{2\pi f(\vect{t})}\right\| \left| \frac{f(\vect{t})}{\tilde{f}_{W}(\vect{t})}-1 \right| \end{dmath}

For the first term in \ref{eq:gradient}, by using Equation~(\ref{eq:nabla_bound}) and Equation~(\ref{eq:difference_combine}), we get  
\begin{equation}\label{eq:first_part}
\left\|\frac{\nabla\tilde{f}_W(\vect{t})-\nabla f(\vect{t}) }{2\pi f(\vect{t})}  \right\|\frac{f(\vect{t})}{\tilde{f}_W(\vect{t})} \leq \frac{\pi\gamma\|\vect{t}\| }{2\pi f(\vect{t})}\frac{f(\vect{t})}{f(\vect{t})-2\gamma}= \frac{\gamma\|\vect{t}\|}{2(f(\vect{t})-2\gamma)}.
\end{equation}

For the second term in \ref{eq:gradient}, by using Theorem~\ref{thm:DRS} and Equation~(\ref{eq:difference_combine}), we get 
\begin{equation}\label{eq:second_part}
\left\| \frac{\nabla f(\vect{t})}{2\pi f(\vect{t})}\right\| \left| \frac{f(\vect{t})}{\tilde{f}_{W}(\vect{t})}-1 \right| \leq \frac{5}{4}\|\vect{t}\| \left( \frac{f(\vect{t})}{f(\vect{t})-2\gamma} -1 \right)= \frac{10\gamma}{4(f(\vect{t})-2\gamma)}\|\vect{t}\|.
\end{equation}

From Equation~(\ref{eq:gradient}), Equation~(\ref{eq:first_part}) and Equation~(\ref{eq:second_part}), we get 

\[    \left\|\frac{\nabla\tilde{f}_W(\vect{t})}{2\pi \tilde{f}_W(\vect{t})}- \frac{\nabla f(\vect{t})}{2\pi f(\vect{t})} \right\| \leq \frac{3\gamma \cdot \|\vect{t}\|}{f(\vect{t})-2\gamma} \leq \frac{6\gamma}{\rho(\vect{t})}\|\vect{t}\|.
\]
\end{proof}

\begin{lemma}\label{lemma:small_Works}
Let $\cL \subset \mathbb{R}^n$ be a lattice with $\rho(\cL) = 1 + \epsilon$ for $\epsilon \in (0, 1/400)$. Let $W = (\vect{w}_1,\dots, \vect{w}_N)$ be sampled independently from $D_{\cL^*}$ with $N \geq n^8/\sqrt{\epsilon}$. Suppose that for some $\vect{t}\in \mathbb{R}^n$ such that $\|\vect{t}\|\leq \epsilon^{1/8}/(1000n)$, we can compute $\tilde{f}_W(\vect{t})$ and $\nabla \tilde{f}_W(\vect{t})$ which satisfy
\begin{enumerate}
    \item $|\tilde{f}_W(\vect{t})-f_W(\vect{t})|\leq \frac{\epsilon^{1/4}}{100}$,
    \item $\|\nabla \tilde{f}_W(\vect{t})-\nabla f_W(\vect{t})\|\leq \frac{\epsilon^{1/4}}{100}\|\vect{t}\|$.
\end{enumerate} Then with probability $1- 2^{-\Omega(n)}$, 
$\left\|\frac{\nabla \tilde{f}_W(\vect{t})}{2\pi \tilde{f}_W(\vect{t})}+\vect{t}\right\|\leq 3\epsilon^{1/4}\|\vect{t}\|$ holds.
\end{lemma}
\begin{proof}

By Corollary~\ref{cor:DRS4.7}, with at least $1-2^{-\Omega(n)}$ probability, $\|\nabla f_W (\vect{t})\|\leq (2\pi+4\epsilon^{1/4})\|\vect{t}\|$ and $|f_W(\vect{t})|>1-\frac{\epsilon^{1/4}}{100}$. By triangle inequality and both the assumptions, with probability at least $1-2^{-\Omega(n)}$, we have  \begin{equation}
    \|\nabla \tilde{f}_W(\vect{t})\|\leq (2\pi+5\epsilon^{1/4})\|\vect{t}\| \text{ and } |\tilde{f}_W(\vect{t})|>1-\frac{\epsilon^{1/4}}{50}.\end{equation}
    
    It implies that,
\begin{align*}
    \left\|\frac{\nabla \tilde{f}_W(\vect{t})}{2\pi \tilde{f}_W(\vect{t})}-\frac{\nabla f_W(\vect{t})}{2\pi f_W(\vect{t})}\right\|&\leq \left\|\frac{\nabla \tilde{f}_W(\vect{t})-\nabla f_W(\vect{t})}{2\pi \tilde{f}_W(\vect{t})}\right\|+\frac{\|\nabla f_W(\vect{t})\|}{2\pi}\left|\frac{1}{\tilde{f}_W(\vect{t})}-\frac{1}{f_W(\vect{t})}\right|\\
    & \leq \frac{\epsilon^{1/4}\|\vect{t}\|}{2\pi(1-\epsilon^{1/4}/50)}+\frac{(2\pi+5\epsilon^{1/4})\|\vect{t}\|}{2\pi}\cdot \frac{\epsilon^{1/4}/100}{(1-\epsilon^{1/4}/100)(1-\epsilon^{1/4}/50)}\\
    & \leq 2\epsilon^{1/4}\|\vect{t}\|,
\end{align*}
By Lemma~\ref{lemma:DRS4.7}, we know that, with at least $1-2^{-\Omega(n)}$ probability we have $\left\|\frac{\nabla {f}_W(\vect{t})}{2\pi {f}_W(\vect{t})}+\vect{t}\right\|\leq \epsilon^{1/4}\|\vect{t}\|$. Hence, by triangle inequality and union bound, we have $\left\|\frac{\nabla \tilde{f}_W(\vect{t})}{2\pi \tilde{f}_W(\vect{t})}+\vect{t}\right\|\leq 3\epsilon^{1/4}\|\vect{t}\|$ with probability greater than $1-2^{-\Omega(n)}$.

\end{proof}

\begin{theorem}\label{th:theorem_BDDtoDGS_quantum}
Let $\cL \subset \mathbb{R}^n$ be a lattice with $\rho(\cL) = 1+ \epsilon$ with $\epsilon \in (0, 1/400)$. Let $s_\epsilon =(\frac{1}{\pi}\ln\frac{2(1+\epsilon)}{\epsilon})^{0.5}$, $\delta_{\max}=0.5-\frac{2}{\pi s_\epsilon^2}$, $\delta(\vect{t})=\max\{1/8, \|\vect{t}\|/s_\epsilon\}$ and $W = (\vect{w}_1,\dots, \vect{w}_N)$ be sampled independently from  $\cD_{\cL^*}$ where $N \geq n^8 \ln(1/\epsilon)/\sqrt{\epsilon}$. Suppose that for some $\vect{t}\in \mathbb{R}^n$ such that $\|\vect{t}\|\leq \delta_{\max} s_\epsilon$, we can compute $\tilde{f}_W(\vect{t})$ and $\nabla \tilde{f}_W(\vect{t})$ which satisfy 
\begin{enumerate}
    \item $|\tilde{f}_W(\vect{t})-f_W(\vect{t})|\leq \frac{\epsilon^{1/4}}{100}$,
    \item $\|\nabla \tilde{f}_W(\vect{t})-\nabla f_W(\vect{t})\|\leq \frac{\|\vect{t}\|\cdot\epsilon^{1/4}}{100}$.
\end{enumerate}
Then with probability at least $1-2^{-\Omega(n)}$, $\left\|\frac{\nabla \tilde{f}_W(\vect{t})}{2\pi \tilde{f}_W(\vect{t})}+\vect{t}\right\| \leq \epsilon^{(1-2\delta(\vect{t}))/4}\|\vect{t}\|$.

\end{theorem}

\begin{proof}
Lemma~\ref{lemma:small_Works} shows that the theorem is satisfied for all $\vect{t}$ with $\|\vect{t}\| \leq \epsilon^{1/8}/(1000n)$. So we consider the case when $\epsilon^{1/8}/(1000n) < \|\vect{t}\| \leq \delta_{\max}s_\epsilon$. 
By Lemma~\ref{lemma:DRS2.14}, for such $\vect{t}$,
\[
f(\vect{t}) \geq \rho(\vect{t}) \geq e^{-\pi\delta^2_{\max} s^2_\epsilon} > \epsilon^{\delta^2_{\max}}/2 \geq \epsilon^{1/4}/2.
\]
Also we know that 
\begin{itemize}
    \item By Lemma~\ref{lemma:DRS4.10}, $\|\nabla {f}_W(\vect{t})-\nabla f(\vect{t})\|\leq \epsilon^{1/4}\|\vect{t}\|/100$ holds with probability $\geq 1-2^{-\Omega(\epsilon^{1/2}N/100^2)}=1-2^{-\Omega(n)}.$
    \item By Lemma~\ref{lemma:DRS4.12}, $\|{f}(\vect{t})- f(\vect{t})\|\leq \epsilon^{1/4}/100$ holds with probability $\geq 1-2^{-\Omega(\epsilon^{1/2}N/100^2)}=1-2^{-\Omega(n)}.$
\end{itemize}

    Therefore, by using Theorem~\ref{thm:gradient_difference} with $\gamma =\frac{\epsilon^{1/4}}{50\pi}$, we have that with probability $\geq 1-2^{-\Omega(n)}$,
    \begin{align*}
        \left\|\frac{\nabla \tilde{f}_W(\vect{t})}{2\pi \tilde{f}_W(\vect{t})}+\vect{t}\right\| &\leq \frac{6\gamma}{\rho(\vect{t})}\|\vect{t}\|+ \left\|\frac{\nabla f(\vect{t})}{2\pi f(\vect{t})}+\vect{t}\right\| \\
        & \leq \frac{\epsilon^{1/4}}{25}\cdot e^{\pi\|\vect{t}\|^2}\|\vect{t}\| + 12(\epsilon/2)^{1-2\delta(\vect{t})} \|\vect{t}\|  & & \text{(by Theorem~\ref{thm:DRS})}\\
        & \leq \frac{3\eps^{1/4}}{25}\left(\frac{2(1+\eps)}{\eps}\right)^{\delta(\vect{t})^2} \|\vect{t}\|  + 12(\epsilon/2)^{1-2\delta(\vect{t})} \|\vect{t}\| \\
        & \leq \frac{\epsilon^{0.25-\delta(\vect{t})^2}}{6}\cdot \|\vect{t}\| + \frac{4}{5}(\epsilon)^{(1-2\delta(\vect{t}))/4} \|\vect{t}\| \\
        & \leq \epsilon^{(1-2\delta(\vect{t}))/4}\|\vect{t}\|.
    \end{align*}
\end{proof}

\begin{theorem}\label{theorem_BDDtoDGS_quantum}
For any integer $n>0$, $\epsilon\in (e^{-n^2},1/400)$ and lattice $\cL\subset \cR^n$, let $\phi(\cL)=\frac{\sqrt{\ln(1/\eps)/\pi -o(1) }}{2\eta_\eps(\cL^*)}$ and $N=\frac{n^8 \ln(1/\eps)}{\sqrt{\eps}}$. Given $N$ vectors sampled from $\cD_{\cL^*,\eta_\eps(\cL^*)}$ stored in QRAM, there exists an algorithm that solves $\frac{\phi(\cL)}{\lambda_1(\cL)}$-$\BDD$ with probability greater than $1-2^{-\Omega(n)}$ using $\sqrt{N}\cdot \poly(n)$ arithmetic operations, require $\mathcal{O}(\poly(n)+\log N)$ classical space, $\poly(n)$ qubits and QRAM of size $N\cdot\poly(n)$.

\end{theorem}
\begin{proof}
Let $s_\epsilon=(\frac{1}{\pi}\log(\frac{2(1+\epsilon)}{\epsilon}))^{1/2}$ and $\delta_{max}=\frac{1}{2}-\frac{2}{\pi s_\epsilon^2}$. Let $W=(\vect{w}_1,\cdots,\vect{w}_N)$ be the given set of vectors sampled from $\cD_{\cL^*,\eta_{\epsilon}(\cL^*)}$.

The algorithm takes target $\vect{t}\in \cR^n$ and set of vectors $W$. It then iteratively updates $\vect{t}\leftarrow \vect{t}+\frac{\nabla\tilde{f}_W(\vect{t})}{2\pi\tilde{f}_W(\vect{t})}$ for $1+\lceil 8\log(\sqrt{n}s_\epsilon)/\log(1/\epsilon)\rceil$ times. It then scans the first $\sqrt{N}$ vectors from $W$ and take the first $n$ linearly independent vectors of length bounded by $\sqrt{n}\cdot \eta_\epsilon(\cL^*)$ as set $V^*=(\vect{v}_1^*,\cdots,\vect{v}_n^*) \subset W$ ,if no such set exists then abort. Compute $V=(\vect{v}_1,\cdots,\vect{v}_n)$ such that $\vect{v}_i\cdot \vect{v}_j^*=\delta_{i,j}$ and return $\sum c_i\vect{v}_i$ where $c_i=\lfloor \vect{v}_i^*\cdot \vect{t}\rceil$. 

By scaling the lattice appropriately, we can assume that $\rho(\cL)=1+\eps$, so that $\eta_{\eps}(\cL^*)=1$.  Let $\vect{t}'=\vect{t}-\vect{w}$ such that $\|\vect{t}'\|\leq \delta_{max}s_\eps$ for some $\vect{w}\in \cL$. There exists such a vector $\vect{t}'$ because of the promise of the $\phi(\cL)/\lambda_1(\cL)$-$\BDD$ and the fact that $\phi(\cL)\leq {\delta_{max}s_\eps}$. From Lemma~\ref{lemma:DGS-large-vector-bound}, we get that with probability $1-N\cdot e^{-2n^2}\geq 1-2^{\Omega(n)}$, $\forall i\in [N], \|\vect{w}_i\|\leq 2n+1$. Hence by using Corollary~\ref{cor:quantum_estimate_f_w}, we can estimate $\tilde{f}_W(\vect{t})$ and $\nabla\tilde{f}_W(\vect{t})$ which satisfy \[|\tilde{f}_{W}(\vect{t})-f_W(\vect{t})|\leq \frac{\eps^{1/4}}{100} \text{ and } \|\nabla\tilde{f}_{W}(\vect{t})-\nabla f_W(\vect{t})\|\leq \frac{\eps^{1/4}}{100}\|\vect{t}\|\] with probability $1-2^{-\Omega(n)}$. The algorithm make $\mathcal{O}(\sqrt{N})$ arithmetic operations. Corollary~\ref{cor:quantum_estimate_f_w} also says that the distribution of $\tilde{f}_W(\vect{t})$ and $\nabla\tilde{f}_W(\vect{t})$ is periodic over the lattice. Therefore, by using Theorem~\ref{th:theorem_BDDtoDGS_quantum} we can say that by every update of $\vect{t}$ by $\vect{t}\leftarrow \vect{t}+\frac{\nabla\tilde{f}_W(\vect{t})}{2\pi\tilde{f}_W(\vect{t})}$, the output vector is of form $\vect{w}+\vect{t}^*$ where $\|\vect{t}^*\|$ shrinks by factor of at least $\eps^{(1-2(1/4))/4}=\eps^{1/8}$ with probability $1-2^{-\Omega(n)}$. Hence by  $1+\lceil 8\log(\sqrt{n}s_\epsilon)/\log(1/\epsilon)\rceil$ updates, we get  $\vect{t}=\vect{w}+\vect{t}^*$ such that $\|\vect{t}^*\|\leq 1/(2\sqrt{n})$. For the correctness of the proof, it is sufficient to show that $\vect{w}=\sum_{i}c_i\vect{v}_i$. Note that $\lfloor \langle \vect{v}_i^*,\vect{t} \rangle \rceil=\lfloor \langle \vect{v}_i^*,\vect{t}^* \rangle \rceil+\langle \vect{v}_i^*,\vect{w}\rangle$. By Cauchy-Schwarz, we get $\lfloor \langle \vect{v}_i^*,\vect{t}^* \rangle \rceil=0$. Hence, we get the vector $\vect{w}$ as the output with probability greater than $1-2^{-\Omega(n)}$.

Now, we will show that there exist a set $V^*\subset\{\vect{w}_1,\cdots,\vect{{w}_{n^2}}\}$ containing $n$ linearly independent vectors of length at most $\sqrt{n}$ and algorithm will abort with probability at most $2^{-\Omega(n^2)}$ over the given set $W$. Let $W'=(\vect{w}_1,\cdots,\vect{w}_{n^2})$. By Lemma~\ref{lemma:DGS-large-vector-bound}, with at least $1-n^2\cdot e^{-n}$ probability all the vectors in set $W'$ has length at most $\sqrt{n}\cdot \eta_\eps(\cL^*)$. From the Lemma~\ref{lemma:DRS4.5}, we know that 
\[\|Hf_{W'}(\vect{0})+2\pi\mat{I}_n\| \leq \frac{4\pi\eps}{1+\eps}\left(\log\left(\frac{2(1+\eps)}{\eps}+1\right)+1 \right) <2\pi \]
with probability at least $1-2^{-\Omega(n^2)}$. Hence, we have $Hf_{W'}(\vect{0})$ is invertible and $Hf_{W'}(\vect{0})=\frac{-4\pi^2}{m}\sum_{i=1}^m \vect{w}_i\vect{w}_i^{T}$ (from Equation~\ref{eq:defn-f_W}). It implies that $W'$ spans $\real^n$ and completes the proof. 

\end{proof}

\section{Improved algorithms for BDD}\label{section-:algorithm-BDD}

We obtain a $\BDD$ oracle with decoding distance $\alpha$ by using the same reduction as above but
making each call cheaper. This is achieved by building a sampler that directly samples
at the smoothing parameter, hence avoiding the $\sqrt{2}$ factor, allowing us to take
a bigger $\varepsilon$.
In~\cite{ADRSD15}, it was shown how to construct a dense lattice
$\cL'$ whose smoothing parameter $\eta(\cL')$ is $\sqrt{2}$ times smaller than the original
lattice, and that contains all lattice points of the original lattice.
Suppose that we first use such a dense lattice to construct a corresponding
discrete Gaussian sampler with standard deviation equal to $s=\sqrt{2}\eta(\cL')$.
We then do the rejection sampling on condition that the output is in the original lattice
$\cL$. We thus have constructed a discrete Gaussian sampler of $\cL$ whose standard deviation
is $\sqrt{2}\eta(\cL')=\eta(\cL)$.  Nevertheless, $|\cL'/\cL|$ will be at least $2^{0.5n}$,
which implies that this procedure needs at least $2^{0.5n}$ input vectors to produce an
output vector.

The complexity of our BDD algorithms heavily depends on the quantity $\beta(\cL)$ which is related
to the kissing number of the lattice (see Section~\ref{section_preliminaries}). For this reason,
we first provide complexity bounds that depend on $\beta(\cL)$ and then obtain complexity bounds
in the worst case ($\beta(\cL)\leqslant 2^{0.402}$) as corollaries. We first show how to efficiently
sample a discrete Gaussian at the smoothing parameter.

\begin{lemma}\label{lem:dgs_at_smoothing}
    There is a probabilistic algorithm that, given a lattice $\cL \subset \real^n$, $m \in \intg_{+}$
    and $s\geq \eta_{1/3}(\cL)$ as input, outputs $m$ samples from a distribution
    $(m \cdot 2^{-\Omega(n^2)})$-close to $D_{\cL,s}$ in expected time $m\cdot 2^{n/2+o(n)}$ and space $(m+2^{n/2})\cdot 2^{o(n)}$.
    Furthermore, all samples have $\poly(n)$ bit-size.
\end{lemma}
\begin{proof}
Let $a=\frac{n}{2}+4$. We repeat the following until we output $m$ vectors. We use the algorithm
in Lemma \ref{lemma_mod2} to obtain a lattice $\cL' \supset \cL$ of index $2^a$.
We then run the algorithm from Theorem~\ref{theorem_dgsabovesmoothing} with input $(\cL', s)$
to obtain a list of vectors from $\cL'$. We output the vectors in this list that belong to $\cL$.
The correctness of the algorithm, assuming it outputs anything, is clear as long as the samples
obtained from Theorem~\ref{theorem_dgsabovesmoothing} are (sufficiently) independent, which we will
prove below.

By Theorem~\ref{theorem_dgsabovesmoothing}, we  obtain, in time and space $2^{(n/2)+o(n)}$,
$M\leqslant 2^{n/2}$ vectors that are $2^{-\Omega(n^2)}$-close to $M$ vectors independently sampled from
$D_{\cL',s}$. The theorem guarantees that $M=2^{n/2}$ if $s\geqslant \sqrt{2}\eta_{1/2}(\cL')$.
Also, by Lemma \ref{lemma_mod2}, with probability at least $1/2$, we have
$s\geq \eta_{1/3}(\cL)\geq \sqrt{2}\eta_{1/2}(\cL')$. Note that when $s<\sqrt{2}\eta_{1/2}(\cL')$,
the samples obtained from Theorem~\ref{theorem_dgsabovesmoothing} are still $2^{-\Omega(n^2)}$-close
to $M$ vectors independently sampled from $D_{\cL',s}$ but $M$ could be much lower than $2^{n/2}$ or even $0$.
On the other hand, if $s\geqslant\sqrt{2}\eta_{1/2}(\cL')$ then $M=2^{n/2}$.

Assume that $s\geqslant\sqrt{2}\eta_{1/2}(\cL')$, which happens with probability at least $1/2$.
From these $M=2^{n/2}$ vectors, we will reject the vectors which are not in lattice $\cL$.
It is easy to see that the probability that a vector sampled from the distribution $D_{\cL',s}$
is in $\cL$ is at least 
$\rho_s(\cL)/\rho_s(\cL') \ge \frac{1}{2^a}$
using Lemma~\ref{lemma_above}. Thus, the probability that we obtain at least one vector from
$\cL$ (which is distributed as $D_{\cL, s}$) is at least
\begin{align*}
    \frac{1}{2} \left(1 -(1-1/2^{a})^{2^{n/2}}\right) = \frac{1}{2} \left(1 -(1-1/2^{n/2 + 4})^{2^{n/2}}\right)
    &\ge \frac{1}{2} \cdot \left(1 - e^{-2^{n/2}/2^{n/2+4}}\right)
    = \frac{1}{2} (1 - e^{-1/16}).
\end{align*}
It implies that after rejection of vectors, with constant probability we will get at least one vector from $D_{\cL,s}$.  
Thus, the expected number of times we need to repeat the algorithm is $O(m)$ until we obtain vectors
$\vect{y_1},\hdots,\vect{y_m}$ whose distribution is statistically close to being independently distributed from $D_{\cL,s}$. 
The time and space complexity is clear from the algorithm.

We can ensure that all samples have $\poly(n)$ bit-size by first generating more samples (say twice the amount)
and throwing away all samples of norm larger than $\exp(\Omega(n^2))$. Since the vectors are sampled from a Gaussian with width
at most\footnote{Here we are using our assumption that the basis vectors have size at most $2^{o(n)}$.}
$\exp(O(n))$, the error induced by throwing away the tail of the distribution is smaller than
$2^{-\Omega(n^2)}$.
\end{proof}

\subsection{Reduction from $\BDD$ to $\DGS$}
In \cite{DRS14}, authors gave an algorithm for $\BDDP$ which requires a sampler from the discrete Gaussian distribution
\emph{exactly at the smoothing parameter} $\eta_\varepsilon(\cL)$ which is generally not known. As $\BDDP$ allows preprocessing on the input lattice for arbitrary time, they are able to assume that preprocessing advice strings contains vectors sampled exactly at smoothing parameter. In this section, we present a modification of their algorithm, that gives a reduction from $\BDD$ to $\DGS$. In this paper, our goal is to solve $\BDD$ so we will be very precise with the runtime and not assume any preprocessing.

\begin{restatable}{theorem}{theoremBDDtoDGSnew}\label{theorem_BDDtoDGS_new}%
    For any $\alpha>0$,  any $e^{-n^2}\leqslant\epsilon\leqslant \min(e^{-n^\alpha},1/200)$.
    There exists an algorithm that, on input $\cL$ and with constant probability, constructs a
    classical and a quantum (with QRAM) $\frac{\phi(\cL)}{\lambda_1(\cL)}$-$\BDD$ oracle
    for $\cL$ by doing $\poly(n)$ calls to a $0.5$-$\operatorname{hDGS}^m_{\eta_\varepsilon}$ sampler
    on the lattice $\cL^*$ and requires storage space $m\cdot\poly(n)$,
    where $m = O(\frac{n \ln(1/\epsilon)}{\sqrt{\epsilon}})$ and $\phi(\cL) \equiv \frac{\sqrt{\ln(1/\epsilon)/\pi-o(1)}}{ 2\eta_\epsilon(\cL^*)}$.
    Each call to the classical oracle takes time $m\cdot\poly(n)$ and
    space $O(\poly(n)+\ln m)$ excluding the storage space of the preprocessing.
    Every call to the quantum oracle with QRAM takes time $\sqrt{m}\cdot\poly(n)$, classical space $O(\poly(n)+\ln m)$, $\poly(n)$ qubits and requires a QRAM of size
    $m\cdot\poly(n)$
    that contains the preprocessed data\footnote{Here, we are assuming that the Gaussian samples appear in the streaming fashion.}.
\end{restatable}
The proof follows from Theorems~\ref{theorem_BDDtoDGS} and \ref{theorem_BDDtoDGS_quantum}. For completeness a detailed proof is given in  Section~\ref{sec:new_bdd_red}.

\subsection{BDD when $\varepsilon$ is small}

In order to go further, we will make heavy use of Theorem~\ref{theorem_BDDtoDGS_new}
and Lemma~\ref{lemma_smoothinglambda} to relate the smoothing parameter to other parameters of
the lattice. A small difficulty when applying Lemma~\ref{lemma_smoothinglambda} is the case distinction
on $\varepsilon$. We will start by using inequality~\eqref{ineq:smoothingtosvsmall} which will require
to take very small values of $\varepsilon$ when sampling the discrete Gaussian.

\begin{lemma}\label{lemma_alphaBDD}
    For any sufficiently large $n$, any lattice $\cL\subset \mathbb{R}^n$
    and any $A$ such that $\tfrac{1}{2\ln2}-b+o(1)\leqslant A\leqslant 1$,
    there exists a randomized algorithm that creates a classical and a quantum (with QRAM) $\alpha$-$\BDD$ oracle in time
    $2^{(A+1)n/2+o(n)}$ and space $2^{0.5n+o(n)}$, where $\alpha=\frac{1}{2}\sqrt{\tfrac{A}{A+b}}$
    and $b=\log_2\beta(\cL)$. Every call to the classical oracle takes time $2^{A n/2+o(n)}$ and space $\poly(n)$,
    excluding the space of the preprocessed data.
    Every call to the quantum oracle with QRAM takes time $2^{An/4+o(n)}$, classical space $\poly(n)$, $\poly(n)$ qubits and requires a QRAM of size $2^{An/2+o(n)}$
    that contains the preprocessed data.
\end{lemma}

\begin{proof}
    Let $\varepsilon = 2^{-An}$, $A\leqslant 1$ to be fixed later. We know that $\eta_{\varepsilon}(\cL^*)>\eta_{1/3}(\cL^*)$
    for any sufficiently large $n$ ($n>\tfrac{1}{A}\log_23$) by the monotonicity of the smoothing parameter function.
    Hence the $\DGS^m_{\eta_{1/3}}$ sampler from Lemma~\ref{lem:dgs_at_smoothing} can be used as
    a $\DGS^m_{\eta_{\varepsilon}}$ sampler, for any $m\in\nat$.

    By Theorem~\ref{theorem_BDDtoDGS_new} we can construct a $\alpha-$BDD
    such that each call takes time $m\cdot\poly(n)=2^{An/2+o(n)}$ and space $\poly(n)$, where
    $\alpha=\phi(\cL)/\lambda_1(\cL)=\frac{\sqrt{\ln(1/\epsilon)/\pi-o(1)}}{ 2\eta_\epsilon(\cL^*)\lambda_1(\cL)}$
    and $m=O(\frac{n\log(1/\eps)}{\sqrt{\eps}})=2^{An/2+o(n)}$.
    The preprocessing consists of $\poly(n)$ calls to the $\DGS^m_{\eta_\varepsilon}$ sampler described above
    and requires space $m\cdot\poly(n)$. Hence the total complexity is $\poly(n)\cdot m\cdot 2^{n/2+o(n)}=2^{(A+1)n/2+o(n)}$
    in time and $2^{An/2+o(n)}\leqslant 2^{n/2+o(n)}$ in space.
    By using Lemma~\ref{lemma_smoothinglambda}, inequality~\eqref{ineq:smoothingtosvsmall},
    only valid when $\varepsilon \leq (e/\beta(\cL)^2+o(1))^{-\frac{n}{2}}$, we have that
    \[
        \lambda_1(\cL) \eta_\varepsilon(\cL^*) < \sqrt{\frac{\ln(1/\varepsilon)+n\ln\beta(\cL)+o(n)}{\pi}}.
    \]
    Hence we can guarantee that
    \[
        \alpha
            =\frac{\sqrt{\ln(1/\varepsilon)/\pi-o(1)}}{ 2\eta_\varepsilon(\cL^*)\lambda_1}
            >\frac{\sqrt{\ln(1/\varepsilon)/\pi-o(1)}}{2\sqrt{\frac{\ln(1/\varepsilon)+n\ln\beta(\cL)+o(n)}{\pi}}}
            =\frac{1}{2}\sqrt{\frac{\ln(1/\varepsilon)+o(1)}{\ln(1/\varepsilon)+n\ln\beta(\cL)}}
            =\frac{1}{2}\sqrt{\tfrac{A}{A+b}}+o(1)
    \]
    where $b=\log_2\beta(\cL)$.
    Furthermore, as noted above,
    this inequality is only valid when $\varepsilon \leq (e/\beta^2+o(1))^{-\frac{n}{2}}$, that is $A\geqslant\tfrac{1}{2\ln2}-b+o(1)$.
    Finally, note that since $b\leqslant0.402$, we must have $A\geqslant0.32$ and the inequality holds
    as soon as $n\geqslant5\geqslant\tfrac{1}{A}\log_23$. Finally note that Theorem~\ref{theorem_BDDtoDGS_new} requires
    $\varepsilon\leqslant1/200$ which holds as soon as $n\geqslant 17\geqslant \tfrac{1}{A}\ln200$.

    The quantum algorithm is exactly the same but using the quantum oracle of
    Theorem~\ref{theorem_BDDtoDGS_new} which then has running time
    $\sqrt{m}\cdot\poly(n)=2^{An/4+o(n)}$ and requires a QRAM of size $m\cdot\poly(n)=2^{An/2+o(n)}$.
\end{proof}

We can reformulate the previous lemma by expressing the complexity in terms of $\alpha$ instead of
some arbitrary constant $A$.

\begin{corollary}\label{coro_alphaBDD}
    For any $n\geqslant5$, any lattice $\cL\subset \mathbb{R}^n$
    and any $\alpha$ such that $\tfrac{1}{2}\sqrt{1-2b\ln2}+o(1)\leqslant\alpha<\tfrac{1}{2}\sqrt{\tfrac{1}{1+b}}$,
    there exists a randomized algorithm that creates a classical and a quantum (with QRAM) $\alpha$-$\BDD$ oracle in time
    $2^{(A+1)n/2+o(n)}$ and space $2^{0.5n+o(n)}$.
    Every call to the classical oracle takes time
    $2^{An/2+o(n)}$ and space $\poly(n)$, excluding the space of the preprocessed data;
    and every call to the quantum oracle with QRAM takes time $2^{An/4+o(n)}$, classical space $\poly(n)$, $\poly(n)$ qubits and requires a QRAM of size $2^{An/2+o(n)}$
    that contains the preprocessed data, where $A=\tfrac{4b\alpha^2}{1-4\alpha^2}$ and $b=\log_2\beta(\cL)$.
    
\end{corollary}
\begin{proof}
    Apply Lemma~\ref{lemma_alphaBDD} for some $A$ to be fixed later. Observe that
    $\alpha=\frac{1}{2}\sqrt{\tfrac{A}{A+b}}$ so $A=\tfrac{4b\alpha^2}{1-4\alpha^2}$.
    Now the constraints $\tfrac{1}{2\ln2}-b+o(1)\leqslant A\leqslant 1$ become
    \begin{align*}
        \tfrac{1}{2\ln2}-b+o(1)\leqslant \tfrac{4b\alpha^2}{1-4\alpha^2}
        &\Leftrightarrow (\tfrac{1}{2\ln2}-b+o(1))(1-4\alpha^2)\leqslant 4b\alpha^2\\
        &\Leftrightarrow \tfrac{1}{4\ln2}-\tfrac{b}{2}+o(1)\leqslant \tfrac{\alpha^2}{\ln2}\\
        &\Leftrightarrow \tfrac{1}{2}\sqrt{1-2b\ln2}+o(1)\leqslant \alpha
    \end{align*}
    and
    \[
        \tfrac{4b\alpha^2}{1-4\alpha^2}\leqslant 1
        \Leftrightarrow 4(1+b)\alpha^2\leqslant 1
        \Leftrightarrow \alpha\leqslant \tfrac{1}{2}\sqrt{\tfrac{1}{1+b}}.
    \]
\end{proof}

\subsection{BDD when $\varepsilon$ is large}

The inequality~\eqref{ineq:smoothingtosvsmall} in Lemma~\ref{lemma_smoothinglambda} tells us
that if we take an extremely small $\epsilon$ to compute the $\BDD$ oracle,
we can find a $\BDD$ oracle with $\alpha(\cL)$ almost $1/2$.
However the time complexity for each call of the oracle will be very costly.
On the other hand, if we use the inequality~\eqref{ineq:smoothingtosvlarge} in Lemma~\ref{lemma_smoothinglambda}
with a larger $\epsilon$, each call of the oracle will take much less time,
but the constraint on the decoding coefficient $\alpha$ will be different. It is therefore
important to study this second regime as well. Note that inequality \eqref{ineq:smoothingtosvlarge}
actually applies to all $\varepsilon\in(0,1)$ but is mostly useful when $\varepsilon$ is large.

\begin{lemma}\label{lemma_alphaBDD_alt}
    For any sufficiently large $n$, any lattice $\cL\subset \mathbb{R}^n$
    and any $\tfrac{1}{n}\log_23\leqslant A\leqslant 1$,
    there exists a randomized algorithm that creates a classical and a quantum (with QRAM) $\alpha$-$\BDD$ oracle in time
    $2^{(A+1)n/2+o(n)}$ and space $2^{0.5n+o(n)}$, where $\alpha=\frac{2^{-A}\sqrt{A}\sqrt{2e\ln2}}{2\beta(\cL)}-o(1)$.
    Every call to this oracle takes time $2^{An/2+o(n)}$ and space $\poly(n)$,
    excluding the space of the preprocessed data.
    Every call to the quantum oracle with QRAM takes time $2^{An/4+o(n)}$, classical space $\poly(n)$, $\poly(n)$ qubits and requires a QRAM of size $2^{An/2+o(n)}$
    that contains the preprocessed data.

\end{lemma}

\begin{proof}
    Let $\varepsilon = 2^{-An}$, $A\leqslant 1$ to be fixed later. We know that $\eta_{\varepsilon}(\cL^*)>\eta_{1/3}(\cL^*)$
    for any sufficiently large $n$ ($n>\tfrac{1}{A}\log_23$) by the monotonicity of the smoothing parameter function.
    Hence the $\DGS^m_{\eta_{1/3}}$ sampler from Lemma~\ref{lem:dgs_at_smoothing} can be used as
    a $\DGS^m_{\eta_{\varepsilon}}$ sampler, for any $m\in\nat$.

    By Theorem~\ref{theorem_BDDtoDGS_new} we can construct an $\alpha-$BDD
    such that each call takes time $m\cdot\poly(n)=2^{An/2+o(n)}$ and space $\poly(n)$, where
    $\alpha=\phi(\cL)/\lambda_1(\cL)=\frac{\sqrt{\ln(1/\epsilon)/\pi-o(1)}}{ 2\eta_\epsilon(\cL^*)\lambda_1(\cL)}$
    and $m=O(\frac{n\log(1/\eps)}{\sqrt{\eps}})=2^{An/2+o(n)}$.
    The preprocessing consists of $\poly(n)$ calls to the $\DGS^m_{\eta_\varepsilon}$ sampler described above
    and requires space $m\cdot\poly(n)$. Hence the total complexity is $\poly(n)\cdot m\cdot 2^{n/2+o(n)}=2^{(A+1)n/2+o(n)}$
    in time and $2^{An/2+o(n)}\leqslant 2^{n/2+o(n)}$ in space.
    By using inequality~\eqref{ineq:smoothingtosvlarge},
    in Lemma~\ref{lemma_smoothinglambda}, we have that
    \[
        \lambda_1(\cL) \eta_\epsilon(\cL^*) < \sqrt{\frac{\beta(\cL)^2n}{2\pi e}}\cdot \epsilon^{-1/n}(1+o(1)).
    \]
    Hence we can guarantee that
    \[
        \alpha
            =\frac{\sqrt{\ln(1/\epsilon)/\pi-o(1)}}{2\eta_\epsilon(\cL^*)\lambda_1(\cL)}
            >\frac{1}{2}\sqrt{\frac{2e\ln\frac{1}{\epsilon}-o(1)}{n}}\cdot\beta(\cL)^{-1}\epsilon^{\frac{1}{n}}\cdot(1-o(1))
            =\frac{2^{-A}\sqrt{A}\cdot\sqrt{2e\ln2}}{2\beta(\cL)}-o(1).
    \]
    The quantum algorithm is exactly the same but using the quantum oracle of
    Theorem~\ref{theorem_BDDtoDGS_new} which then has running time
    $\sqrt{m}\cdot\poly(n)=2^{An/4+o(n)}$ and requires a QRAM of size $m\cdot\poly(n)=2^{An/2+o(n)}$.
\end{proof}

\begin{corollary}\label{coro_alphaBDDalt}
    For any $n\geqslant2$, any integer $m>0$, any lattice $\cL\subset \mathbb{R}^n$
    and any $\tfrac{\sqrt{e\ln3}}{\sqrt{2n}\beta(\cL)}\leqslant\alpha\leqslant\tfrac{1}{2\beta(\cL)}$, where $b=\log_2\beta(\cL)$,
    there exists a randomized algorithm that creates a classical and a quantum (with QRAM) $(\alpha+o(1))$-$\BDD$ oracle in time $2^{(A+1)n/2+o(n)}$ and space $2^{n/2+o(n)}$.
    Every call to the classical oracle takes time
    $2^{An/2+o(n)}$ and space $\poly(n)$, excluding the space of the preprocessed data;
    and every call to the quantum oracle with QRAM takes time $2^{An/4+o(n)}$, classical space $\poly(n)$, $\poly(n)$ qubits and requires a QRAM of size $2^{An/2+o(n)}$
    that contains the preprocessed data, where
    \[
        A=-\tfrac{1}{2\ln 2}W(-\tfrac{4\alpha^2\beta(\cL)^2}{e})
    \]
    where $W$ is the principal branch of the Lambert W function.
    Furthermore, the above expression of $A$ is a continuous and increasing function of $\beta(\cL)$.
\end{corollary}
\begin{proof}
    By Lemma~\ref{lemma_alphaBDD_alt},
    we can build an oracle for any $\tfrac{1}{n}\log_23\leqslant A\leqslant 1$ such
    that the decoding radius is $\alpha=\frac{2^{-A}\sqrt{A}\sqrt{2e\ln2}}{2\beta(\cL)}-o(1)$.
    Hence, we want to find $A$ such that
    \[
        \frac{2^{-A}\sqrt{A}\sqrt{2e\ln2}}{2\beta(\cL)}=\alpha\qquad\text{and}\qquad
            \tfrac{1}{n}\log_23\leqslant A\leqslant 1.
    \]
    Let $f:A\mapsto 2^{-A}\sqrt{A}$ so that the first condition is equivalent to
    \begin{equation}\label{eq:cond_bdd_existence_A}
        f(A)=\tfrac{2\alpha\beta(\cL)}{\sqrt{2e\ln2}}.
    \end{equation}
    Now assume that \eqref{eq:cond_bdd_existence_A} holds and let $y=-2A\ln(2)$, then it is equivalent to
    \[
        e^yy=-2\ln(2)\tfrac{2\alpha^2\beta(\cL)^2}{e\ln2},
    \]
    that is
    \begin{equation}\label{eq:bdd_eq_y}
        e^yy=-\tfrac{4\alpha^2\beta(\cL)^2}{e}.
    \end{equation}
    This equation admits a solution if and only if 
    \begin{equation}
        -\tfrac{4\alpha^2\beta(\cL)^2}{e}\geqslant-\tfrac{1}{e}
        \quad
        \Leftrightarrow
        \quad
        \alpha\leqslant\tfrac{1}{2\beta(\cL)}.\label{eq:bdd_cond_alpha_1}
    \end{equation}
    Assuming this is the case, \eqref{eq:bdd_eq_y} can admit up to two solutions.
    However, since the complexity increases with $A$,
    we want the solution that minimizes $A$, \emph{i.e.} that maximizes $y$. The largest of the (up to)
    two solutions of \eqref{eq:bdd_eq_y} is always given by be the principal branch $W$ of the Lambert W function:
    \begin{equation}\label{eq:bdd_sol_A}
        y=W(-\tfrac{4\alpha^2\beta(\cL)^2}{e})
        \quad\text{that is}\quad
        A=-\tfrac{1}{2\ln 2}W(-\tfrac{4\alpha^2\beta(\cL)^2}{e})
    \end{equation}
    and always satisfies $y\geqslant-1$. In particular, we always have $A\leqslant\tfrac{1}{2\ln 2}$.
    Now check that $f$ is strictly increasing over $[0,\tfrac{1}{2\ln2}]$. Hence,
    the condition $\tfrac{1}{n}\log_23\leqslant A$ is equivalent to
    \begin{align*}
        &f(\tfrac{1}{n}\log_23)\leqslant f(A)\\
        &\Leftrightarrow f(\tfrac{1}{n}\log_23)^2\leqslant \left(\tfrac{2\alpha\beta(\cL)}{\sqrt{2e\ln2}}\right)^2
            &&\text{by \eqref{eq:cond_bdd_existence_A}}\\
        &\Leftrightarrow 2^{-\tfrac{2}{n}\log_23}\tfrac{1}{n}\log_23
            \leqslant\tfrac{2\alpha^2\beta(\cL)^2}{e\ln2}\\
        &\Leftrightarrow \tfrac{e\ln3}{2n\beta(\cL)^2}9^{-\tfrac{1}{n}}\leqslant \alpha^2\\
        &\Leftarrow \tfrac{e\ln3}{2n\beta(\cL)^2}\leqslant \alpha^2.
        \numberthis\label{eq:bdd_cond_alpha_2}
    \end{align*}
    In summary, we can always take $A$ as in \eqref{eq:bdd_sol_A} assuming \eqref{eq:bdd_cond_alpha_1}
    and \eqref{eq:bdd_cond_alpha_2} hold.
\end{proof}

\subsection{Putting everything together}

We have analyzed the construction of $\alpha$-BDDs in two regimes, based on Lemma~\ref{lem:dgs_at_smoothing}.
It is not a priori clear which construction is better and in fact we will see that it depends
in a nontrivial way on the relation between $\alpha$ and $\beta(\cL)$.

\begin{theorem}\label{theorem_alphaBDD}
    For any sufficiently large $n$, any $m>0$, any $\tfrac{\sqrt{e\ln3}}{\sqrt{2n}\beta(\cL)}\leqslant\alpha<\tfrac{1}{2}\sqrt{\tfrac{1}{1+b}}$
    and any lattice $\cL\subset \mathbb{R}^n$,
    there exists a randomized algorithm that creates a classical and a quantum (with QRAM) $(\alpha+o(1))$-$\BDD$ oracle in time $2^{(A+1)n/2+o(n)}$ and space $2^{n/2}$.
    Every call to the classical oracle takes time
    $2^{An/2+o(n)}$ and space $\poly(n)$, excluding the space of the preprocessed data;
    and every call to the quantum oracle with QRAM takes time $2^{An/4+o(n)}$, classical space $\poly(n)$, $\poly(n)$ qubits and requires a QRAM of size $2^{An/2+o(n)}$
    that contains the preprocessed data, where
    \[
        A=\begin{cases}
                -\tfrac{1}{2\ln 2}W(-\tfrac{4\alpha^2\beta(\cL)^2}{e})&\text{when }b<\tfrac{1-4\alpha^2}{2\ln2}\\
                \tfrac{4\alpha^2}{1-4\alpha^2}b&\text{when }b\geqslant\tfrac{1-4\alpha^2}{2\ln2}
            \end{cases}
    \]
    where $W$ is the principal branch of the Lambert W function and $b=\log_2\beta(\cL)$.
    Furthermore, the above expression of $A$ is a continuous and increasing function of $b$.
\end{theorem}
\begin{proof}
    Let $\tfrac{\sqrt{e\ln3}}{\sqrt{2n}\beta(\cL)}\leqslant\alpha<\tfrac{1}{2}\sqrt{\tfrac{1}{1+b}}$ and $b=\log_2\beta(\cL)$.
    By Corollary~\ref{coro_alphaBDD}, we can build an $\alpha$-$\BDD$ if
    $\tfrac{1}{2}\sqrt{1-2b\ln2}\leqslant\alpha$, in which
    case the complexity will depend on $A=A_1(\alpha,b):=\tfrac{4b\alpha^2}{1-4\alpha^2}$.
    By Corollary~\ref{coro_alphaBDDalt}, we can build an $\alpha$-$\BDD$ if
    $\alpha<\tfrac{1}{2}\sqrt{1-2b\ln2}$,
    in which case the complexity will depend on $A=A_2(\alpha,b):=-\tfrac{1}{2\ln 2}W(-\tfrac{4\alpha^2\beta(\cL)^2}{e})$.
    In both cases, the BDD oracle can be created in time $2^{(A+1)n/2+o(n)}$, space $2^{0.5n+o(n)}$
    and each call takes time $2^{An/2+o(n)}$. Now observe that
    \[
        \alpha<\tfrac{1}{2}\sqrt{1-2b\ln2}
            \;\Leftrightarrow\; 4\alpha^2<1-2b\ln2
            \;\Leftrightarrow\; b<\tfrac{1-4\alpha^2}{2\ln2}.
    \]
    Let $b^*:=\tfrac{1-4\alpha^2}{2\ln2}$, then there are two cases:
    \begin{itemize}
        \item If $b\geqslant b^*$ then $\tfrac{1}{2}\sqrt{1-2b\ln2}\leqslant\alpha$
            so only Corollary~\ref{coro_alphaBDD} applies and we can build a $\alpha$-$\BDD$.
            In this case the complexity exponent is $A_1(\alpha,b)=\tfrac{4\alpha^2}{1-4\alpha^2}b$.
        \item If $b<b^*$ then $\alpha<\tfrac{1}{2}\sqrt{1-2b\ln2}$ so Corollary~\ref{coro_alphaBDDalt}
            applies but Corollary~\ref{coro_alphaBDD} does not for this particular value of $\alpha$.
            However, we can apply Corollary~\ref{coro_alphaBDD} to build a $\alpha'$-$\BDD$ oracle with
            $\alpha'\geqslant\alpha_1^{min}(b):=\tfrac{1}{2}\sqrt{1-2b\ln2}>\alpha$.
            We will show that the $\alpha$-$\BDD$ of Corollary~\ref{coro_alphaBDDalt} is always more
            efficient than the $\alpha'$-$\BDD$ of Corollary~\ref{coro_alphaBDD} in this case
            and the complexity exponent will thus be $A_2(\alpha,b)=-\tfrac{1}{2\ln 2}W(-\tfrac{4\alpha^2\beta(\cL)^2}{e})$.
    \end{itemize}
    Assume that $b<b^*$, we claim that
    $A_1(\alpha,b)\geqslant A_2(\alpha',b)$ for any $\alpha'\geqslant\alpha_1^{min}(b)$. Indeed, on the one hand
    $A_2$ is an increasing function of $b$ so
    \[
        A_2(\alpha,b)
            <A_2(\alpha,\tfrac{5}{18\ln2})=
            -\tfrac{1}{2\ln 2}W(-4\alpha^2e^{-4\alpha^2})
            =\tfrac{4\alpha^2}{2\ln 2}
            =\tfrac{2\alpha^2}{\ln 2}
            \qquad\text{since }W(xe^x)=x.
    \]
    On the other hand, $A_1$ is an increasing function of $\alpha$ so
    \[
        A_1(\alpha,b)
            \geqslant A_1(\alpha_1^{min}(b),b)=\tfrac{1-2b\ln2}{2\ln2}
    \]
    which is a decreasing function of $b$, therefore
    \[
        A_1(\alpha,b)\geqslant A_1(\alpha_1^{min}(b^*),b^*)=\tfrac{1-2b^*\ln2}{2\ln2}=\tfrac{2\alpha^2}{\ln 2}>A_2(\alpha,b).
        %\qedhere
    \]
\end{proof}

\begin{corollary}\label{cor_1_3BDD}
    For any sufficiently large integer $n$, any integer $m>0$, any lattice $\cL\subset \mathbb{R}^n$
    there exists a randomized algorithm that creates a classical and a quantum (with QRAM) $1/3$-$\BDD$ oracle in time
    $2^{0.6604n+o(n)}$ and space $2^{0.5n+o(n)}$.
    Every call to the classical oracle takes time $2^{0.1604n+o(n)}$ and space $\poly(n)$,
    excluding the space of the preprocessed data.
    Every call to the quantum oracle with QRAM takes time $2^{0.0802n+o(n)}$, classical space $\poly(n)$, $\poly(n)$ qubits and requires a QRAM of size $2^{0.1604n+o(n)}$
    that contains the preprocessed data.
\end{corollary}
\begin{proof}
    By Theorem~\ref{theorem_alphaBDD}, the value of $A$ increases with $b=\log_2\beta(\cL)$.
    Since $b\leqslant 0.402$ and $0.402\geqslant\tfrac{5}{18\ln2}$,
    we always have $A\leqslant\tfrac{4}{5}0.402+o(1)
    \leqslant0.3208+o(1)$ and we obtain the result.
\end{proof}

\section{Solving SVP by spherical caps on the sphere}
\label{sec:spherical_caps}

We now explain how to reduce the number of queries to the $\alpha$-$\BDD$ oracle.
Consider a uniformly random target vector $\vect{t}$ such that
$\alpha(1-\frac{1}{n})\lambda_1(\cL)\leq \norm{\vect{t}} <\alpha \lambda_1(\cL)$,
it satisfies the condition of Theorem \ref{theorem_enlargeBDD},
\emph{i.e}.  $\operatorname{dist}(\cL,\vect{t})\leq \alpha\lambda_1(\cL)$.
We enumerate all lattice vectors within distance $2\alpha\lambda_1(\cL)$ of $\vect{t}$ and
keep only the shortest nonzero one. We show that for any $\alpha\geqslant\tfrac{1}{3}$,
we will get the shortest nonzero vector of the lattice
with probability at least $2^{-cn+o(n)}$ for some $c$ that depends on $\alpha$. By repeating this $2^{cn+o(n)}$ times,
the algorithm will succeed with constant probability. The optimal choice of $\alpha$ is not obvious
and is deferred to Section~\ref{section_relation_kissing}.

\begin{figure}
    \centering
    \begin{tikzpicture}[scale=0.8]
    
    \draw (0,0) coordinate (O) node[below left] {$O$} circle (1*3cm);
    \draw (0,0) circle (0.422*3cm);
    %
    %\draw (-1.2,2.6)  node {$\phi$};
    \draw[thick,dashed]  (0,0.422*3) coordinate (P)
        node[above right] {$P$} [red] circle (0.844*3cm);
    \coordinate (A) at (0,0.41*3);
            \fill (A) circle[radius=2pt];
    \coordinate (B) at (0,0);
            \fill (B) circle[radius=2pt];    
    \coordinate (C) at (0,3);
            \fill (C) circle[radius=2pt];
    \coordinate (D) at (-2.5,1.75);
            \fill (D) circle[radius=2pt];
    \coordinate (E) at (2.5,1.75);
            \fill (E) circle[radius=2pt];   
    \draw (B) -- (A);
    \draw [red,thick,dashed](A) -- (D);
    \draw [red,thick,dashed] (0,0.422*3) -- (0,3+0.422*2);
    \draw ($(P)!0.5!(D)$) node[above,red] {$2\alpha\lambda_1$};
    \draw (0,0) -- (0,-3);
    \draw (0,0) -- (0.41*3,0);
    \draw (0.633,-0.3) node {$r$};
    \draw [blue,thick] (O) -- (E) node[midway,above,blue] {$\lambda_1$};
    \draw [blue,thick] (O) -- (D);
    \draw[blue] (O) pic[pic text={$\phi$},draw=blue,angle eccentricity=1.5] {angle = P--O--D};
    \end{tikzpicture}
    \caption{One can cover the sphere of radius $\lambda_1$ by balls of radius $2\alpha\lambda_1$,
    where $\tfrac{1}{3}\leqslant\alpha<\tfrac{1}{2}$, whose centers (here $P$) are at distance $r$ from the origin $O$.
    Each such ball covers a spherical cap of half-angle $\phi$.\label{figure_covering}}
\end{figure}
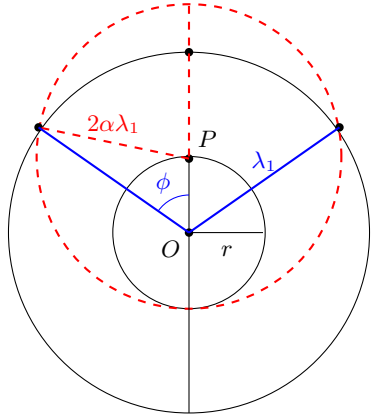

\begin{theorem}\label{theorem_bdd_svp_spherical_cap}
    Assume we can create an $\alpha$-$\BDD$ oracle, with $\alpha\geqslant\tfrac{1}{3}$, in time
    $T_c$, space $S_c$ such that each call takes time $T_o$. Then
    there is a randomized algorithm that solves, with constant probability, $\svp$ in space $S_c$ and time
    \[
        T_c+\frac{2^{o(n)}T_o}{\alpha^n}
    \]
    Furthermore, there is a quantum algorithm that
    solves $\svp$ in classical space $S_c$ using a polynomial number of qubit and time
    \[
        T_c+\frac{2^{o(n)}T_o}{\alpha^{n/2}}.
    \]
    If we can create a quantum $\alpha$-$\BDD$ oracle with QRAM in classical time
    $T_c$ and space $S_c$ such that each call takes quantum time $T_o$,
    $\poly(n)$ qubits and requires a QRAM of size $M_o$, then there is a quantum
    algorithm with QRAM that solves $\svp$ in classical space $S_c$ using a polynomial number of qubit, time
    \[
        T_c+\frac{2^{o(n)}T_o}{\alpha^{n/2}}.
    \]
    and requires a QRAM of size $M_o$.
\end{theorem}

\begin{proof}
    On input lattice $\cL(\basis)$, use the LLL algorithm~\cite{LLL82} to get a number $d$
    (the norm of the first vector of the basis) that satisfies
    $\lambda_1(\cL)\leq d\leq 2^{n/2}\lambda_1(\cL)$. For $i=1,\hdots,n^2$,
    let $d_i=d/(1+\frac{1}{n})^i$. There exists a $j$ such that
    $\lambda_1(\cL)\leq d_j\leq (1+\frac{1}{n})\lambda_1(\cL)$. We repeat the following procedure
    for all $i=1,\hdots,n^2$:

    Fix $N\in\mathbb{N}$ to be fixed later. 
    For $j=1$ to $N$, pick a uniformly random vector $\vect{v_{ij}}$ on the surface of the ball of radius
    $r(1-\frac{1}{n}) d_i$. By Theorem~\ref{theorem_enlargeBDD},
    we can enumerate $2^n$ lattice points using the function $f_{ij}:\intg^n_2 \rightarrow \cL$ defined by
    \begin{equation}\label{eq:def_fij}
        f_{ij}(\vect{x})=\basis \vect{x}- 2 \cdot\BDD_{\alpha}(\cL,(\basis\vect{x}- \vect{v_{ij}})/2).
    \end{equation}
    At each step we only store the shortest nonzero vector. At the end, we output the shortest among them.
    The running time of the algorithm is straightforward.
    We make $2^n$ queries to a $\alpha$-$\BDD$ oracle that takes time $T_o$,
    we further repeat this $n^2N$ times. Therefore the algorithm takes time $T_c+2^n\cdot n^2N\cdot T_o$
    and space $S_c$. This entire procedure is summarized in Algorithm~\ref{alg:bdd_svp_spherical_cap}.

    \begin{algorithm}[h]
        \caption{Solving SVP by spherical caps on the sphere\label{alg:bdd_svp_spherical_cap}}
        \begin{algorithmic}[1]
            \REQUIRE basis $\vect{B}$ of a lattice $\cL\subset\mathbb{R}^n$
            \REQUIRE an $\alpha$-$\BDD$ oracle (for a well-chosen $\alpha$)
            \ENSURE a shortest non-zero vector of $\cL$
            \STATE use LLL to get a number $d$ that satisfies $\lambda_1(\cL)\leq d\leq 2^{n/2}\lambda_1(\cL)$.
            \STATE $\vect{z}\gets$ any basis vector
            \FOR{$i=1,\hdots,n^2$}
                \STATE $d_i\gets d(1+\frac{1}{n})^{-i}$
                \STATE $N\gets\frac{A_n(\lambda_1)}{V_{n-1}(2\alpha\lambda_1)}$\Comment{see \eqref{eq:def_An_Vn}}
                \FOR{$j=1,...,N$}
                    \STATE $\vect{v}_{ij}\gets$ random vector of norm $r(1-\tfrac{1}{n})d_i$
                    \FOR{$\vect{x}\in\set{0,1}^n$}
                        \STATE $\vect{y}\gets f_{ij}(\vect{x})$\Comment{see \eqref{eq:def_fij}}
                        \IF{$\norm{\vect{y}}<\norm{\vect{x}}$}
                            \STATE $\vect{z}\gets\vect{y}$\Comment{shorter vector}
                        \ENDIF
                    \ENDFOR
                \ENDFOR
            \ENDFOR
            \RETURN $\vect{z}$
        \end{algorithmic}
    \end{algorithm}

    To prove the correctness of the algorithm, it suffices to show that there exists an $i\in [n^2]$
    for which the algorithm finds the shortest vector with high probability.
    Recall that there exists an $i$ such that $\lambda_1(\cL)\leq d_i\leq (1+\frac{1}{n})\lambda_1(\cL)$
    and let that index be $k$. We will show that for a uniformly random vector $\vect{v}$
    of length $r(1-\frac{1}{n})d_k $, if we enumerate $2^n$ vectors by the function  $f:\intg^n_2 \rightarrow \cL$,
    \begin{equation}\label{eq:def_f_spherical_cap}
        f(\vect{x})=\basis\vect{x}- 2 \cdot\text{BDD}_{\alpha}(\cL,(\basis\vect{x}- \vect{v})/2),
    \end{equation}
    then with probability $\delta$, whose expression is derived in the next paragraph,
    there exists $\vect{x}\in \intg^n_2$ such that $f(\vect{x})$ is the shortest nonzero lattice vector;
    we will then choose $N=1/\delta$ so that repeating $N$ times this process finds the shortest 
    vector with probability bounded from below by a constant.

    To that aim, we show that we can cover the sphere of radius $\lambda_1$ by $N$ balls of radius
    $2\alpha \lambda_1$ whose centers are at distance $r (1-\frac{1}{n})d_k\leq r\lambda_1$
    from the origin (see Figure~\ref{figure_covering} where we took $r=\alpha$).
    We have two concentric circles of radius $r (1-\frac{1}{n}) d_k$ and $\lambda_1$,
    and let $P$ be a uniformly random point on the surface of the ball of radius
    $r(1-\frac{1}{n})d_k$. A ball of radius $2\alpha\lambda_1$ at center $P$ will cover the spherical
    cap with angle $\phi$ of the ball of radius $\lambda_1$. For convenience, write $r=s\lambda_1$
    for some $s$. We can calculate the optimal choice of $r$
    by noting that if we take the center of the caps to be at distance $r$
    then the angle $\phi$ satisfies
    $\cos\phi=\frac{\lambda_1^2+r^2-4\alpha^2\lambda_1^2}{2r\lambda_1}=f(s)$ by the law of cosines,
    where $f(x)=\frac{1+x^2-4\alpha^2}{2x}$.
    We want to maximize the angle $\phi$, since the area we can
    cover increases with $\phi$. For minimizing $\cos(\phi)$, we
    minimize $f$.
    Check that $f(x)$ is decreasing until $\sqrt{1-4\alpha^2}$ and then increasing.
    We conclude that the optimal radius is when
    $s=\sqrt{1-4\alpha^2}$ and this gives an optimal angle $\phi$
    such that $\cos\phi=\sqrt{1-4\alpha^2}$
    and therefore $\sin\phi=2\alpha$.

    Now observe that the surface area of any such cap
    is lower bounded by the surface area of the base of the
    cap, which is a $(n-1)$-dimensional sphere of radius $\lambda_1\sin\phi$. Hence the number of spherical
    caps required to cover the surface of sphere is in the order
    of $N:=A_n(\lambda_1)/V_{n-1}(\lambda_1\sin\phi)$ where $A_n$ (resp. $V_n$) is the surface area
    (resp. volume) of
    a $n$-dimensional sphere:
    \begin{equation}\label{eq:def_An_Vn}
        A_n(\lambda_1)=\frac{2\pi^{n/2}\lambda_1^{n-1}}{\Gamma(n/2)},
        \qquad V_{n-1}(\lambda_1\sin\phi)=\frac{\pi^{(n-1)/2}\lambda_1^{n-1}\sin^{n-1}\phi}{\Gamma((n+1)/2)}.
    \end{equation}
    Thus we have 
    \[
        N=\frac{A_n(\lambda_1)}{V_{n-1}(\lambda_1\sin\phi)}
                =\frac{2^{o(n)}}{\sin^{n-1}\phi}
            =\frac{2^{-n+o(n)}}{\alpha^n}.
    \]
    If we randomly choose the center $\vect{v}$ of the sphere,
    the corresponding spherical caps will cover the shortest vector with probability at least
    $1/N$. By Theorem~\ref{theorem_enlargeBDD}, the list $\{f(\vect{x})\mid \vect{x}\in \intg^n_2\}$
    will contain all lattice points within radius $2\alpha d_k$ from $\vect{v}$.
    Hence, the list will contain a shortest vector with probability $1/N$.
    By repeating this process $N$ times, we can find the shortest vector with constant probability.
    Therefore, an upper bound of the total time complexity of our method can be expressed as
    \[
        T_c+2^nN\cdot T_o
            =T_c+\frac{2^{o(n)}T_o}{\alpha^n}.
    \]
    In the quantum case, we can apply the quantum minimum finding algorithm in Theorem~\ref{theorem_Qmin} to speed up search.
    Let $f$ be the function defined in \eqref{eq:def_f_spherical_cap}.
    The algorithm works on three quantum registers and our goal is to build a superposition of states of the form $\ket{\boldsymbol{s}}\ket{f(\boldsymbol{s})}\ket{x}$
    where $x=\norm{f(\boldsymbol{s})}$ most of the time (see the definition
    of $U$ below).
    Recall that we assumed that we can create a classical
    $\alpha$-BDD in time $T_c$ and space $S_c$ such that each
    call takes time $T_o$.
    Hence, we can first create the oracle in time $T_c$ and space
    $T_c$, and then hardcode the preprocessed data that we obtained
    into the oracle to obtain a new, self-contained, oracle that
    still runs in time $T_o$ and polynomial space.
    Let $\varepsilon=o(\tfrac{n}{\ln T_o})$ and apply Corollary~\ref{cor:irreversible_to_quantum}
    to construct a quantum circuit $\mathcal{O}_{BDD}$ of size $T_o^{1+\varepsilon}=2^{o(n)}T_o$ on $\poly(n)$ qubits
    that satisfies
    $\mathcal{O}_{BDD}|  \vect{s} \rangle |  \vect{0} \rangle = |  \vect{s} \rangle |   f(\vect{s}) \rangle$ for all $\vect{s}\in \mathbb{Z}^n_3$.
    We then construct another quantum circuit $U$ satisfying 
    \begin{equation*}
    U(| \vect{\omega}\rangle|  0\rangle) = \begin{cases}
    | \vect{\omega}\rangle| \norm{\vect{\omega}}\rangle &\text{if }\vect{\omega}\neq \vect{0}\\
    | \vect{\omega}\rangle| \norm{\boldsymbol{Be_1}}+1\rangle &\text{if }\vect{\omega} = \vect{0},
    \end{cases}
    \end{equation*}
    Here $\vect{e_1}\in \intg^n$ is a vector whose first coordinate is one and the rest are zero.
    We then consider the following quantum circuit (we have not drawn ancilla qubits):
    \begin{center}
        \begin{tikzpicture}
            \draw[thick] (2,-1/3) rectangle (4,-5/3) node[midway] {$\mathcal{O}_{BDD}$};
            \draw[thick] (6,-1) rectangle (8,-1-4/3) node[midway] {$U$};
            \draw (0,-2/3) node[left] {$\ket{\vect{s}}$} -- (2,-2/3)
                  (4,-2/3) -- (10,-2/3) node[right] {$\ket{\vect{s}}$};
            \draw (0,-4/3) node[left] {$\ket{\vect{0}}$} -- (2,-4/3)
                  (4,-4/3) -- (6,-4/3) node[above=-1pt,midway] {$\scriptstyle \ket{f(\vect{s})}$}
                  (8,-4/3) -- (10,-4/3) node[right] {$\ket{f(\vect{s})}$};
            \draw (0,-6/3) node[left] {$\ket{0}$} -- (6,-6/3)
                  (8,-6/3) -- (10,-6/3) node[right] {$\ket{\norm{f(\vect{s})}}$ or $\ket{\norm{\boldsymbol{Be_1}}+1}$};
        \end{tikzpicture}
    \end{center}
    This circuit $\mathcal{O}$ has size $2^{o(n)}T_o^{1+\varepsilon}$,
    satisfies $\mathcal{O}\ket{\boldsymbol{s}}\ket{\vect{0}}\ket{0}=\ket{\boldsymbol{s}}\ket{f(\boldsymbol{s})}\ket{\norm{f(\boldsymbol{s})}}$
    if $f(\boldsymbol{s})\neq 0$ and $\mathcal{O}\ket{\boldsymbol{s}}\ket{\vect{0}}\ket{0}=\ket{\boldsymbol{s}}\ket{f(\boldsymbol{s})}\ket{\norm{\boldsymbol{Be_1}}+1}$
    and uses $\poly(n)$ qubits.
    We can now apply the quantum minimum finding algorithm from Theorem~\ref{theorem_Qmin} on the first and third registers
    and the index $\vect{s'}$ of a shortest vector in this list. The output of the algorithm will be $f(\vect{s'})$.
    We can then build a circuit to generate the random vectors $\vect{v_{ij}}$
    above and therefore build a circuit that associates to every $(i,j,\vect{s})$ the lattice vector $f_{ij}(\vect{s})$.
    By a similar argument as above, using Corollary~\ref{cor:irreversible_to_quantum}, we conclude that
    this quantum circuit has size $2^{o(n)}T_o$ and uses a polynomial
    number of qubits.
    Finally, we can apply the quantum minimum finding algorithm on
    the set of $(i,j,\vect{s})\in[n^2]\times[N]\times[2^n]$ and obtain the shortest vector of that
    list by making $\sqrt{n^2N2^n}=2^{n/2+o(n)}\sqrt{N}$
    queries to the BDD oracle.
    The total running time
    (including preprocessing) is therefore
    \[
        T_c+2^{n/2+o(n)}\sqrt{N}\cdot 2^{o(n)}T_o
        =T_c+\frac{2^{o(n)}T_o}{\alpha^{n/2}}
    \]

    Lastly, we show that the quantum algorithm will output a shortest non-zero vector with constant probability.
    Since $\norm{\boldsymbol{Be_1}}+1> \lambda_1(\cL)$, with at least $1/2$ probability one will find the index
    $(i,j,\vect{s})$ such that $f_{ij}(\vect{s})$ is a shortest in the list by using the quantum minimum finding algorithm. Therefore it suffices to show that there is an index $(i,j,\vect{s})$
    such that $\|f_{ij}(\vect{s})\|=\lambda_1(\cL)$.
    By the analysis above,
    the list $\{f_{ij}(\vect{s})|  \vect{s}\in \mathbb{Z}^n_3\}$ contains the shortest nonzero vector with at least constant
    probability. Hence, the algorithm will find the
    shortest vector of the lattice with constant probability.

    When using a QRAM, the proof is exactly the same except that
    we put the samples in a QRAM instead of hardcoding them.
    This QRAM will have size $M_o$ and
    everything else is the same.
\end{proof}

\begin{corollary}\label{coro_spherical_cap_SVP}
    There is a randomized algorithm that solves $\svp$ in time $2^{1.669n+o(n)}$ and in space $2^{0.5n+o(n)}$ with constant probability.
    There is a quantum algorithm that solves $\SVP$ in time $2^{0.9497n+o(n)}$
    and classical-space $2^{0.5n+o(n)}$ with a polynomial number of qubits.
    There is a quantum algorithm with QRAM that solves the $\svp$ in time $2^{0.8345n+o(n)}$, using a QRAM of size $2^{0.293n+o(n)}$, $\poly(n)$ qubits
    and classical space $2^{0.5n}$.
\end{corollary}

\begin{proof}
    Apply Theorem~\ref{theorem_alphaBDD} with\footnote{The optimal value of $\alpha$ was found numerically,
    see Section~\ref{section_relation_kissing}.}$\alpha=0.3853$: since $0\leqslant b=\log_2\beta(\cL)\leqslant 0.402$, we indeed have that
    $\alpha<\tfrac{1}{2}\sqrt{\tfrac{1}{1+b}}$ so we can create a $(\alpha+o(1))$-$\BDD$ oracle in time $T_c=2^{(A+1)n/2+o(n)}$,
    space $S_c=2^{0.5n+o(n)}$ such that each call takes time $T_o=2^{An/2+o(n)}$ where
    $A=A(b)$ is given by Theorem~\ref{theorem_alphaBDD}. The theorem also guarantees that $A(b)$ increases
    with $b$ so $A\leqslant A(0.402)$. But $0.402\geqslant\tfrac{1-4\alpha^2}{2\ln2}\approx0.2929$
    so $A(0.402)=\tfrac{4\alpha^2}{1-4\alpha^2}\cdot 0.402\approx 0.5862$ by Theorem~\ref{theorem_alphaBDD}.
    Apply Theorem~\ref{theorem_bdd_svp_spherical_cap}
    to get a randomized algorithm that solves $\svp$ in space $S_c$ and in time
    \[
        T
            :=T_c+\frac{2^{o(n)}T_o}{\alpha^n}
            =T=2^{(A+1)n/2+o(n)}+2^{An/2-n\log_2\alpha+o(n)}
            =2^{1.6690n+o(n)}.
    \]
    For the quantum case, we use $\alpha=0.3473$
    and the same reasoning gives $A\approx 0.3738$
    and the total complexity
    \[
        T'
            :=T_c+\frac{2^{o(n)}T_o}{\alpha^n}
            =T=2^{(A+1)n/2+o(n)}+2^{An/2-n/2\log_2\alpha+o(n)}
            =2^{0.9497n+o(n)}.
    \]
    Finally, for the quantum algorithm with a QRAM, the running
    time of the quantum oracle is $T_o'=\sqrt{T_o}$ and
    it uses a QRAM of size $T_q=2^{An/2+o(n)}$.
    We use $\alpha=0.3853$ (again) to obtain $A\approx 0.5862$
    and a total time complexity of
    \[
        T''
            :=T_c+\frac{2^{o(n)}T_o'}{\alpha^n}
            =T=2^{(A+1)n/2+o(n)}+2^{An/4-n/2\log_2\alpha+o(n)}
            =2^{0.8345n+o(n)}
    \]
    with a QRAM of size $T_q=2^{0.293n+o(n)}$.
\end{proof}

\section{Dependency of the SVP on a quantity related to the kissing number}\label{section_relation_kissing}

In the previous sections, we obtained several algorithms for SVP and bounded their complexity
using the only known bound on the quantity $\beta(\cL)$, which is related to the lattice kissing number
(see Section~\ref{sym:beta_number}):
$\beta(\cL)\leqslant 2^{0.402}$.
The complexity of those algorithms is highly affected by this quantity and
since $\beta(\cL)$ can be anywhere between $1$ and $2^{0.402}$ (see Section~\ref{section_preliminaries}),
we will study the dependence of the time complexity in $\beta(\cL)$. Recall that
$b=\log_2\beta(\cL)$.

In order to avoid doing the analysis three times, we introduce a factor $\nu$ that is $1$ for classical
algorithms and $\tfrac{1}{2}$ for quantum algorithms. We also introduce a factor $\xi$
that is $1$ for classical and plain quantum algorithms, and $\tfrac{1}{2}$ for quantum algorithms
with QRAMs.
We now can reformulate 
the time complexity in Theorem~\ref{theorem_bdd_svp_spherical_cap}
as
\begin{equation}\label{eq_time_compl_param_BDD}
    T_c+\frac{2^{o(n)}T_o}{\alpha^{\nu n}}
\end{equation}
We instantiate the algorithm in Theorem~\ref{theorem_bdd_svp_spherical_cap} with
the $\alpha$-$\BDD$ oracle provided by Theorem~\ref{theorem_alphaBDD} which satisfies
\[
    T_c=2^{(A+1)n/2+o(n)},\qquad T_o=2^{A\xi n/2+o(n)}\quad\text{and}\quad S_c=2^{0.5n+o(n)}
\]
where
\[
    A=\begin{cases}
                -\tfrac{1}{2\ln 2}W(-\tfrac{4\alpha^2\beta(\cL)^2}{e})&\text{when }b<\tfrac{1-4\alpha^2}{2\ln2}\\
                \tfrac{4\alpha^2}{1-4\alpha^2}b&\text{when }b\geqslant\tfrac{1-4\alpha^2}{2\ln2}
            \end{cases}.
\]
The new expression
of the time complexity to
\begin{equation}\label{eq_time_compl_param_A}
    2^{(A+1)n/2+o(n)}+2^{(A\xi/2-\nu\log_2\alpha)n+o(n)}.
\end{equation}
Note that for the classical and plain quantum algorithms, the preprocessing is always
negligible compared to cost of the queries but this is not necessarily
the case when using a QRAM.
The optimal choice of $\alpha$ is not obvious: by increasing the decoding radius, we reduce
the number of queries but increase the cost of each queries.
While we could, in principle, obtain closed-form expressions for
the optimal value of $\alpha$, those are too complicated to be
really helpful. Instead, we express it as an optimization program. Formally,
we have $T=2^{c(b,\nu,\xi)n+o(n)}$ where
\[
    c(b,\nu,\xi)=\min_{\alpha\in[\tfrac{1}{3},\tfrac{1}{2})}\max\left(\tfrac{A+1}{2},\tfrac{A\xi}{2}-\nu\log_2\alpha)\right)
\]
where $A$ and $\cos\phi$ are given by the expressions above that depend on $\alpha$. We numerically
computed the graph of this function and plotted the result on Figure~\ref{fig:complexity_svp_all}
for the classical, plain quantum and quantum with QRAM algorithms respectively.
As mentioned earlier, it is reasonable to conjecture that $\gamma(\cL) = 2^{o(n)}$ for most lattices. We obtain the following result when $\gamma(\cL)=\beta(\cL)^n$ is subexponential in $n$:

\begin{theorem}
    There
    is an algorithm that, given a full rank lattice $\cL$ of rank $n$ with $\gamma(\cL) = 2^{o(n)}$, finds a shortest vector in $\cL$ and runs
    \begin{itemize}
        \item in classical time $2^{1.292n+o(n)}$ and space $2^{0.5n}$,
        \item in quantum time $2^{0.750n+o(n)}$, classical space $2^{0.5n}$
            and $\poly(n)$ qubits,
        \item in quantum time $2^{0.667n+o(n)}$, classical space $2^{0.5n}$,
            $\poly(n)$ qubits and using a QRAM of size $2^{0.167n+o(n)}$.
    \end{itemize}
\end{theorem}

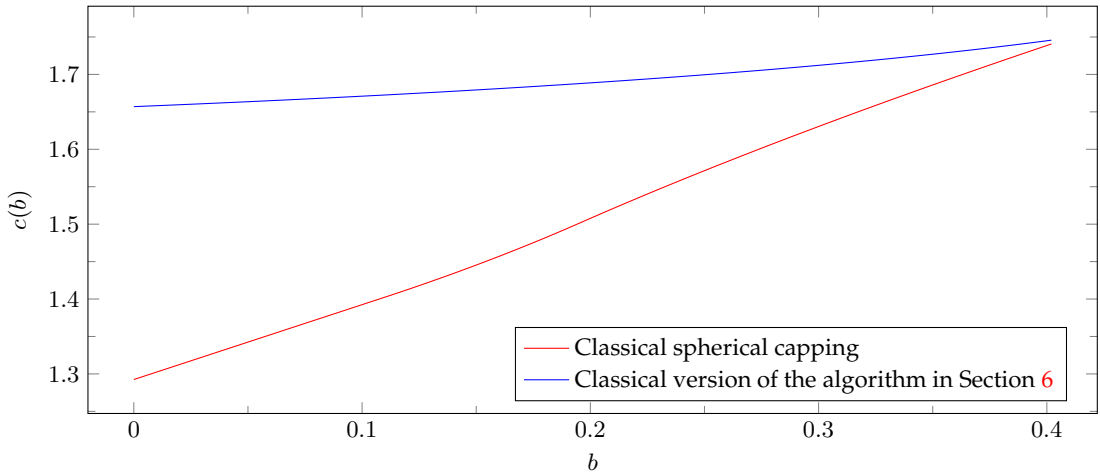
\begin{figure}
    \centering
    \begin{tikzpicture}[scale = 1]
        \begin{axis}[xlabel={$b$},ylabel={$c(b)$},
                width=15cm,height=7cm,
                minor y tick num=1,
                xtick distance=0.1,
                ytick distance=0.1,
                minor x tick num=1,
                enlarge x limits=0.05,
                legend columns=-1,
                legend to name=named legend,
                ]
            \addplot[red] table{svp_complexity_spherical_caps_classical.txt};
            \addlegendentry[black]{Classical};
            \addplot[blue] table{svp_complexity_spherical_caps_quantum.txt};
            \addlegendentry[black]{Quantum};
            \addplot[green!70!black] table{svp_complexity_spherical_caps_quantum_improved.txt};
            \addlegendentry[black]{Quantum with QRAM};
        \end{axis}
    \end{tikzpicture}
    \ref{named legend}
    \caption{(Exponent $c(b)$ of the) time complexity of the
        spherical capping algorithm, plotted against $b=\log_2\beta(\cL)$.
        The complexity of the algorithms is $2^{c(b)n+o(n)}$.}
        \label{fig:complexity_svp_all}
\end{figure}

\section{Direct application to the $\ZLIP$ problem}\label{section-:ZLIP}
The Lattice Isomorphism Problem (LIP) consist of recovering an orthogonal linear transformation sending one lattice to another,  assuming its existance. The case of the trivial lattice $\mZ^n$ is of particular interest ($\ZLIP$) for cryptographic purposes \cite{DvW22,BGPS23}. In \cite{Ducas23}, the author proposed a reduction from $\ZLIP$ of dimension $n$ to $\SVP$ of dimension $n/2 + 1$.

\begin{theorem}[\cite{Ducas23}]
$\ZLIP$ can be solved by making polynomially many calls to a Shortest Vector Problem (\SVP) oracle in dimension at most $n/2 + 1$.
\end{theorem}

Using Corollary~\ref{coro_spherical_cap_SVP}, we directly obtain the following theorem.

\begin{theorem}
 There exists a provable quantum algorithm that solves $\ZLIP$ in time $2^{0.417n+o(n)}$. The algorithm requires a QRAM of size $2^{0.147n+o(n)}$, $\poly(n)$ qubits and $2^{0.25n}$ classical space. 
\end{theorem}

\section*{Acknowledgments} We would like to thank Pierre-Alain Fouque, Paul Kirchner, Amaury Pouly and Noah Stephens-Davidowitz for useful comments and suggestions.

\bibliographystyle{alpha}
\bibliography{references}

\appendix
\section{Reduction from $\cvp$ to $\DGS$}
\label{sec:new_bdd_red}
The goal of this section is to improve Theorems~\ref{theorem_BDDtoDGS} and
\ref{th:theorem_BDDtoDGS_quantum} to not require
a sampler for the discrete Gaussian distribution
\emph{exactly at the smoothing parameter} $\eta_\varepsilon(\cL)$.
Indeed, we usually do not
know $\eta_\varepsilon(\cL)$ and it is a nontrivial problem to even estimate it \cite{CDLP13}.
It was stated in \cite[Theorem~7.3]{ADRSD15} that the reduction still holds if we only provide a
DGS oracle above the smoothing parameter.
For completeness, we provide a self-contained proof in the form of Theorem~\ref{theorem_BDDtoDGS_new}.
We first prove some technical lemmas on the discrete Gaussian distribution.

\begin{lemma}\label{lem:tailbound_new}
    For any lattice $\cL\subset\real^n$, $s>0$ and $r\geqslant s\sqrt{n}/\lambda_1(\cL)$,
    \[
        \rho_s(\cL\setminus B_n(r\lambda_1(\cL))
            \leqslant r^n\beta(\cL)^{n+o(n)}\rho_s(\cL\setminus\set{0})^{r^2}.
    \]
\end{lemma}
\begin{proof}
    Let $t=1+1/n$, $R=r\lambda_1(\cL)$, $r_i=Rt^i$ and $T_i=B_n(r_{i+1})\setminus B_n(r_i)$ for all $i\in\nat$.
    Then by definition of $\beta(\cL)$,
    \[
        |\cL\cap T_i|\leqslant |\cL\cap B_n(r_{i+1})|\leqslant \beta(\cL)^{n+o(n)}(rt^{i+1})^n
    \]
    It follows that
    \[
        \rho_s(\cL\setminus B_n(R))
            =\sum_{i=0}^\infty\rho_s(\cL\cap T_i)\\
            \leqslant\sum_{i=0}^\infty|\cL\cap T_i|e^{-\pi\tfrac{r_{i}^2}{s^2}}\\
            \leqslant \beta(\cL)^{n+o(n)}\sum_{i=0}^\infty\underbrace{(rt^{i+1})^ne^{-\pi\tfrac{R^2}{s^2}t^{2i}}}_{=f(i)}
    \]
    where $f(i)=(rt^{i+1})^ne^{-\pi\tfrac{R^2}{s^2}t^{2i}}$. But check that for all $i\in\nat$,
    \[
        \frac{f(i+1)}{f(i)}
            =t^ne^{-\pi\tfrac{R^2}{s^2}t^{2i}(t^2-1)}
            \leqslant e^{1-2\pi\tfrac{R^2}{ns^2}}
            \leqslant e^{1-3\pi}<1/2
    \]
    where we have used that $R\geqslant s\sqrt{n}$, $t^2-1\leqslant\tfrac{3}{n}$, $t^n\leqslant e$ and $t^{2i}\geqslant 1$. It follows that
    \[
        \rho_s(\cL\setminus B_n(R))
            \leqslant \beta(\cL)^{n+o(n)}\cdot 2\cdot f(0)
            \leqslant \beta(\cL)^{n+o(n)}\cdot 2(rt)^n e^{-\pi\tfrac{R^2}{s^2}}
            \leqslant r^n\beta(\cL)^{n+o(n)}\cdot e^{-\pi\tfrac{r^2\lambda_1(\cL)}{s^2}}.
    \]
    On the other hand,
    \[
        \rho_s(\cL\setminus\set{0})\geqslant e^{-\pi\tfrac{\lambda_1(\cL)^2}{s^2}}
    \]
    so the result follows immediately.
\end{proof}

\begin{lemma}\label{lem:lower_bound_eta_eps_ksquared}
    For any $c>1$,
    for any lattice $\cL\subset\real^n$, $\varepsilon\in(0,1/e)$, we have
    $t\eta_\varepsilon(\cL)\leqslant\eta_{\varepsilon^{t^2}e^{o(1)}}(\cL)$
    where $t=1+\tfrac{1}{n^c}$.
\end{lemma}
\begin{proof}
    Let $s=\eta_\varepsilon(\cL)$, $M=n^c$ and $t=1+\tfrac{1}{M}$ with $c>1$. Let
    $r=M$ and check that
    \begin{equation}\label{eq:lower_bound_r_is_big_enough}
        \tfrac{\sqrt{n}}{s\lambda_1(\cL^*)}
            \leqslant\sqrt{\tfrac{n\pi}{\ln(1/\varepsilon)}}\leqslant r
    \end{equation}
    by Lemma~\ref{lemma_smoothinglambda}, for large enough $n$ since $\varepsilon\leqslant1/e$.
    Let $\Gamma=(\cL^*\setminus\set{0})\cap B_n(r\lambda_1(\cL^*))$, then
    \[
        |\Gamma|\leqslant\beta(\cL^*)^{n+o(n)}r^n.
    \]
    Now observe that
    \begin{align*}
        \rho_{1/ts}(\cL^*\setminus\set{0})
            &=\sum_{x\in\cL^*\setminus\set{0}}\rho_{1/ts}(x)
            =\sum_{x\in\cL^*\setminus\set{0}}\rho_{1/s}(x)^{t^2}\\
            &\geqslant\sum_{x\in\Gamma}\rho_{1/s}(x)^{t^2}
            \geqslant|\Gamma|\left(\tfrac{1}{|\Gamma|}\sum_{x\in\Gamma}\rho_{1/s}(x)\right)^{t^2}
                &&\text{by Jensen's inequality}\\
            &=|\Gamma|^{1-t^2}\rho_{1/s}(\Gamma)^{t^2}\\
            &=|\Gamma|^{1-t^2}\left(\rho_{1/s}(\cL^*\setminus\set{0})-\rho_{1/s}(\cL^*\setminus B_n(r\lambda_1(\cL^*))\right)^{t^2}\\
            &\geqslant \left(\beta(\cL^*)^{n+o(n)}r^n\right)^{1-t^2}
                \left(\varepsilon-\beta(\cL^*)^{n+o(n)}r^n\varepsilon^{r^2}\right)^{t^2}
    \end{align*}
    by \eqref{eq:lower_bound_r_is_big_enough}, Lemma~\ref{lem:tailbound_new}
    and since $\rho_{1/s}(\cL^*\setminus\set{0})=\varepsilon$. Now observe that $t^2-1\leqslant 3/M$ and
    \[
        \left(\beta(\cL^*)^{n+o(n)}r^n\right)^{-3/M}
            =e^{-3n^{1-c}\left(\ln\beta(\cL^*)+o(1)+c\ln n\right)}
            =e^{-3n^{1-c}O(\ln n)}
            =e^{-o(1)}
    \]
    since $c>1$. We also have that
    \begin{align*}
        \beta(\cL^*)^{n+o(n)}r^n\varepsilon^{r^2}
            &=e^{n\ln\beta(\cL^*)+nc\ln(n)+o(n)-(r^2-1)\ln\tfrac{1}{\varepsilon}}\varepsilon\\
            &\leqslant e^{O(n\ln n)-n^{2c}}\varepsilon
                &&\text{since }\varepsilon\leqslant1/e\\
            &\leqslant e^{o(1)}\varepsilon
    \end{align*}
    for large enough $n$. It follows that
    \[
        \rho_{1/ts}(\cL^*\setminus\set{0})
            \geqslant e^{o(1)}(\varepsilon e^{o(1)})^{t^2}
            \geqslant\varepsilon^{t^2}e^{o(1)}
    \]
    since $t\leqslant2$. Therefore we must have $\eta_{\varepsilon^{t^2}e^{o(1)}}(\cL)\geqslant ts$.
    
\end{proof}

We are in a position to prove the theorem, which we restate below.

\theoremBDDtoDGSnew*

\begin{proof}
    First we note that we can easily identify an interval $I=[a,b]$ such that $\eta_\varepsilon(\cL^*)\in[a,b]$
    and $\tfrac{b}{a}\leqslant 2^{n+o(n)}$. Indeed, by e.g. \cite[Lemma~2.11 and Claim~2.13]{Regev09}
    one has
    \[
        \sqrt{\ln(1/\varepsilon)}{\pi}
            \leqslant
            \lambda_1(\cL)\eta_\varepsilon(\cL^*)
            \leqslant
            \sqrt{n}
    \]
    so $\eta_\varepsilon(\cL^*)\in \tfrac{1}{\lambda_1(\cL)}[c,d]$
    where $\tfrac{d}{c}=\sqrt{\tfrac{n}{\ln(1/\varepsilon)}}=O(\sqrt{n})$ since $\varepsilon\leqslant1/200$.
    Furthermore, by running the LLL algorithm on $\cL$ and taking the length of the shortest basis
    vector, we obtain a length $\ell$ such that $2^{-n}\ell\leqslant\lambda_1(\cL)\leqslant\ell$.
    It follows that $\eta_\varepsilon(\cL^*)\in[a,b]:=[\tfrac{c}{\ell},\tfrac{2^nd}{\ell}]$ and
    $\tfrac{b}{a}=2^n\tfrac{d}{c}=2^{n+o(n)}$.

    Now let $c=2+\alpha$, $\delta=1+\tfrac{1}{n^c}$ and $N=\left\lceil\tfrac{\ln(b/a)}{\ln\delta}\right\rceil$.
    Check that $N=\poly(n)$ since $\tfrac{b}{a}=2^{n+o(n)}$. Furthermore, if we let $s_i=a\delta^i$
    for $i=1,\ldots,N$ then there must exists some $i_0$ such that
    \begin{equation}\label{eq:interval_eta_eps}
        \tfrac{1}{\delta}s_{i_0}\leqslant\eta_\varepsilon(\cL^*)\leqslant s_{i_0}.
    \end{equation}
    The preprocessing stage of the algorithm consists in calling the $\operatorname{hDGS}^m_{\sigma}$
    sampler, where $\sigma(\cL)=\eta_\varepsilon(\cL^*)$, with parameter $s_i$ to obtain a lists $L_i$ of $m$ vectors, for each $i=1,\ldots,N$,
    and storing all the lists. This requires $N=\poly(n)$ calls and we need to store $m\cdot\poly(n)$ vectors. Note that the hDGS sampler is allowed to
    return less than $m$ samples if $s_i>\sigma(\cL)$: in this case,
    we do not care about the distribution of the vectors anyway so we can
    add random vectors until we get $m$ samples when that happens.

    We now describe the oracle for $\CVP^\phi$. On input $\vect{t}\in \mR^n$, for each $i=1,\ldots,N$, the oracle calls the algorithm of Theorem~\ref{theorem_BDDtoDGS} on
    input $\vect{t}$ and provides the list $L_i$ to the algorithm in place of the DGS samples.
    Hence, for each $i$, we either obtain a lattice vector $\vect{y}_i$ or the algorithm from Theorem~\ref{theorem_BDDtoDGS}
    fails and we let $\vect{y}_i=\vect{0}$. Finally, the oracle returns the point closest to $\vect{t}$ in the list $\vect{y}_1,\ldots,\vect{y}_N$.
    
    As noted in Remark~\ref{remark:space_complexity_theorem_BDDtoDGS}, if we assume
    that all the DGS samples have $\poly(n)$ bit-size then the reduction from
    Theorem~\ref{theorem_BDDtoDGS} has time complexity $m\cdot \poly(n)$
    and space complexity $O(\poly(n)+\ln m)$ \emph{excluding the storage space of the $m$ vectors provided by the DGS}.
    Furthermore, when the basis vectors of $\cL$ have bit-size $\poly(n)$ (which is
    the case in this paper by assumption) then we can ensure that all DGS samples have
    $\poly(n)$ bit-size by first generating more samples (say twice the amount)
    and throwing away all samples of norm larger than $\exp(\Omega(n^2))$. Since the vectors are sampled from a Gaussian with width
    at most $\exp(O(n))$ (since the basis vectors have size at most $2^{o(n)}$,
    the error induced by throwing away the tail of the distribution is smaller than $2^{-\Omega(n^2)}$. In the quantum setting, we use Theorem~\ref{theorem_BDDtoDGS_quantum} instead of Theorem~\ref{theorem_BDDtoDGS} which gives exactly the same result except for the
    exponent in the complexity of the oracle.
    Therefore running time is clear so it remains to prove that the algorithm actually solves $\CVP^\phi$
    on $\cL$.

    We first note that when called on $s_{i_0}$, the $\operatorname{hDGS}^m_\sigma$ sampler will return
    $m$ vectors that are $0.5-$close to $m$ samples from $D^{m}_{\cL,s}$ since $s_{i_0}\geqslant\sigma(\cL)=\eta_\varepsilon(\cL^*)$
    by \eqref{eq:interval_eta_eps}. Furthermore, by Lemma~\ref{lem:lower_bound_eta_eps_ksquared} we have
    \[
        t\eta_\varepsilon(\cL^*)
            \leqslant\eta_{\varepsilon^{t^2}e^{o(1)}}(\cL^*)
    \]
    where $t=1+\tfrac{1}{n^c}=\delta$. It follows by \eqref{eq:interval_eta_eps} that
    \[
        \eta_\varepsilon(\cL)
            \leqslant s_{i_0}
            \leqslant \delta\eta_\varepsilon(\cL^*)
            \leqslant \eta_{\epsilon^{\delta^2}e^{o(1)}}(\cL^*).
    \]
    But the map $\varepsilon\mapsto\eta_{\epsilon}(\cL)$ is continuous and decreasing, so it follows that
    \[
        s_{i_0}=\eta_{\varepsilon'}(\cL^*)
        \quad\text{for some }
        \varepsilon^{\delta^2}e^{o(1)}\leqslant\varepsilon'\leqslant\varepsilon.
    \]
    Therefore by Theorem~\ref{theorem_BDDtoDGS} and Remark~\ref{remark:space_complexity_theorem_BDDtoDGS}, with constant probability
    over the choice of $L_{i_0}$, the (deterministic) algorithm of Theorem~\ref{theorem_BDDtoDGS} solves
    $\CVP^{\psi}$ when given $L_{i_0}$, where
    \[
        \psi(\cL)=\frac{\sqrt{\ln(1/\varepsilon')/\pi-o(1)}}{ 2\eta_{\varepsilon'}(\cL^*)}
    \]
    assuming that $m=|L_{i_0}|\geqslant m':=O(\frac{n \ln(1/\varepsilon')}{\sqrt{\varepsilon'}})$
    which holds because
    \[
        \frac{n \ln(1/\varepsilon')}{\sqrt{\varepsilon'}}
            \leqslant \frac{n\delta^2\ln(1/\varepsilon)+o(n)}{\sqrt{\varepsilon^{\delta^2}e^{o(1)}}}
            \leqslant \frac{n\ln(1/\varepsilon)+n^{1-c}\ln(1/\varepsilon)+o(n)}{\sqrt{\varepsilon}\varepsilon^{n^{-c}/2} e^{o(1)}}
            =O\left(\frac{n \ln(1/\varepsilon)}{\sqrt{\varepsilon}}\right)
    \]
    since $n^{1-c}\ln\varepsilon=n^{1-c+\alpha}=o(1)$ and $\varepsilon^{n^{-c}/2}=e^{n^{\alpha-c}/2}=e^{o(1)}$.
    It follows that, with constant probability over the preprocessing, our oracle solves $\CVP^\psi$.
    Check that $\varepsilon^{\delta^2}e^{o(1)}\geqslant \varepsilon^{\delta^2+o(1)}$ since $\varepsilon<\tfrac{1}{e}$
    and thus by Lemma~\ref{lemma_smoothing},
    \[
        \eta_{\varepsilon^{\delta^2}e^{o(1)}}(\cL^*)
            \leqslant \eta_{\varepsilon^{\delta^2+o(1)}}(\cL^*)
            \leqslant (\delta^2+o(1))\eta_{\varepsilon}(\cL^*)
            \leqslant (1+o(1))\eta_{\varepsilon}(\cL^*)
    \]
    since $\delta=1+o(1)$. Hence we have
    \[
        \psi(\cL)
            \geqslant \frac{\sqrt{\ln(1/\epsilon)/\pi-o(1)}}{ 2(1+o(1))\eta_{\epsilon}(\cL^*)}
            \equiv \phi(\cL).
    \]
\end{proof}

\end{document}